%% file: paper.tex
\documentclass[11pt]{article}

\input{packages.tex}
\input{environments.tex}

\input{commands.tex}

\title{
On the Noise Sensitivity of the Randomized SVD
}
\author[1]{Elad Romanov \thanks{\texttt{eromanov@stanford.edu}}}

\affil[1]{Department of Statistics, Stanford University}

\date{}

\SetKwInput{KwInput}{Input}                % Set the Input
\SetKwInput{KwOutput}{Output}              % set the Output

\begin{document}

    \maketitle

	\begin{abstract}

The randomized singular value decomposition (R-SVD) is a popular sketching-based algorithm for efficiently 
computing the partial SVD of a large matrix. When the matrix is low-rank,
the R-SVD produces its partial SVD exactly; but when the rank is large, it only yields an approximation.

Motivated by applications in data science and principal component analysis (PCA), we analyze the R-SVD under 
 a low-rank signal plus noise measurement model;
 specifically, when its input is a spiked random matrix.
The singular values produced by the R-SVD are shown to exhibit a BBP-like phase transition: when the SNR exceeds a certain detectability threshold, 
that depends on the dimension reduction factor,
the largest singular value is an outlier;
below the threshold, no outlier emerges from the bulk of singular values. We further compute
 asymptotic formulas for
 the overlap between the ground truth signal singular vectors and the approximations produced by the R-SVD. 

 Dimensionality reduction has the adverse affect of amplifying the noise in a highly nonlinear manner.
Our results demonstrate the statistical advantage---in both signal detection and estimation---of the R-SVD over more naive sketched PCA variants; the advantage is especially dramatic when the sketching dimension is small. Our analysis is asymptotically exact, and substantially more fine-grained than 
 existing operator-norm error bounds for the R-SVD, which largely fail to give meaningful error estimates in the moderate SNR regime.
It applies for a broad family of sketching matrices previously considered in the literature, including Gaussian i.i.d. sketches, random projections,  and the sub-sampled Hadamard transform, among others. 

Lastly, we derive an optimal singular value shrinker for singular values and vectors obtained through the R-SVD, which may be useful for applications in matrix denoising.

	\end{abstract}

    \section{Introduction}

	The singular value decomposition (SVD) is a fundamental tool in numerical linear algebra that is widely used in a variety of applications across engineering, machine learning and statistics \cite{jolliffe2002principal,golub2013matrix,hastie2009elements,Anderson1959AnIT}.
	Standard algorithms for computing the full SVD of an $n$-by-$m$ matrix perform roughly 
	$nm \min(n,m)$ floating point operations \cite{golub2013matrix}.
	 Owing to this cubic dimensional dependence, exactly computing the SVD of even moderately-sized matrices can be a prohibitively expensive. 
    
    When the matrix of interest is low-rank---either exactly or approximately---it is often sufficient in practice to compute 
    % only 
    a {partial} (or truncated) SVD. 
    Matrices of this kind are especially commonplace in data science applications \cite{udell2019big}.
    Motivated in part by the challenges of today's ``big-data'' age, there has been a considerable effort to devise fast and numerically stable algorithms for the partial SVD. 
    This paper is concerned with one particular, popular, state-of-the-art fast SVD method: the  
 randomized SVD (R-SVD) algorithm of \cite{liberty2007randomized,woolfe2008fast,rokhlin2010randomized,halko2011algorithm,halko2011finding}.
As testament for its practical relevance, the popular machine learning library \texttt{scikit-learn} uses this algorithm (with constant number of power iterations, by default $5$) for its implementation of the truncated SVD \cite{pedregosa2011scikit}. Furthermore, large-scale efforts to create ``standarized'' code packages for randomized numerical linear algebra, including the R-SVD, are currently ongoing \cite{murray2023randomized}.

In its most basic form, the R-SVD proceeds along the following lines. Given an $n$-by-$m$ data matrix $\bY \in \RR^{n\times m}$, the R-SVD first reduces its dimension, multiplying it from the right $\tilde{\bY}=\bY\bOmega^\T \in \RR^{n\times d}$, where $\bOmega\in \RR^{d\times m}$ is a random sketching matrix ($d$ being the sketching dimension); 
then, from the sketched matrix $\tilde{\bY}$, computes an approximated projection matrix onto the span of its large left singular vectors (a so-called ``Range Finder'' \cite{halko2011finding,murray2023randomized}), for example by means of the QR decomposition: $\tilde{\bY}=\bQ\bR$ \footnote{Recall: $\bQ$ has $\rank(\tilde{\bY})\le d$ many columns, which constitute an orthonormal basis for the column space.}; and finally, projects the original matrix onto this subspace, $\hbY=\bQ\bQ^\T\bY$---obtaining a reduced matrix whose full SVD approximates the partial SVD of $\bY$.\footnote{Importantly, to compute the SVD of $\hat{\bY}$ one can first compute the SVD of $\bQ^\T\bY$, which has smaller dimension than $\bY$ when $d<n$, and then multiply the left singular vectors by $\bQ$.}
When  $\bY$ is low-rank, specifically its rank is smaller than the sketching dimension $d$, the R-SVD yields its exact partial SVD (with high probability).
Prior theoretical works on the R-SVD have focused on showing that even when this is not the case---that is, $\bY$ is high-rank but nonetheless exhibits ``fast'' spectral decay---the reduced matrix $\hbY$ is still a good approximation for $\bY$:
essentially on par, with respect to operator norm, with its true truncated SVD.
% with respect to operator norm; by ``good'' here we mean on par with the exact partial SVD of $\bY$. 
% (Additional details shall be given momentarily.)

The ``approximation-theoretic'' perspective mentioned above, wherein the goal of the R-SVD is to produce a {\it low-rank approximation} to the matrix (or operator) $\bY$, is a natural one in numerical analysis.
However, in the context of statistics and data analysis---for 
example principal component analysis (PCA) and other spectral methods---it might not be entirely aligned with what one might truly be interested in.
A perhaps more relevant question is:
{\it to what degree does the R-SVD preserve the large principal components of $\bY$?} 
In certain regimes, 
this question turns out to be substantially more delicate than just low-rank approximation. Certainly, when the leading singular values of $\bY$ are much larger than the sub-leading ones (that is: $\bY$ exhibits ``very fast'' spectral decay), one could obtain meaningful error bounds for the principal components by means of ``general-purpose'' singular vector perturbation inequalities (e.g. Davis-Kahan \cite{yu2015useful}). 
However, in the challenging regime where the singular values of $\bY$ are all of the same scale, these perturbation bounds become largely uninformative. 
Accordingly, this setting calls for a more fine-grained analysis of the R-SVD.

% While completely general fine-grained analysis of the R-SVD appears to be challenging. 
This paper takes a first step towards this goal, focusing on matrices $\bY$ in the form of a low-rank signal plus noise $\bY = \bX + \bZ$, where the entries of $\bZ$ are i.i.d. centered Gaussians. Specificially, we work under a variant of Johnstone's spiked model \cite{johnstone2001distribution}, an asymptotic framework wherein the rank of the signal $\bX$ is constant while the dimensions $n,m\to\infty$. In this model, the signal-to-noise ratio (SNR) is normalized such that the singular values of $\bY$---whether corresponding to signal or to pure noise---are all of the same scale. As mentioned, this is a regime where existing operator norm bounds (e.g. \cite{halko2011algorithm}) coupled with standard singular vector perturbation inequalities yield rather uninformative error bounds. Our approach is entirely different, and based on asymptotically {exact} (in the large-dimensional limit) computations, using tools from random matrix theory. 
Our results quantify in a very precise sense the loss of SNR caused by dimensionality reduction. Focusing on small undersampling ratio $\beta\equiv d/m\ll 1$, our results unveil, for example, the following behavior:
\begin{itemize}
    \item Signal principal components (PCs) whose singular values satisfy $\sigma_i \lesssim \beta^{-1/8}$ are non-detectable: the corresponding singular values of $\hbY$ (the reduced matrix) are indistinguishable from noise. Moreover, the corresponding singular vectors produced by the R-SVD are entirely de-correlated from their signal counterparts.
    \item Stronger signal PCs $\sigma_i\gtrsim \beta^{-1/8}$  produce outliers in the spectrum of $\bY$, and are consistently detectable. The corresponding singular vectors are aligned with the signal PCs; in fact, the angle between the PCs concentrates around a deterministic quantity, that we compute exactly.  
\end{itemize}
As will be made clear later on,
this kind of fine-grained information is entirely indiscernible from the operator norm approximation bounds previously given in the literature.

\subsection{The merits of a signal plus noise analysis}

	This paper analyzes the R-SVD under a signal plus noise framework.  
	We justify why we believe a study of this kind is worthwhile---and indeed, natural---from several angles.
	
	{\it The SVD is a fundamental tool in data analysis.} Perhaps the most well-known example here is 
%	This study is motivated, in part, by the ubiquity of the SVD as a data analysis tool, and by
	 principal component analysis (PCA) 
%	 in particular
 \cite{jolliffe2002principal,hastie2009elements,Anderson1959AnIT}.
	Suppose one has a data set consisting of \Revision{$n$ points $\bm{y}_1,\ldots,\bm{y}_n$ in $m$ dimensions.} PCA is, at its core, a technique for discovering a latent low-dimensional linear structure in the data (the end goal could vary: dimension reduction, exploratory analysis, interpretability, or something else). One forms the data matrix $\bY\in \RR^{n\times m}$ whose \Revision{rows are $\bm{y}_1,\ldots,\bm{y}_n$.} (Assume, for simplicity, that the data set is centered.) The largest \Revision{right} singular vectors of $\bY$, equivalently eigenvectors of the sample covariance matrix \Revision{$\bY^\T\bY/n$}, are the directions along which the variation among the data points is largest. Accordingly, if the spectrum of $\bY$ contains few singular values that are distinctively larger than the others, then their corresponding directions are considered ``important''. 
	A probabilistic framework which is natural in the context of PCA is that of a low-rank factor model {\cite{tipping1999probabilistic}}. The data points are modeled as a sum $\bm{y}_i = \bm{x}_i + \bm{z}_i$, where $\bm{x}_i$ is a low-dimensional latent ``signal''---the part of $\bm{y}_i$ which is considered ``informative'', and which lies in some unknown latent low-dimensional subspace, shared across all $\bm{x}_i$'s--- and $\bm{z}_i$ is isotropic noise. The resulting data matrix $\bY$ has the form of a signal plus noise matrix $\bY=\bX+\bZ$, with $\bX$ being low-rank.
	There has been a great deal of literature studying the spectral behavior of the data matrix $\bY$, and in particular how close are the observed PCs (the leading left singular values of $\bY$) to the signal PCs, which span the latent low-dimensional subspace, see e.g. \cite{johnstone2006high,bai2010spectral,vershynin2018high,wainwright2019high}. In particular, it is well-known that in high dimensions, namely when $n,m$ are comparable and large, the PCs of $\bY$ are inconsistent estimates of their population counterparts. 
	
	 Contemporary data sets are often very large, and so from a computational standpoint, it would be helpful to replace the full SVD operation by the fast R-SVD when performing PCA on the data matrix. In the recent scientific literature, we have already seen several such papers where actual, real-world, \emph{massive} data is analyzed in this fashion, see for example \cite{hie2019efficient,linderman2019fast,laughney2020regenerative} (among many others). We foresee that this trend will grow, as the size of typical data sets encountered in applications increases (for example, in genomics and single-cell data). An important point is that the R-SVD only gives an {approximation} to the true truncated SVD, and so one wonders: from a {statistical} point of view, how much do we lose by using it over the exact, but computationally expensive, full SVD? The present paper aims to give a precise answer to this question under the {spiked model} \cite{johnstone2001distribution}, which is a popular and mathematically rich framework for thinking about PCA and related problems. 
	
	{\it The SVD as a tool for denoising.}
	A signal plus model noise model is natural in the context of low-rank matrix recovery, a problem that has been extensively studied in signal processing and machine learning, cf. \cite{wright2009robust,candes2010matrix,candes2012exact,donoho2014minimax,davenport2016overview,bun2017cleaning}. That is, we would like to estimate an unknown low-rank signal matrix $\bX$ from noisy measurements $\bY=\bX+\bZ$. One popular and simple approach to this problem is \emph{singular value shrinkage} {\cite{perry2009cross,shabalin2013reconstruction,nadakuditi2014optshrink,gavish2017optimal}}. This denoising method is based on the SVD: one takes the SVD of $\bY$, and systematically deflates its singular values to account for the effects of the noise---in particular, all but the leading empirical PCs of $\bY$ should typically be cut off. How does the optimal denoising rule change if instead of the full SVD, one performs singular value shrinkage on the R-SVD of $\bY$? How much additional error would the process of dimensionality reduction introduce into the denoising problem? In this paper we derive, in particular, the optimal shrinkage rule to use in this setting. 
	
	{\it Noise sensitivity and ``smoothed analysis''.}
	When the matrix $\bX$ is low-rank (specifically, its rank is smaller than the sketching dimension), the R-SVD yields an exact partial SVD of $\bX$. In real-world settings, however, data matrices are rarely exactly low-rank, for example due to measurement noise. How far do the singular values and vectors returned by the R-SVD (applied on the noisy measurement $\bY=\bX+\bZ$) deviate from the ground truth ($\bX$) when measurement noise enters the picture?
%	Another viewpoint through which a signal plus noise perspective is natural to consider is that of noise tolerance or `smoothed analysis'' {\cite{Spielman2001SmoothedAO}}. Specifically, the randomized SVD recovers perfectly the true truncated SVD when $\bX$ is low rank; how does its performance deteriorate when one adds to the picture measurement noise? 
	Our results give a {precise} answer to this question---one that is considerably more fine-grained than existing error bounds \Revision{for} the R-SVD \cite{halko2011finding}---in an idealized model where $\bX$ is very low-rank. Dimensionality reduction has the effect of {amplifying} the noise in a highly nonlinear fashion. We show, in particular, that at any fixed signal-to-noise ratio (SNR), one can only reduce the dimension up to a certain point---which we calculate exactly---before the signal in the reduced matrix becomes \emph{completely} swamped by the noise. That is, below that breakdown point, the singular values and vectors obtained through the R-SVD become {completely decorrelated} from the ground truth.	
	% Deferring the details for later, we remark that our findings largely corroborate the widely-held belief that the R-SVD is robust to measurement noise, even at rather extreme undersampling ratios. 

%	Lastly, several past works have considered 
%	Lastly, we mention in passing another related line of works, which studied the effects of sketching and dimensionality reduction in linear regression, cf. \TODO{\cite{?}}
	
	\subsection{The randomized SVD algorithm}
	\label{sec:Intro:RSVD}

	As mentioned before, this paper analyzes the popular randomized SVD (R-SVD) algorithm developed in \cite{liberty2007randomized,woolfe2008fast,rokhlin2010randomized,halko2011algorithm,halko2011finding}.

	Let $\bY$ be some given matrix, and $k$ be the reference rank. The R-SVD algorithm aims to find a low-rank approximation to $\bY$ (of rank $d$, larger than $k$) which is nearly on par with with its best rank-$k$ approximation. Recall that by the classical Eckart-Young-Mirsky theorem (cf. \cite{horn2012matrix}), the $k$-truncated SVD of $\bY$
	\begin{align}\label{eq:TSVD-form}
		\bY_k = \sum_{i=1}^k \sigma_i\bu_i\bv_i^\T\,, \qquad\textrm{where}\qquad \bY \overset{\textrm{SVD}}{=}\sum_{i=1}^{n\wedge m}\sigma_i \bu_i\bv_i^\T\,,
	\end{align}
	is the best rank-$k$ approximation of $\bY$ with respect to any orthogonally invariant, and in particular the operator, norm:
	\begin{align}\label{eq:TSVD-error}
		\min_{\rank(\hbY)\le k}\|\bY-\hbY\| = \|\bY-\bY_k\| = \sigma_{k+1} \,.
	\end{align}
	Let $d\ge k+1$ be the {sketching dimension}; \cite{halko2011finding} suggests, for example, $d=2k$.
	The randomized SVD algorithm constructs, in time $O(dnm)$, a rank $d$ matrix $\hbY$ such that the error $\|\bY-\hbY\|$ is comparable to \eqref{eq:TSVD-error}. The details of the algorithm, as described in \cite[Page 227]{halko2011finding} (``Prototype for Randomized SVD''), are briefly summarized below:
	
	\begin{itemize}
			\item {\bf Input:} $\bY$ the data matrix; $d \,(> k)$ the sketching dimension; $q$ number of power iterations.
			\item {\bf Step I:} ``Randomized Range Finder'':
			\begin{enumerate}
					\item Let $\bOmega\in \RR^{d\times m}$ be an i.i.d. Gaussian sketch matrix.
					\item Form the $n\times d$ matrix $\tilde{\bY}=(\bY\bY^\T)^q\bY\bOmega^\T$. 
					\item Construct $\bQ\in\RR^{m\times d}$ whose columns are an orthonormal basis of $\mathrm{range}(\tilde{\bY})$. To find such $\bQ$, one can use the QR decomposition: $\tilde{\bY}=\bQ\bm{R}$.
				\end{enumerate}  
		\item {\bf Step II:} SVD on a reduced matrix:
		\begin{enumerate}
				\item Form $\bB=\bQ^\T \bY \in \RR^{d\times m}$.
				\item Compute the SVD: $\bB=\tilde{\bU}\hat{\bSigma}\hat{\bV}^\T$.
				\item Compute $\hat{\bU}=\bQ \tilde{\bU}$.  
			\end{enumerate}
		Return: $\hbY=\hat{\bU}\hat{\bSigma}\hat{\bV}^\T$, an approximated partial SVD of $\bY$.
		\end{itemize}
%	The success of the randomized SVD owes to the following remarkable fact (which can be quantified \cite{halko2011finding}): when one reduces the dimension of $\bX$ from the right, using a random sketching matrix, the \emph{strong} left singular vectors of $\bX$ are approximately preserved.

     The present paper considers the R-SVD algorithm in its most basic form, with $q=0$ power iterations. The authors of \cite{halko2011finding} provide the following guarantee on the expected error, in operator norm, of the R-SVD:
	\begin{align}\label{eq:Halko-Bound}
		\Expt \| \bY - \hbY\| \le \left( 1 + \frac{4\sqrt{d}}{d-k-1} \sqrt{n\wedge m}\right)\sigma_{k+1}\,.
	\end{align}
	(\cite[Theorem 1.1]{halko2011finding}.)\footnote{Note that $k$ does not appear explicitly in the description of the algorithm. That is, \eqref{eq:Halko-Bound} holds for \emph{every} $k<d$. Also note that the expression in the parentheses on the right-hand side increases as $k$ increases, while $\sigma_k$ decreases; in particular, the bound \eqref{eq:Halko-Bound} is typically non-monotonic in $k$.  } 
    We remark that an improved bound, which is actually tight in a { worst-case} sense (but is also considerably more cumbersome to state), was proven in \cite{witten2015randomized}. While that bound improves on the constants in \eqref{eq:Halko-Bound}, its qualitative dependence on $d,k,m,n$ is essentially the same.
 
    While power iterations are known to \emph{dramatically} improve the performance of the R-SVD algorithm (at the expense of additional computational overhead), their treatment is beyond the analysis presented in this paper. (Even further improved variants of the basic R-SVD exist, for example 
    the block Krylov method of \cite{musco2015randomized}.) 
    From a theoretical perspective, we primarily aim to develop a finer-grained understanding of the importance of the ``Range Finder'' step, over more naive sketched SVD variants.

    \Revision{
    Consider, for example, the following simpler procedure---sometimes refered to as sketched PCA (e.g. \cite{yang2021reduce})---where instead of the two-step procedure described above, one  simply takes the right singular vectors of the randomly-projected data matrix $\tilde{\bY}={\bOmega'}\bY$ (one can similarly incorporate power iterations via ${\bOmega'}\bY(\bY^\T\bY)^q$); the respective singular values and left singular vectors can be approximated via $ {\bY}\hat{\bV}$. 
    How worse does this procedure perform over the more complicated R-SVD? 
    Sketched PCA (with $q=0$) was recently analyzed by \cite{yang2021reduce} under a setting similar to the  present paper (the spiked model). Comparing our results to theirs, we can quantify in a precise sense the statistical advantage of the R-SVD in the context of signal detection and estimation; see Section~\ref{sec:SketchedPCA}.
    }

 % \TODO{TKTK}
 %  Specifically, we only consider its most basic form. Over the years, several improvements over the basic algorithm
	% were introduced;
	% examples include the addition of power iterations \cite{halko2011algorithm}, and the block Krylov method of \cite{musco2015randomized}. While these are known to improve performance---dramatically even---their analysis is beyond the scope of the present paper.
	% We emphasize that the R-SVD algorithm (and by corollary, the analysis undertaken in this paper) is of certain practical interest. For example, the popular machine learning library, \texttt{scikit-learn}, uses this algorithm (with a constant number of power iterations---by default, $5$) for its implementation of the truncated SVD \cite{pedregosa2011scikit}.
	
\subsection{The spiked model}

Motivated by applications in data analysis and PCA, 
we aim to develop a {precise} quantitative picture of the randomized SVD, applied to a signal plus noise matrix $\bY=\bX+\bZ$, where $\bX$ is low-rank. We will work under a variant of the spiked mode, {\cite{johnstone2001distribution}}. Our results are \emph{asymptotic} and pertain to a regime where: 1) The matrix dimensions $m,n$ are both large and comparable, formally, $m,n\to\infty$ at a fixed aspect ratio $\gamma\equiv m/n$; 2) The signal rank, $r\equiv \rank(\bX)$ is constant; 3) The signal-to-noise ratio (SNR) of the problem is moderate, in that the singular values of $\bX$ and the noise $\bZ$ are of the same scale.
In this paper, we consider exclusively an i.i.d. noise matrix $\bZ$, specifically with Gaussian entries. 

The spiked model has found many applications within statistics, signal processing and machine learning (see, for example, \cite{couillet2011random,couillet2022random}), and its theoretical properties are at this point very well understood \cite{baik2005phase,baik2006eigenvalues,paul2007asymptotics,benaych2012singular,bloemendal2016principal}. This model is particularly appealing in the context of PCA owing to its very simple spectral theory, described by the following phenomena: 1) The singular values of $\bY$ are divided into a bulk, and up to $r$ outliers that exceed the bulk; 2) The shape of the bulk is determined by the spectrum of $\bZ$, and corresponds to a Marchenko-Pastur law; 3) The outliers are in a direct correspondence with the signal spikes; for $1\le i \le r$, the $i$-th largest singular value of $\bY$ will be an outlier if and only if the $i$-th population spike exceeds some {detection threshold}; 4) The angles between the population and observed principal components concentrate around deterministic quantities, which can be consistently estimated from the observed spectrum of $\bY$. Importantly, the empirical PCs are \emph{inconsistent} estimates (as $m,n\to\infty$) of the signal spikes. In Section~{\ref{sec:SpikeModelBackground}} we provide the precise details and relevant formulas.
\newline

\subsection{Main contributions and paper structure}

The main mathematical contribution of this paper is the development of an asymptotic theory for the R-SVD, when applied to signal-plus-noise matrices taken from the spiked model. Our results  parallel the existing theory available for the full SVD. The core phenomena is similar: the spectrum of the reduced data matrix---which we use as replacement for the truncated SVD of $\bY$---has a bulk-and-outliers structure, and the angles between the signal and empirical PCs tend to a deterministic limit. Using tools from random matrix theory, we derive asymptotically exact formulas for the positions of the outliers and the corresponding PC angles.

We emphasize that in the statistical regime we are interested in, the data matrix $\bY$ is \emph{not} low-rank, and is in fact well conditioned: $\sigma_1(\bY)/\sigma_{n\wedge m}(\bY)\sim const$; only the underlying signal $\bX$ is low-rank. In this regime, an error bound such as \eqref{eq:Halko-Bound}, while certainly \emph{true}, is of limited usefulness. Namely, at best we could deduce from it (e.g. using singular value and vector perturbation bounds) error bounds---an error interval---that are on the order $O(1)$, the same scale of the very quantities we are after. Thus, to get meaningful results in this regime, error bounds in operator norm, as \eqref{eq:Halko-Bound}, are a priori too crude.

Another important distinction between 
the asymptotic regime considered in this paper, compared to most previous works on the R-SVD such as \cite{halko2011finding}, is that the sketching dimension $d$ is scaled \emph{linearly} with the dimension: $d=\beta m$ for constant $\beta\in (0,1)$ as $n,m\to \infty$. Note that to ensure under our setup a finite bound in \eqref{eq:Halko-Bound}, recalling that $\sigma_{k+1}\sim const$, one must indeed take a linearly scaling $d\sim m\wedge n$. This setup is in line with \cite{yang2021reduce}, which studied sketched PCA in a similar asymptotic regime. 
Lastly, we remark that if one introduces $q\sim \log (n/d)$ power iterations into the basic algorithm, then it is possible to prove error bounds on the spectral norm which are $O(1)$ under any scaling of $d$; see for example \cite{halko2011finding,zhang2022perturbation}. Analyzing a setup of this kind is beyond the scope of our current techniques.

\paragraph*{} 
The paper is structured as follows. In Section~\ref{sec:Setup} we describe in detail the mathematical model within which we work. Section~\ref{sec:SpikeModelBackground} surveys known results on the (full SVD) of the spiked model, which serve as a benchmark for our new results on the R-SVD. 

In Section~\ref{sec:SetupMain} we state our main mathematical results, describing the asymptotic behavior of the singular values and vectors produced by the R-SVD. 

\Revision{
Section~\ref{sec:Discussion} is devoted to discussion, focusing on interpreting our results in the regime of small sketching dimension, $d/m \equiv \beta\ll 1$. We find that the SNR threshold for the emergence of an outlier singular value scales like $\sigma\asymp \beta^{-1/8}$; however, we also find that observed principal components corresponding to singular values of magnitude $\beta^{-1/8}\lesssim \sigma \lesssim \beta^{-1/2}$  tend to be very weakly correlated with the signal. In particular, observed PCs corresponding to outlying singular values which are very far away from the bulk may in fact be weakly correlated with the ground truth. This finding reveals a pitfall for PCA-based exploratory data analysis using the R-SVD, as we explain in Section~\ref{sec:StoryPCA}.

In Section~\ref{sec:SketchedPCA} we compare the R-SVD to sketched PCA, which is a more naive sketching-based method for dimensionality-reduced PCA. Recently \cite{yang2021reduce}, sketched PCA was analyzed under the spiked model. We show that the R-SVD attains better performance than sketched PCA, both for signal detection and estimation---the gap being particularly pronounced at very low sketching dimensions. For detection, while the R-SVD can detect signals of SNR $\sigma\gtrsim \beta^{-1/8}$, sketched PCA can only detect signals of SNR $\sigma\gtrsim \beta^{-1/4}$. For estimation, the R-SVD can reliably estimate the true data principal directions at SNR $\sigma\gtrsim \beta^{-1/4}$, whereas sketched PCA requires $\sigma\gtrsim \beta^{-1/2}$.
}

In Section~\ref{sec:Shrinker} we develop an optimal singular value shrinkage denoiser for the R-SVD. 

In Section~\ref{sec:Experiments} we present numerical experiments that demonstrate the applicability of our asymptotic theory in finite-$n$ settings.  
Finally, Section~\ref{sec:Proofs}
is devoted to the proofs of our main results, with some technical details deferred to the appendix. 

\subsection{Related works: sketching and randomized linear algebra}

While this paper deals exclusively with the R-SVD, we mention in passing that randomized sketching-based methods have been applied in recent years very fruitfully for other problems as well. 

The idea, at its core, is this: given a large matrix, one performs some form of dimensionality reduction to obtain a smaller matrix, on which the costly operation (in this paper: the SVD) is computationally feasible. To reduce the dimension, one typically multiplies by a random {sketching matrix}---for example a Gaussian i.i.d. matrix, or a Haar random projection---though we remark that substantial effort has gone towards constructing ``structured'' sketching matrices that support fast matrix-vector products, cf. \cite{achlioptas2003database,Ailon2009TheFJ,Ailon2008FastDR,rauhut2010compressive,krahmer2011new,kane2014sparser,jain2022fast}. The main mathematical insight, dating back to the pioneering work of Johnson and Lindenstrauss \cite{Johnson1984ExtensionsOL}, is that a random projection preserves, with overwhelming probability, the geometry of sufficiently low-dimensional structures (subspaces, small point clouds). Randomized linear algebra has since become a flourishing field, with a rich and vast literature. For an entry point, we refer to the following survey papers \cite{halko2011finding,mahoney2011randomized,woodruff2014sketching,kannan2017randomized,drineas2018lectures,martinsson2020randomized}.
    
    Besides the SVD,
    randomized sketching and subsampling methods have been extensively employed across a myriad of domains, with the goal of speeding up, or reducing the storage costs of, 
    computations involving large-dimensional matrices. A very partial list, with an eye towards applications in statistics and data science, includes: least squares regression \cite{sarlos2006improved,rokhlin2008fast,drineas2011faster,raskutti2016statistical,dobriban2019asymptotics}, ridge regression \cite{lu2013faster,chen2015fast,gonen2016solving,liu2019ridge}, principal component regression \cite{mor2019sketching}, two sample hypothesis testing \cite{lopes2011more,srivastava2016raptt}, clustering \cite{mixon2021sketching}, optimization \cite{pilanci2015randomized,tropp2017practical}, and many more.

    To our knowledge, only few existing works have studied algorithms from randomized linear algebra through the lens of a signal plus noise model, and under the spiked model in particular. Closest to the present paper is \cite{yang2021reduce}, which studied the asymptotic behavior of {sketched PCA} under the spiked model. 
    % In sketched PCA, one performs dimension reduction on the data matrix $\bY$, multiplying it to the right by a sketch matrix ${\hbY}=\bY\bOmega^\T$, and then taking the SVD. This is different from the R-SVD, which includes an additional subspace projection step; this makes the analysis involved considerably more contrived.
	
 Another relevant paper is \cite{zhang2022perturbation}, which considered, under a signal plus noise framework, a variant of sketched PCA which also include power iterations. Somewhat more loosely related is \cite{ma2022robust}, which studied low-rank matrix recovery from noisy sketches, giving error bounds for the double sketch method of \cite{fazel2008compressed} in the presence of noise. Both of these papers operate (namely, yield informative bounds) in a different asymptotic regime than the one considered here: that of diverging SNR.

	\section{Problem Setup and Assumptions}
 \label{sec:Setup}
	
	In this section, we describe in detail the mathematical model to be analyzed in this paper. It also serves to define some notation that will be used throughout.

	    We consider a setup where one observes an $n$-by-$m$ data matrix $\bY$ of the form $\bY=\bX+\bZ$, where $\bX$ is a low-rank, unknown, ``signal'' matrix---whose singular values and vectors are of interest---and $\bZ$ is ``noise''. 
	We work under a variant of the so-called {spiked model}, introduced by \cite{johnstone2001distribution}, wherein the problem dimensions $n,m\to \infty$ while the signal rank $r$ is held fixed. Consider an SVD of the signal matrix,
	\begin{equation}\label{eq:X-def}
		\begin{split}
			\bX &= \sum_{i=1}^r \sigma_i \bu_i \bv_i^\T  = \bU\bLambda \bV^\T
		\end{split}
	\end{equation}
	where
	\begin{equation}
		\bU=\MatL \bu_1 &\ldots &\bu_r \MatR \in \RR^{n\times r},\quad \bV= \MatL \bv_1 &\ldots &\bv_r \MatR \in \RR^{m\times r},\quad \bLambda=\diag(\sigma_1,\ldots,\sigma_r) .
	\end{equation}
	The matrices $\bU$ and $\bV$ collect, respectively, the left- and right- singular vectors of $\bX$, and so have orthonormal columns. 
	The corresponding singular values, enumerated in decreasing order, are held fixed as $n,m\to\infty$. For simplicity, they are assumed to be distinct:\footnote{This simplifying assumption is common throughout much of the literature on estimation in the spiked model, see for example {\cite{shabalin2013reconstruction,nadakuditi2014optshrink,gavish2017optimal,leeb2021optimal,donoho2020screenot}}.} $\sigma_1>\ldots>\sigma_r>0$. We make no a priori generative assumptions on the spike directions $\bU,\bV$, except that they have orthonormal columns.
	
	The noise matrix $\bZ$ is assumed to have independent and identically distributed (i.i.d.) Gaussian entries: $Z_{i,j}\overset{i.i.d.}{\sim}\m{N}(0,1/\sqrt{nm})$.\footnote{Note that this normalization is not-so standard in works dealing with PCA, such as \cite{yang2021reduce}, where the columns of $\bY$ are interpreted as i.i.d. samples from some high-dimensional distribution. Our normalization (following \cite{bloemendal2016principal}) makes  more sense in the context of matrix denoising, where the dimensions $n,m$ should have an equal role.}	
	We consider the so-called ``high-dimensional'' regime, where the dimensions diverge $n,m\to\infty$ at a constant aspect ratio:
	\begin{equation}
		\frac{m}{n} \to \gamma \in (0,\infty).
	\end{equation}
	
	Importantly, the noise variance is normalized so that the signal ($\bX$) and noise ($\bZ$) singular values are of the same scale, both being constant as $n,m\to\infty$. In particular, this is an SNR regime where consistent estimation (as $n,m\to\infty$) of the signal is generically not possible \cite{cai2016estimating,wainwright2019high}. 
	
	\subsection{Known results on the spiked model}
	\label{sec:SpikeModelBackground}
	
	Much is known about the singular value decomposition (SVD) of the data matrix $\bY$, and the relation between its principal components (PCs) to those of the signal matrix $\bX$. The singular values are arranged in the form of a bulk, whose limiting shape is a Marchenko-Pastur law, plus at most $r$ outliers exceeding the bulk edge. The outliers, and their corresponding singular vectors, are in one-to-one  correspondence with the signal spikes $\sigma_i$ and PCs $\bu_i,\bv_i$. The precise quantitative details {\cite{baik2005phase,baik2006eigenvalues,paul2007asymptotics,benaych2012singular,bloemendal2016principal}} are summarized below:
	
	\begin{fact}[The bulk singular values of $\bY$]
		\label{fact:1}
		The empirical distribution (counting measure) of the bulk singular values squared\footnote{In other words, the non-zero bulk \emph{eigenvalues} of $\bY\bY^\T$.}, $\sigma_{r+1}^2(\bY),\ldots,\sigma_{n\wedge m}^2(\bY)$ converges  weakly almost surely (as $n,m\to\infty$) to a Marchenko-Pastur law with shape parameter $\phi=\gamma\wedge\gamma^{-1}$ and scale parameter $\eta^2=(\gamma\vee\gamma^{-1})^{1/2}$:
		\begin{equation}
			\frac{1}{n\wedge m - r} \sum_{i=r+1}^{n\wedge m} \delta_{\sigma_{i}^2(\bY)} \overset{weakly}{\longrightarrow} \MPDensity_{\gamma\wedge \gamma^{-1},(\gamma\vee\gamma^{-1})^{1/2}} .
		\end{equation}
		The Marchenko-Pastur law with shape $0<\phi\le 1$ and scale $\eta^2>0$ has density
		\begin{equation}\label{eq:MP-law-general}
			\frac{d \MPDensity_{\phi,\eta}}{d\lambda}(\lambda) = \frac{1}{2\pi \eta^2\phi } \frac{\sqrt{(\lambda_{\phi,\eta^2}^{+}-\lambda)(\lambda-\lambda_{\phi,\eta^2}^{-})}}{\lambda},\quad \textrm{supported on}\quad \lambda\in \left[ \lambda_{\phi,\eta^2}^{-},\lambda_{\phi,\eta^2}^{+} \right],
		\end{equation}
		where 
		\begin{equation}
			\lambda_{\phi,\eta^2}^{\pm} = \eta^2(1 \pm \sqrt{\phi})^2.
		\end{equation} 
		
		In our case, the density can be written explicitly in terms of $\gamma$,
		\begin{equation}
			\frac{1}{2\pi \sqrt{\gamma\wedge \gamma^{-1}}} \frac{\sqrt{\Delta_{\gamma}(\lambda)}}{\lambda},\quad \textrm{supported on}\quad \lambda\in \left[ \NoiseEdge_{\gamma}^-,\NoiseEdge_\gamma^+ \right]
		\end{equation}
		where
		\begin{equation}
			\Delta_{\gamma}(\lambda)=(\NoiseEdge_{\gamma}^{+}-\lambda)(\lambda-\NoiseEdge_{\gamma}^{-}),\qquad\textrm{and}\quad \NoiseEdge_{\gamma}^{\pm} 
			% = \lambda_{\gamma\wedge\gamma^{-1},(\gamma\vee\gamma^{-1})^{1/2}} 
			= \gamma^{1/2}+\gamma^{-1/2}\pm 2 .
		\end{equation}
	\end{fact}
	~
	
	\begin{fact}
		[The leading PCs of $\bY$]
		\label{fact:2}
		Define the \emph{spike detection threshold}\footnote{Also often referred to in the literature as the Baik-Ben Arous-P\'{e}ch\'{e} (BBP) phase transition {\cite{baik2005phase}}.} $\sigma^*=1$. 
		\begin{itemize}
			\item {\bf The outlying singular values:} The matrix $\bY$ has at most $r$ singular values exceeding the Marchenko-Pastur bulk. Almost surely, for every constant $i\ge r+1$,
			\begin{equation}
				\sigma_{i}^2(\bY) \longrightarrow \NoiseEdge_\gamma^{+} \equiv \gamma^{-1/2}+\gamma^{1/2} + 2.
			\end{equation}
			A signal spike $\sigma_i$ creates an outlier if and only if it exceeds the detection threshold. If so, its asymptotic location is given by the \emph{spike-forward map}:
   for $1\le i \le r$,
			\begin{equation}\label{eq:SpikeFwdClassical}
				\sigma_{i}^2(\bY) \longrightarrow \begin{cases}
					\m{Y}_{\gamma}^2(\sigma_i)\quad&\textrm{if }\;\sigma_i\ge \sigma^* \\
					\NoiseEdge_{\gamma}^+\quad&\textrm{if }\;\sigma_i\le \sigma^*
				\end{cases}\,,\quad\textrm{}\quad 
				\m{Y}_{\gamma}(\sigma) \equiv \sqrt{ (\gamma^{\frac14}\sigma + \gamma^{-\frac14}\sigma^{-1})(\gamma^{-\frac14}\sigma + \gamma^{\frac14}\sigma^{-1}) } \,.
			\end{equation}
			
			\item {\bf Principal component angles:} Denote by $\hbu_i,\hbv_i$, respectively, the left- and right- singular vectors of $\bY$. There is a 1-1 correspondence between the leading $r$ principal components of $\bX$ and $\bY$:
			\begin{equation}
				\langle \bu_i, \hbu_j\rangle,\; \langle \bv_i, \hbv_j\rangle \longrightarrow 0 \quad\textrm{whenever}\quad 1\le i \ne j \le r \,.
			\end{equation}
			That is, non-corresponding PCs are asymptotically orthogonal.
			Moreover, the angle (equivalently the correlation/overlap) between corresponding PCs converges to a deterministic value. For undetectable spikes, the PCs are asymptotically orthogonal, while above the threshold they become increasingly aligned (as the SNR increases):
			\begin{equation}\label{eq:fact:vectors-u}
				\left|\langle \bu_i, \hbu_j\rangle\right|\longrightarrow \begin{cases}
					\m{U}_{\gamma}(\sigma_i) \quad&\textrm{if }\;\sigma_i\ge \sigma^* \\
					0 \quad&\textrm{if }\;\sigma_i\le \sigma^*
				\end{cases}\,, \qquad \m{U}_{\gamma}(\sigma) \equiv \sqrt{\frac{\sigma^4-1}{\sigma^4+\gamma^{-\frac12}\sigma^2}}\,,
			\end{equation}
			and 
			\begin{equation}\label{eq:fact:vectors-v}
				\left|\langle \bv_i, \hbv_j\rangle\right|\longrightarrow \begin{cases}
					\m{V}_{\gamma}(\sigma_i) \quad&\textrm{if }\;\sigma_i\ge \sigma^* \\
					0 \quad&\textrm{if }\;\sigma_i\le \sigma^*
				\end{cases}\,, \qquad \m{V}_{\gamma}(\sigma) \equiv \sqrt{\frac{\sigma^4-1}{\sigma^4+\gamma^{\frac12}\sigma^2}}\,.
			\end{equation}
			Moreover,
			\begin{equation}\label{eq:fact:dyad}
				\langle \bu_i, \hbu_j\rangle \langle \bv_i, \hbv_j\rangle\longrightarrow \begin{cases}
					{\m{U}_{\gamma}(\sigma)\m{V}_{\gamma}(\sigma_i)} \quad&\textrm{if }\;\sigma_i\ge \sigma^* \\
					0 \quad&\textrm{if }\;\sigma_i\le \sigma^*
				\end{cases}\,.
			\end{equation}
			
		\end{itemize}
		
	\end{fact}
	
	The present paper aims to develop an asymptotic picture, similar to the one described in Facts~\ref{fact:1}-\ref{fact:2}, for the approximated singular values and vectors obtained through the R-SVD algorithm.

	\subsection{The randomized SVD (R-SVD) algorithm}
	
	We repeat the description of the R-SVD algorithm from  \cite{halko2011finding}, introducing some notation along the way. It proceeds as follows: 
	\begin{enumerate}[label=(\Roman*)]
		\item Let $1\le d \le m$ be the {\it sketching dimension}. Let $\bOmega \in \RR^{d\times m}$ be the sketching matrix, which is statistically independent of $\bX,\bZ$.
		Popular choices of $\bOmega$ include a Gaussian i.i.d. matrix, or a projection onto a uniformly (Haar) random $d$-dimensional subspace of $\RR^m$. While our analysis extends beyond these two particular choices, we \emph{will} require that $\bOmega$ acts essentially ``random-like'' on the signal right singular vectors; see details below.        
		One forms the {\it sketched data matrix}:
		\begin{equation}
			\bYtilde = \bY\bOmega^\T \;\in\; \RR^{n\times d}\,.
		\end{equation}
		
		\item Next, one finds an orthonormal basis for the range (column space) of $\bYtilde$. Note that if $d\ge n$, then the range of $\bYtilde$ is identical to that of $\bY$. This case is uninteresting,  and so we will always assume that $d< n$. Thus, the range of $\bYtilde$ is, w.p. $1$ (since $\bZ$ has a density), a $d$-dimensional proper subspace of $\RR^n$. Let $\bPc:\RR^n\to \RR^n$ be the projection operator onto this subspace. Note that $\bPc$ may be easily computed from a QR decomposition of $\bY$: If $\bYtilde = \bm{Q}\bm{R}$ then $\bPc=\bm{Q}\bm{Q}^\T$. 
		
		\item Finally, one takes the SVD of the reduced matrix 
		\begin{equation}\label{eq:hbY-def}
			\hbY = \bPc \bY,
		\end{equation}
		denoted
		\begin{align}
			\hbY \overset{SVD}{=} \sum_{i=1}^d \hsigma_i \hbu_i \hbv_i^\T \,.
		\end{align}
		Central to the R-SVD algorithm is the intuition that the {leading} ($i\ll d$) PCs of $\hbY$ are good proxies for the leading PCs of the un-reduced matrix $\bY$. 
%		Recall, however, that our ultimate interest in these empirical PCs is to be used as stand-ins for the PCs of the signal $\bX$, which are unknown.   
		Note that to compute the SVD of $\hbY$, one only needs to compute the SVD of a $d$-by-$m$ matrix (instead of $n$-by-$m$):
		First, (i) we compute the SVD of $\bm{Q}^\T \tilde{\bY} \in \RR^{d\times m}$; and then (ii) multiply the resulting left singular vectors by $\bm{Q} \in \RR^{n\times d}$.  
		
	\end{enumerate}
	
	As mentioned before, the sketching dimension $d$ is taken proportional to the signal dimensions; specifically, for a constant undersampling ratio $\beta\in (0,1]$, we assume\footnote{In particular, note the requirement $\beta\gamma<1$ (strict inequality). When $\beta\gamma\ge 1$ no dimension reduction is actually performed, so $\hbY=\bY$ exactly, and the formulas from Section~\ref{sec:SpikeModelBackground} apply. Some of our  results are written in terms of certain compound algebraic expressions, which we could not reduce into a concise closed form and which exhibit singularities at $\gamma\beta=1$ (however these singularities ultimately do cancel out). For this reason, the formulas given in Section~\ref{sec:SetupMain} are not directly applicable when $\gamma\beta= 1$.}
	\begin{equation}
		\frac{d}{m}\to \beta \in (0,1] \qquad\textrm{and moreover that}\quad  \frac{d}{n}\to \gamma\beta \in (0,1) \qquad\textrm{as }\quad n,m,d\to\infty .
	\end{equation}

	\paragraph*{The sketching matrix.} We always assume that $\bOmega$ is full-rank: $\rank(\bOmega)=d$. Also note that we may assume without loss of generality that $\bOmega$ is a projection matrix (that is, has orthonormal rows). To see this, take the SVD, $\bOmega=\bU_{\bOmega}\bSigma_{\bOmega}\bV_{\bOmega}$ and write $\tilde{\bY}=\bY\bOmega^\T=\bY \bV_{\bOmega}^\T \bSigma_{\bOmega}\bU_{\bOmega}^\T$. Since $\rank(\bSigma_{\bOmega}\bU_{\bOmega}^\T)=d$, we have $\range(\bY\bOmega^\T)=\range(\bY\bV_{\bOmega}^\T)$. Since 
	$\hbY$ only depends on $\bOmega$ through $\range(\bY\bOmega^\T)$, see \eqref{eq:hbY-def}, it would not change at all if we replaced $\bOmega$ by its matrix of right singular vectors $\bV_{\bOmega}^\T\in \RR^{d\times m}$.
	
	We make a strong incoherence assumption between the signal right singular vectors of $\bX$ and the matrix $\bV_{\bOmega}$: w.p. $1$,
	\begin{equation}\label{eq:assum:incoherence}
		\langle \bV_{\bOmega}^\T \bv_i, \bV_{\bOmega}^\T \bv_j \rangle \longrightarrow \beta \, \Indic{i=j}\quad\textrm{for}\quad 1\le i,j\le r,
	\end{equation}
where recall that $\beta\equiv d/m$. The assumption \eqref{eq:assum:incoherence} dictates that $\bV_{\bOmega}^\T$ behaves on the signal right singular vectors $\{\bv_1,\ldots,\bv_r\}$ essentially 
like a random projection would: 1) it preserves (asymptotically) the orthogonality between different singular vectors; 2) the proportion of energy retained in a component after the projection is equal to the down-sampling ratio: $\|\bV_{\bOmega}^\T \bv_i\|^2\approx d/m$.

Clearly, if $\bOmega\in \RR^{d\times m}$ is either a random projection or an i.i.d. Gaussian matrix, then \eqref{eq:assum:incoherence} holds. (In fact, note that if $\bOmega$ is Gaussian then $\bV_{\bOmega}^\T$ is a random projection.)
Other examples where \eqref{eq:assum:incoherence} is satisfied are when $\bOmega$ is a randomized sub-sampled Hadamard or discrete Fourier transform matrix, or when $\bOmega$ is a uniformly random coordinate sub-sampling operator and the population spikes $\bv_1,\ldots,\bv_r$ are all sufficiently de-localized. For details, see for example \cite{yang2021reduce}.

    \section{Main Results}
    \label{sec:SetupMain}

    This paper studies the spectrum of the reduced data matrix $\hbY$ in \eqref{eq:hbY-def}.
    Our main results describe a phenomenology for its singular values and vectors, that parallels that of Facts~\ref{fact:1}-\ref{fact:2} for the singular values and vectors of $\bY$. Like the un-reduced data matrix, the singular values of $\hbY$ are arranged in a bulk-and-outliers structure, which can be related in an explicit sense to the singular values and vectors of the signal $\bX$.

    Our first theorem pertains to the bulk singular values, namely parallels Fact~\ref{fact:1}. 

    \begin{theorem}
        [The bulk singular values of $\hbY$] 
        \label{thm:1:Bulk}
        The empirical distribution of the bulk singular values squared $\sigma^2_{r+1}(\hbY),\ldots,\sigma^2_{d}(\hbY)$ converges weakly almost surely, as $n,m,d\to\infty$, to a Marchenko-Pastur law with shape and scale parameters
        \begin{equation}\label{eq:thm:QZ-LSD}
            \phi = \frac{\gamma\beta}{1+\gamma-\gamma\beta},\qquad \eta^2 = \gamma^{-1/2}+(1-\beta)\gamma^{1/2}.
        \end{equation} 
        Written explicitly, the limiting density is 
        \begin{equation}\label{eq:thm:1:density}
            \frac{1}{2\pi\beta\sqrt{\gamma}} \frac{\sqrt{\Delta_{\gamma,\beta}(\lambda)}}{\lambda},\qquad \textrm{supported on}\quad \lambda \in \left[ \NoiseEdge_{\gamma,\beta}^{-},\NoiseEdge_{\gamma,\beta}^{+} \right],
        \end{equation}
        where
        \begin{equation}\label{eq:edges-def}
            \NoiseEdge_{\gamma,\beta}^{\pm} = \gamma^{-1/2}+\gamma^{1/2} \pm 2\sqrt{\beta(1+\gamma-\gamma\beta)}\,,
        \end{equation}
        and
        \begin{equation}
            \begin{split}
                \Delta_{\gamma,\beta}(\lambda) 
                &= \left(\NoiseEdge_{\gamma,\beta}^{+}-\lambda\right)\left( \lambda-\NoiseEdge_{\gamma,\beta}^{-} \right) \\
                &= 
                % -\frac1\gamma \left( \delta(y)^2 - 4\beta\gamma(1+\gamma-\gamma\beta \right)\quad\textrm{where}\quad \delta(\lambda) = \sqrt{\gamma}\lambda-1-\gamma\,.
                -\frac{1}{\gamma}\left( (\sqrt{\gamma}\lambda-1-\gamma)^2-4\beta\gamma(1+\gamma-\gamma\beta) \right).
            \end{split}
        \end{equation}
    \end{theorem}

\paragraph*{}
    Next we describe the behavior of the leading singular values and vectors of $\hbY$. To this end, define the following functions $\kappa_{\gamma,\beta}^{(i)}: (\sqrt{\NoiseEdge_{\gamma,\beta}^{+}},\infty) \to \RR$:
    \begin{equation}\label{eq:kappa-1}
       \kappa_{\gamma,\beta}^{(1)}(y) = \frac{\sqrt{\gamma}y^2 + 1-\gamma -\sqrt{-\gamma\Delta_{\gamma,\beta}(y^2)}}{2y} ,
   \end{equation}
   \begin{equation}
       \label{eq:kappa-2}
       \kappa_{\gamma,\beta}^{(2)}(y) = (1-\beta\gamma)\frac1y ,
   \end{equation}
  \begin{equation}\label{eq:kappa-3}
      \kappa_{\gamma,\beta}^{(3)}(y) = \frac{\sqrt{\gamma}(1+\gamma)y^2 - (1+\gamma-2\beta\gamma)^2 - (1+\gamma-2\beta\gamma)\sqrt{-\gamma \Delta_{\gamma,\beta}(y^2)}}{2\gamma(1-\beta\gamma)y^3} ,
  \end{equation}
  \begin{equation}\label{eq:kappa-4}
      \kappa_{\gamma,\beta}^{(4)}(y) 
      = \frac{ \sqrt{\gamma}y^2 - (1+\gamma-2\beta\gamma) - \sqrt{-\gamma\Delta_{\gamma,\beta}(y^2)} }
      { 2\gamma y^2 } ,
  \end{equation}
  \begin{equation}\label{eq:kappa-5}
      {\kappa}^{(5)}_{\gamma,\beta}(y) = \frac
      { -(1+\beta)\sqrt{\gamma}y^2 + (1-\beta) \left((1+\gamma-2\beta\gamma) + \sqrt{-\gamma\Delta_{\gamma,\beta}(y^2)} \right)}
      {2(1-\beta\gamma)y^2},
  \end{equation}
  \begin{equation}\label{eq:kappa-6}
      \kappa^{(6)}_{\gamma,\beta}(y) = \frac{\sqrt{\gamma}y^2 - (1-\gamma) - \sqrt{-\gamma \Delta_{\gamma,\beta}(y^2)}}{2\gamma y} ,
  \end{equation}
  \begin{equation}
      \label{eq:kappa-7}
      \kappa_{\gamma,\beta}^{(7)}(y) = 
      \frac{1-\beta}{y}
      % (1-\beta)\frac
      %  {\sqrt{\gamma}y^2 + \left(1+\gamma -2\beta\gamma\right) - \sqrt{-\gamma\Delta_{\gamma,\beta}(y^2)}}
      %  {2(1+\gamma-\beta\gamma)y},
  \end{equation}
  \begin{equation}
      \label{eq:kappa-8}
      \kappa^{(8)}_{\gamma,\beta}(y) = (1-\beta)\frac
      {\sqrt{\gamma}y^2 + 1 -\gamma -\sqrt{-\gamma\Delta_{\gamma,\beta}(y^2)}}
      {2(1-\beta\gamma)y },
  \end{equation}
  \begin{equation}
      \label{eq:kappa-9}
      \kappa_{\gamma,\beta}^{(9)}(y) = \sqrt{\gamma}(1-\beta)(1-\beta\gamma) \frac
       {\sqrt{\gamma}y^2 + \left(1+\gamma -2\beta\gamma\right) - \sqrt{-\gamma\Delta_{\gamma,\beta}(y^2)}}
       {2(1+\gamma-\beta\gamma)y} .
  \end{equation}
  Furthermore, define 
  the  matrix-valued function $\bKc_{\gamma,\beta}:(\sqrt{\NoiseEdge_{\gamma,\beta}^+},\infty)\to \mathrm{Sym}^6(\RR)$:
  \begin{equation}\label{eq:K-func}
      \bKc_{\gamma,\beta}(y) = \MatL 
          \kappa_{\gamma,\beta}^{(1)}(y) &\kappa_{\gamma,\beta}^{(2)}(y) &0 &0 &0 &0\\
      \kappa_{\gamma,\beta}^{(2)}(y) &\kappa_{\gamma,\beta}^{(2)}(y) &0 &0 &0 &-(1-\gamma\beta)  \\
      0 &0 &\kappa_{\gamma,\beta}^{(3)}(y) & \kappa_{\gamma,\beta}^{(4)}(y) &\kappa_{\gamma,\beta}^{(5)}(y)  &0 \\
      0 &0 &\kappa_{\gamma,\beta}^{(4)}(y) &\kappa_{\gamma,\beta}^{(6)}(y)&\kappa_{\gamma,\beta}^{(7)}(y) &0\\
      0 &0 &\kappa_{\gamma,\beta}^{(5)}(y)(y) &\kappa_{\gamma,\beta}^{(7)}(y) (y)&\kappa_{\gamma,\beta}^{(8)}(y) &0\\ 
      0 &-(1-\gamma\beta) &0 &0 &0 &\kappa_{\gamma,\beta}^{(9)}(y)    
      \MatR
  \end{equation}
  Denote also the following symmetric $6$-by-$6$ matrix:
  \begin{equation}
      \label{eq:H}
      \Hbb = 
      \MatL \0 & \bI^{-}_3 \\
      \bI^{-}_3 &\0
      \MatR \qquad\textrm{where} \quad \bI^{-}_3 = \MatL 1 &0 &0 \\ 0 &-1 &0 \\ 0 &0 &1 \MatR .
    %   \MatL
    %   0 &0 &0 &1 &0 &0 \\
          % 0 &0 &0 &0 &-1 &0 \\
          % 0 &0 &0 &0 &0 &1 \\
          % 1 &0 &0 &0 &0 &0\\
          % 0 &-1 &0 &0 &0 &0\\
          % 0 &0 &1 &0 &0 &0
    %   \MatR.
  \end{equation}

  The solution to the generalized eigenvalue problem
  \begin{equation}
    \label{eq:eigenvalue-problem}
    \det\left( \bKc_{\gamma,\beta}(y) - s\bHc\right) = 0,\qquad s\ge 0\,,
  \end{equation}
  plays an important role in our results to follow.

  \begin{proposition}\label{prop:PosEigFunc-properties}
    For any $y> \sqrt{\NoiseEdge^+_{\gamma,\beta}}$, Eq. \eqref{eq:eigenvalue-problem} has a unique positive root $s> 0$. 

    Denote this root by $\PosEigFunc(y)$. Furthermore, the following holds.

    \begin{enumerate}
        \item $\PosEigFunc(y)$ is a generalized eigenvalue of multiplicity $1$: $\rank(\bKc_{\gamma,\beta}(y)-\PosEigFunc(y)\bHc)=5$.
        \item The function $y\mapsto \PosEigFunc(y)$ is strictly decreasing, with $\PosEigFunc(\infty) = 0$, and its value at $y=\NoiseEdgeUpper$ is
        \Revision{
        \begin{equation}\label{eq:BBP-squared}
            \frac{1}{\PosEigFunc^2 ( \sqrt{\NoiseEdge_{\gamma,\beta}^+} )}
            = 
            \sqrt{v^2 + \sqrt{\frac{\rho}{\gamma}}\beta^{-1/2}} -v
        \end{equation}
        where
        \begin{equation}\label{eq:v-rho}
        v=\frac12(\gamma^{1/2}+\gamma^{-1/2}-\beta\gamma^{1/2}-\sqrt{\beta\rho}),\qquad 
        \rho = 1+\gamma-\beta\gamma\,.
        \end{equation}
        }
%         \begin{equation}
%             \label{eq:BBP-squared}
%             \begin{split}
% %        \left(1/\PosEigFunc ( \sqrt{\NoiseEdge_{\gamma,\beta}^+} )\right)^2
% \frac{1}{\PosEigFunc^2 ( \sqrt{\NoiseEdge_{\gamma,\beta}^+} )}
%         &=
%         \frac{\sqrt{\beta  \gamma  (1+\gamma -\beta  \gamma )}-(1+\gamma -\beta  \gamma )
%         }{2 \sqrt{\gamma }} \\
%         &+ 
%         \frac{
%         \sqrt{\left(\frac{4}{\beta }-2 (1+\gamma -\beta  \gamma )\right) \sqrt{\beta  \gamma  (1+\gamma -\beta  \gamma )}+(1+\gamma -\beta  \gamma ) (1+\gamma )}}
%         {2 \sqrt{\gamma }}\,.
%         \end{split}
%         \end{equation}
        \item The functional inverse $\PosEigFunc^{-1}(\cdot)$ has the following explicit formula:
        \begin{equation}
            \label{eq:PosEigFunc-inv}
            \PosEigFunc^{-1}(s) = 
            \sqrt{\frac{\left(s^2+\left(1+s^4\right) \sqrt{\gamma }+s^2 \gamma \right) \left(\beta  \gamma +s^4 (1+\gamma -\beta  \gamma )+2 s^2 \beta  \sqrt{\gamma } (1+\gamma -\beta  \gamma )\right)}{s^2 \left(s^2+\beta  \sqrt{\gamma }\right) \sqrt{\gamma }
   \left(\sqrt{\gamma }+s^2 (1+\gamma -\beta  \gamma )\right)}}\,.
        \end{equation}
    \end{enumerate}
  \end{proposition}
%   The proof of Proposition~\ref{prop:PosEigFunc-properties} appears in Appendix, Section~\ref{sec:proof-prop:PosEigFunc-properties}. 

\paragraph*{}
\Revision{
Define the \emph{spike detection threshold} $\BBP =1/\PosEigFunc ( \sqrt{\NoiseEdge_{\gamma,\beta}^+} )$; explicitly,
\begin{equation}\label{eq:BBP}
    \BBP = \sqrt{ \sqrt{v^2 + \sqrt{\frac{\rho}{\gamma}}\beta^{-1/2}} -v }
\end{equation}
with $v,\rho$ as in \eqref{eq:v-rho}. (Note that for $\beta=1$, $\rho=1,v=\frac12(\gamma^{-1/2}-1)$; so $\sigma_{\gamma,\beta=1}^*=1$ coincides with the detection threshold from Section~\ref{sec:SpikeModelBackground}.)
}

Define also
the \emph{spike-forward map}, $\SpikeFunc(\sigma) = \PosEigFunc^{-1}(1/\sigma)$, given explicitly by the formula
\begin{align}
    \label{eq:SpikeFunc}
    \SpikeFunc(\sigma) 
    &= \sqrt{\frac{\left(\sqrt{\gamma }+\sigma ^2\right) \left(1+\sqrt{\gamma } \sigma ^2\right) \left(1+\gamma -\beta  \gamma +2 \beta  \sqrt{\gamma } (1+\gamma -\beta  \gamma ) \sigma ^2+\beta  \gamma  \sigma ^4\right)}{\sqrt{\gamma } \sigma ^2 \left(1+\gamma
    -\beta  \gamma +\sqrt{\gamma } \sigma ^2\right) \left(1+\beta  \sqrt{\gamma } \sigma ^2\right)}} \nonumber \\
	&= \m{Y}_{\gamma}(\sigma) \sqrt{\frac{\left(1+\gamma -\beta  \gamma +2 \beta  \sqrt{\gamma } (1+\gamma -\beta  \gamma ) \sigma ^2+\beta  \gamma  \sigma ^4\right)}{\left(1+\gamma
			-\beta  \gamma +\sqrt{\gamma } \sigma ^2\right) \left(1+\beta  \sqrt{\gamma } \sigma ^2\right)}}\,,
\end{align}
where $\m{Y}_{\gamma}(\sigma)$ is defined in \eqref{eq:SpikeFwdClassical}.
Note that $\SpikeFunc(\cdot)$ is  a bijective, increasing map between $(\BBP,\infty)\mapsto (\sqrt{\NoiseEdge_{\gamma,\beta}^+},\infty)$. 

\begin{figure}
	\centering
	\includegraphics[width=0.45\textwidth]{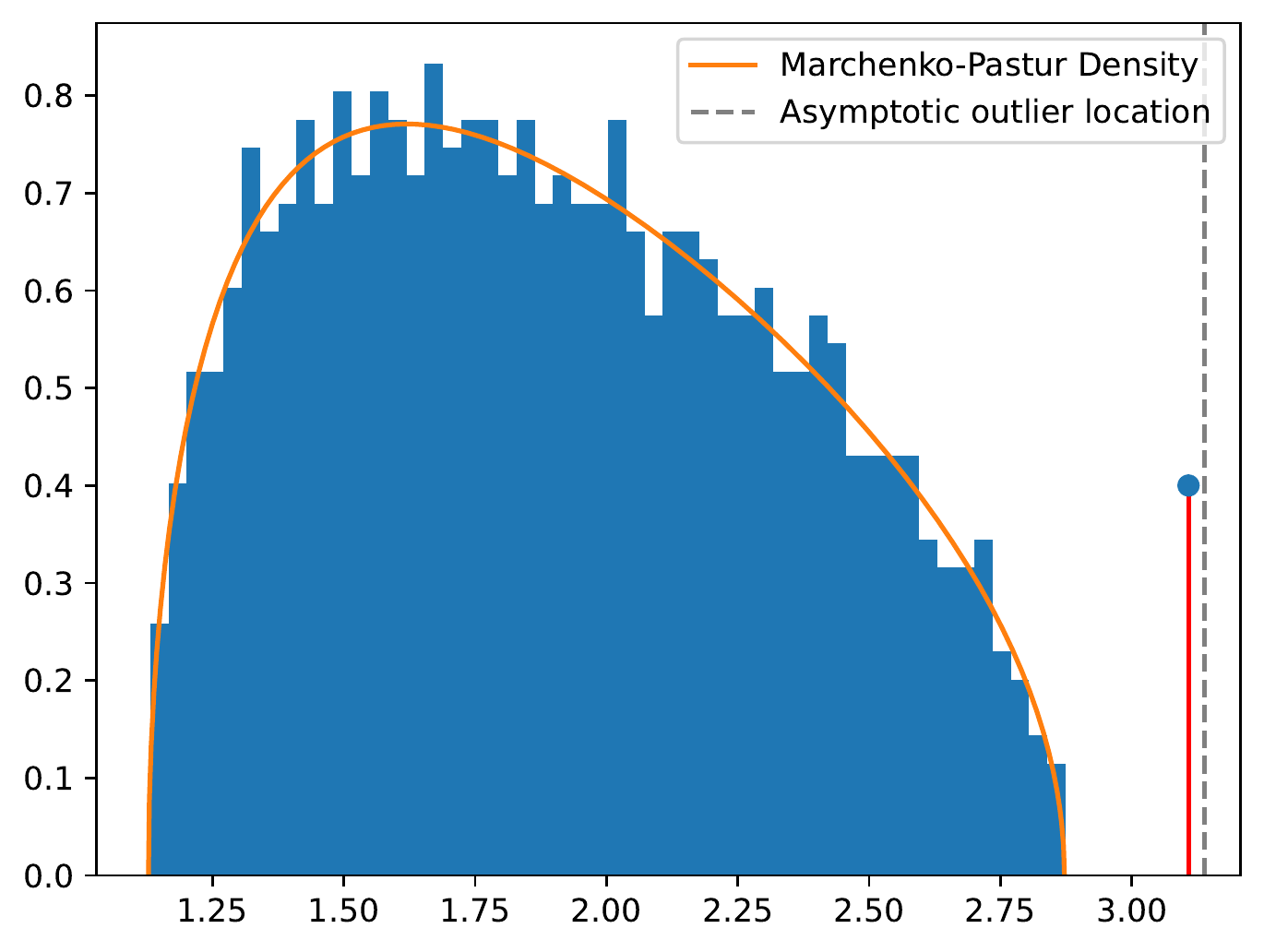}
    \includegraphics[width=0.45\textwidth]{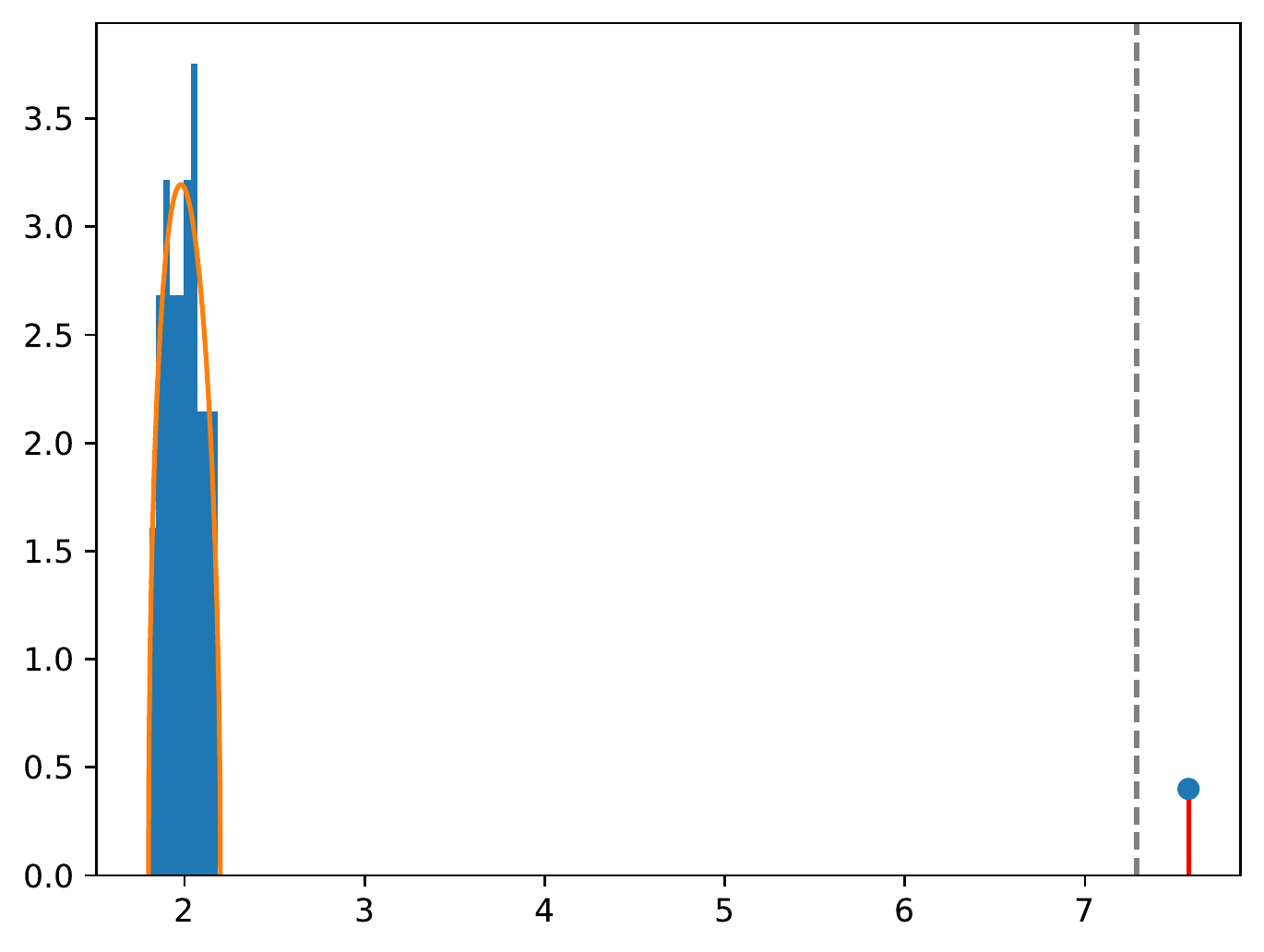}
	\caption{
    \Revision{    
    An illustration of Theorems~\ref{thm:1:Bulk} and \ref{thm:2:Outliers}.
	Plotted is a histogram of the eigenvalues of $\hbY\hbY^\T$ in a single spiked setup $\bY=\sigma \bu\bv^\T+ \bZ$ ($r=1$).
    The sub-leading eigenvalues, $\sigma_2^2(\hbY),\ldots,\sigma_d^2(\hbY)$ are arranged in a bulk, whose limiting shape (orange curve) is a Marcheko-Pastur law per Theorem~\ref{thm:1:Bulk}. The largest eigenvalue $\sigma_1^2(\hbY)$ is an outlier, where asymptotically, per Theorem~\ref{thm:2:Outliers}, $\sigma_1^2(\hbY) \to \SpikeFunc^2(\sigma)$ (grey vertical line).
    In both instances $n=m=10^4$ (hence $\gamma=1$). Left: $d=10^3$, corresponding to $\beta=1/10$, and $\sigma=\BBP+0.4\approx 1.6$. Right: $d=50$, corresponding to $\beta=1/200$, and $\sigma=\beta^{-1/3}\approx 5.85$.     
    % with $n=m=10^4$ ($\gamma=1$), $m=10^3$ ($\beta=1/10$) and $\sigma=\BBP+0.4\approx 1.6$. 
    }
	}
\label{fig:MainHistogram}
\end{figure}

The following result describes the behavior of the $r$ largest singular values of $\hbY$. It parallels the first bullet of Fact~\ref{fact:2}:
    \begin{theorem}[The outlying singular values of $\hbY$]
    \label{thm:2:Outliers}
        The matrix $\hbY$ has at most $r$ singular values exceeding the upper bulk edge from Theorem~\ref{thm:1:Bulk}. To wit, for every constant $i\ge r+1$, almost surely,
        \[
            \sigma_{i}^2(\hbY) \longrightarrow \NoiseEdge_{\gamma,\beta}^+ .
        \]
        As for the largest $r$ singular values, they satisfy:
        \begin{equation}\label{eq:thm:2:Outliers}
            \sigma_i^2(\hbY)\longrightarrow
            \begin{cases}
            \SpikeFunc^2(\sigma_i)\quad&\textrm{if }\;\sigma_i > \BBP, \\
                \NoiseEdge_{\gamma,\beta}^+ \quad&\textrm{if }\;\sigma_i\le \BBP
            \end{cases} .
        \end{equation}
        % where $\SpikeFunc(\cdot)$ is defined in \eqref{eq:SpikeFunc}.
    \end{theorem}  

	\paragraph*{}
	In Figure~\ref{fig:MainHistogram} we provide a visual illustration of Theorems~\ref{thm:1:Bulk} and \ref{thm:2:Outliers}: it plots the histogram of the eigenvalues of $\hbY\hbY^\T$ with one spike above the detection threshold. All but one of the eigenvalues are arranged in a Marcheko-Pastur bulk; the largest eigenvalue is an outlier, whose asymptotic location is described by Theorem~\ref{thm:2:Outliers}.
	Figure~\ref{fig:MainSpikeFwd} shows the detection threshold $\BBP$ and the spike-forward map $\SpikeFunc(\cdot)$ for selected parameter combinations.
	\begin{figure}
		\centering
		\begin{subfigure}[t]{0.45\textwidth}
			\centering
			\includegraphics[width=\textwidth]{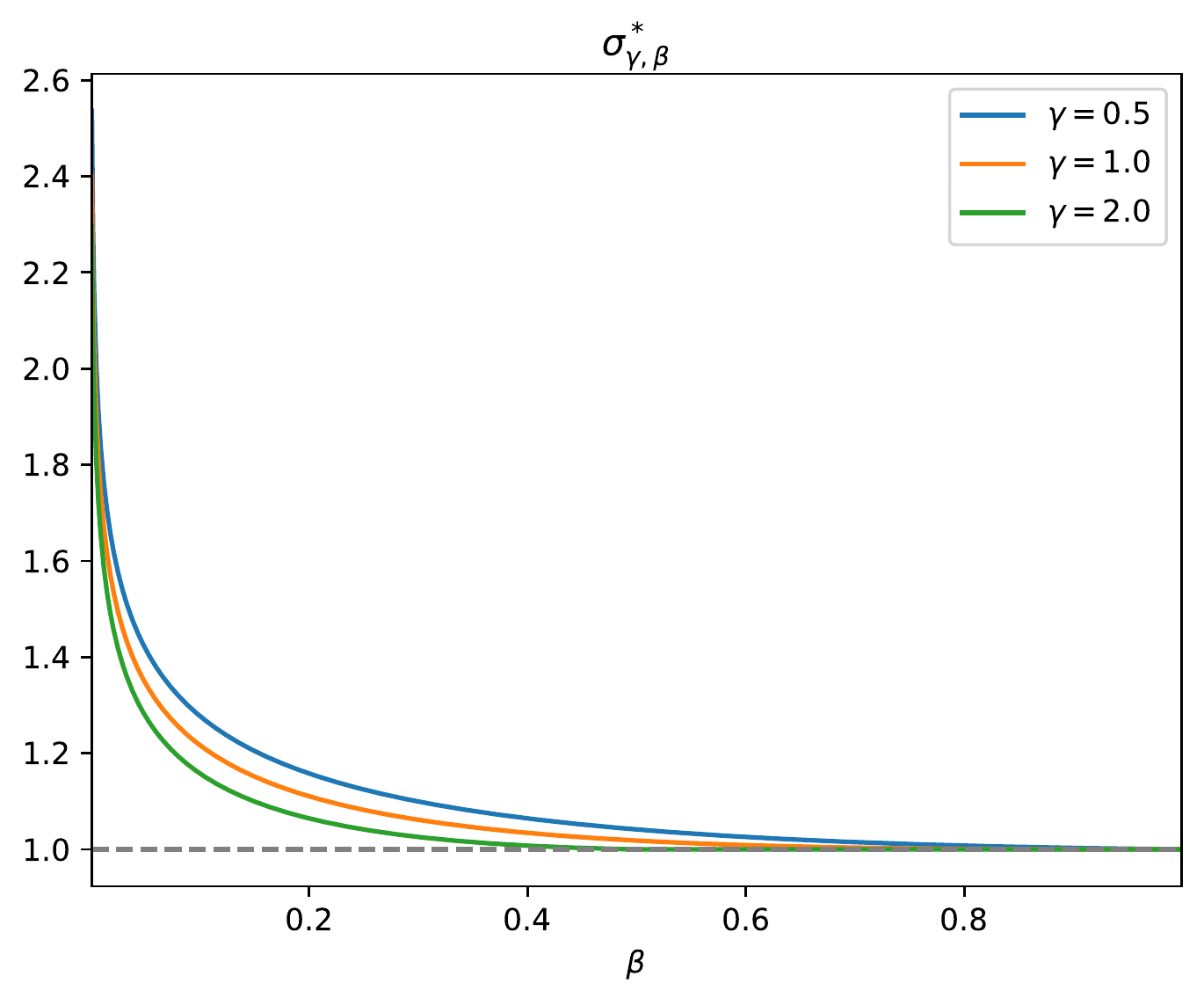}
			\caption{The detection threshold $\BBP$ as a function of $\beta$. Dashed line: $\sigma^*=1$, the threshold with no dimension reduction; attained once $\beta=\min\{1,1/\gamma\}$.}
		\end{subfigure}
	\hfill
		\begin{subfigure}[t]{0.45\textwidth}
			\centering
			\includegraphics[width=\textwidth]{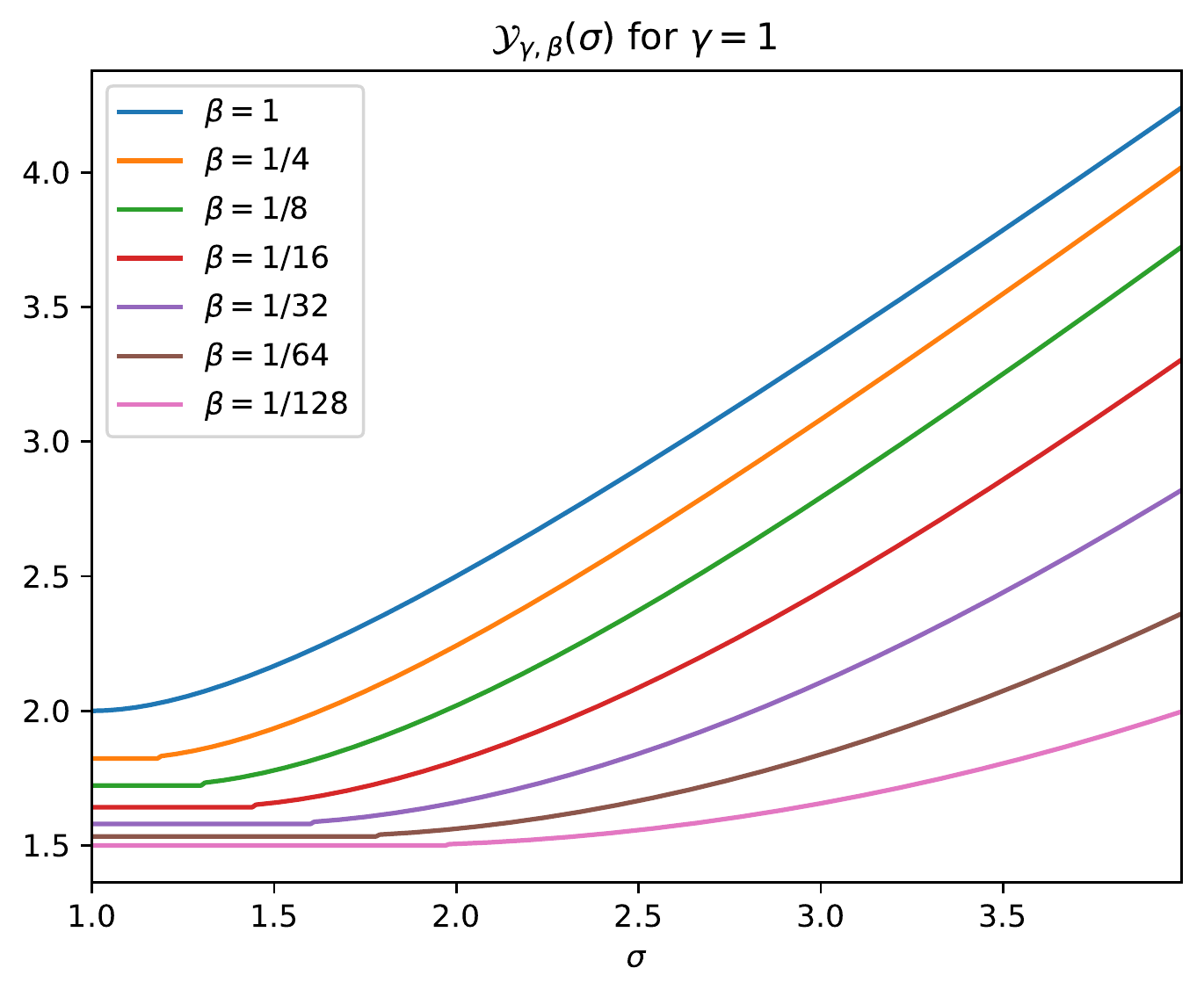}
			\caption{
                \Revision{The spike-forward map $\SpikeFunc(\cdot)$ for $\gamma=1$ and selected values of $\beta$.} 
%				Dashed line: the curve $\m{Y}(\sigma)=\sigma$; in the limit $\sigma\to\infty$, the spike-forward map aligns with this curve.
}
		\end{subfigure}
	\caption{}
	\label{fig:MainSpikeFwd}
	\end{figure}

    \paragraph*{}
    Next, analogously to \eqref{eq:fact:vectors-u}-\eqref{eq:fact:vectors-v}, we calculate limiting expressions for the asymptotic overlaps between the signal and observed PCs. To this end, define the following auxiliary functions:
    \begin{align}
        \tau_{\gamma,\beta}^{(i)}(y)&=-\frac12 y \left( \frac1y \cdot \kappa_{\gamma,\beta}^{(i)}(y)\right)', \qquad\textrm{for}\quad 1\le i \le 3, \\
        \tau_{\gamma,\beta}^{(i)}(y)&=-\frac12 \left( \kappa_{\gamma,\beta}^{(i)}(y)\right)', \qquad\textrm{for}\quad 4\le i \le 5, \\
        \tau_{\gamma,\beta}^{(i)}(y)&=-\frac{1}{2y} \left( y \cdot\kappa_{\gamma,\beta}^{(i)}(y)\right)', \qquad\textrm{for}\quad 6\le i \le 9,
    \end{align}
    where $(\cdot)'$ is the derivative with respect to $y$.
    Define $\bTc_{\gamma,\beta}:(\sqrt{\NoiseEdge_{\gamma,\beta}^+},\infty)\to \mathrm{Sym}^6(\RR)$:
\begin{equation}\label{eq:T-func}
    \bTc_{\gamma,\beta}(y) = \MatL 
        \tau_{\gamma,\beta}^{(1)}(y) &\tau_{\gamma,\beta}^{(2)}(y) &0 &0 &0 &0\\
    \tau_{\gamma,\beta}^{(2)}(y) &\tau_{\gamma,\beta}^{(2)}(y) &0 &0 &0 &0  \\
    0 &0 &\tau_{\gamma,\beta}^{(3)}(y) & \tau_{\gamma,\beta}^{(4)}(y) &\tau_{\gamma,\beta}^{(5)}(y)  &0 \\
    0 &0 &\tau_{\gamma,\beta}^{(4)}(y) &\tau_{\gamma,\beta}^{(6)}(y)&\tau_{\gamma,\beta}^{(7)}(y) &0\\
    0 &0 &\tau_{\gamma,\beta}^{(5)}(y)(y) &\tau_{\gamma,\beta}^{(7)}(y) (y)&\tau_{\gamma,\beta}^{(8)}(y) &0\\ 
    0 &0 &0 &0 &0 &\tau_{\gamma,\beta}^{(9)}(y)    
    \MatR
\end{equation}

    \begin{theorem}[Principal component angles for detectable spikes.]
    \label{thm:3:Angles}
        Let $1\le i \le r$ be a detectable spike, namely, such that $\sigma_i>\BBP$. Denote $y_i=\SpikeFunc(\sigma_i)$. Then, 
        \begin{itemize}
            \item (PC decoupling.) For all $1\le j \le r$, $j\ne i$, a.s., 
            \begin{equation}
                \langle \bu_i, \hbu_j\rangle \longrightarrow 0,\qquad  \langle \bv_i, \hbv_j\rangle \longrightarrow 0.
            \end{equation}
            \item (Limiting angles.) A.s.,
            \begin{equation}
                \left|\langle \bu_i, \hbu_j\rangle \right|\longrightarrow \m{U}_{\gamma,\beta}(\sigma_i),\qquad  \left|\langle \bv_i, \hbv_j\rangle \right|\longrightarrow \m{V}_{\gamma,\beta}(\sigma_i),
            \end{equation}
        	and
        	\begin{equation}
        		\langle \bu_i, \hbu_j\rangle\langle \bv_i, \hbv_j\rangle\longrightarrow \m{U}_{\gamma,\beta}(\sigma_i)\m{V}_{\gamma,\beta}(\sigma_i) \,,
        	\end{equation}
            where $\m{U}_{\gamma,\beta}(\cdot),\m{V}_{\gamma,\beta}(\cdot)$ are obtained by as follows.
            
            Take any $\bm{d}^{(i)} \in \ker\left( \bKc_{\gamma,\beta}(y_i) - (1/\sigma_i)\Hbb \right)\subset \RR^6$ such that
            \begin{equation}\label{eq:thm:3:1}
                \langle \bm{d}^{(i)}, \bTc_{\gamma,\beta}(y_i)\bm{d}^{(i)}\rangle = 1.
            \end{equation}
            (By Proposition~\ref{prop:PosEigFunc-properties}, there are exactly two such vectors, which are antipodal points.) 
            Then
            \begin{align}
                \VAngle(\sigma_i) &= \left| (\bd^{(i)})_1 \right| / \sigma_i \\
                \UAngle(\sigma_i) &= \left| (\bd^{(i)})_4 \right| / \sigma_i .
            \end{align}  
        %     Define 
        %     \begin{equation*}
        %         \bm{c}^{(i)} = \MatL 
        %         \underline{\bSigma}(\sigma_i)^{-1} &\0 \\
        %         \0 & \left(\underline{\bSigma}(\sigma_i)^{-1}\right)^\T 
        %         \MatR \bm{d}^{(i)}\qquad\textrm{where}\qquad
        %         \underline{\bSigma}(\sigma_i)^{-1} = \MatL
		% 1/\sigma_i & 0 & 0 \\
		% 0 &-1/\sigma_i	& \frac{\beta\sqrt{\gamma}}{1-\gamma\beta} \\
		% 0 &1-\gamma\beta 	&1/\sigma_i
		% \MatR .
        %     \end{equation*}  
        \end{itemize}
        % Then
        % \begin{equation}
        %     \m{U}_{\gamma,\beta}(\sigma_i)=(\bc^{(i)})_{1},\qquad \m{V}_{\gamma,\beta}(\sigma_i)=(\bm{c}^{(i)})_4 .
        % \end{equation}
    \end{theorem}

    While in principal computing closed-form formulas for $\UAngle(\cdot),\VAngle(\cdot)$ should be possible (the vectors $\bd^{(i)}$ are only $6$-dimensional), we have not been able to simplify the resulting expressions into a reasonably concise form, even with the aide of a computer algebra system.\footnote{In any case, we believe that such closed-form formulas would offer, from a practical point of view, only a small advantage over the current statement of Theorem~\ref{thm:2:Outliers}.} In Figure~\ref{fig:MainAngles} we show a plot of the product $\UAngle(\sigma)\VAngle(\sigma)$ for selected parameter combinations.
    
    \begin{remark}
    	To compute $\UAngle(\sigma),\VAngle(\sigma)$ in a numerically stable manner, one could follow these steps:
   		\begin{enumerate}
   			\item Compute exactly $y=\SpikeFunc(\sigma)$ using \eqref{eq:SpikeFunc}.
   			\item Find (up to machine precision) a non-zero solution $\bm{d}\in \RR^6$ satisfying $\left( \bKc_{\gamma,\beta}(y) - \sigma\bHc \right) \bm{d} = \0$. In practice, take $\bm{d}$ to be the eigenvector of $\bKc_{\gamma,\beta}(y) - \sigma\bHc \in \mathrm{Sym}^6(\RR)$ whose absolute value is smallest.
   		\item Finally, normalize $\bm{d}$ so that its $\bTc_{\gamma,\beta}(y)$-weighted norm is $1$:
   		\begin{align*}
   			\bar{\bm{d}}(y) \equiv \frac{1}{\sqrt{\langle \bm{d}, \bTc_{\gamma,\beta}(y_i)\bm{d}\rangle}}\bm{d}\,,
   		\end{align*}
   		and take $\VAngle(\sigma)=|(\bar{\bm{d}})|_1/\sigma$ and $\UAngle(\sigma)=|(\bar{\bm{d}})_4|/\sigma$.
   	\end{enumerate}
    \end{remark}
    
	\begin{figure}
	\centering
	\begin{subfigure}[b]{0.45\textwidth}
		\centering
		\includegraphics[width=\textwidth]{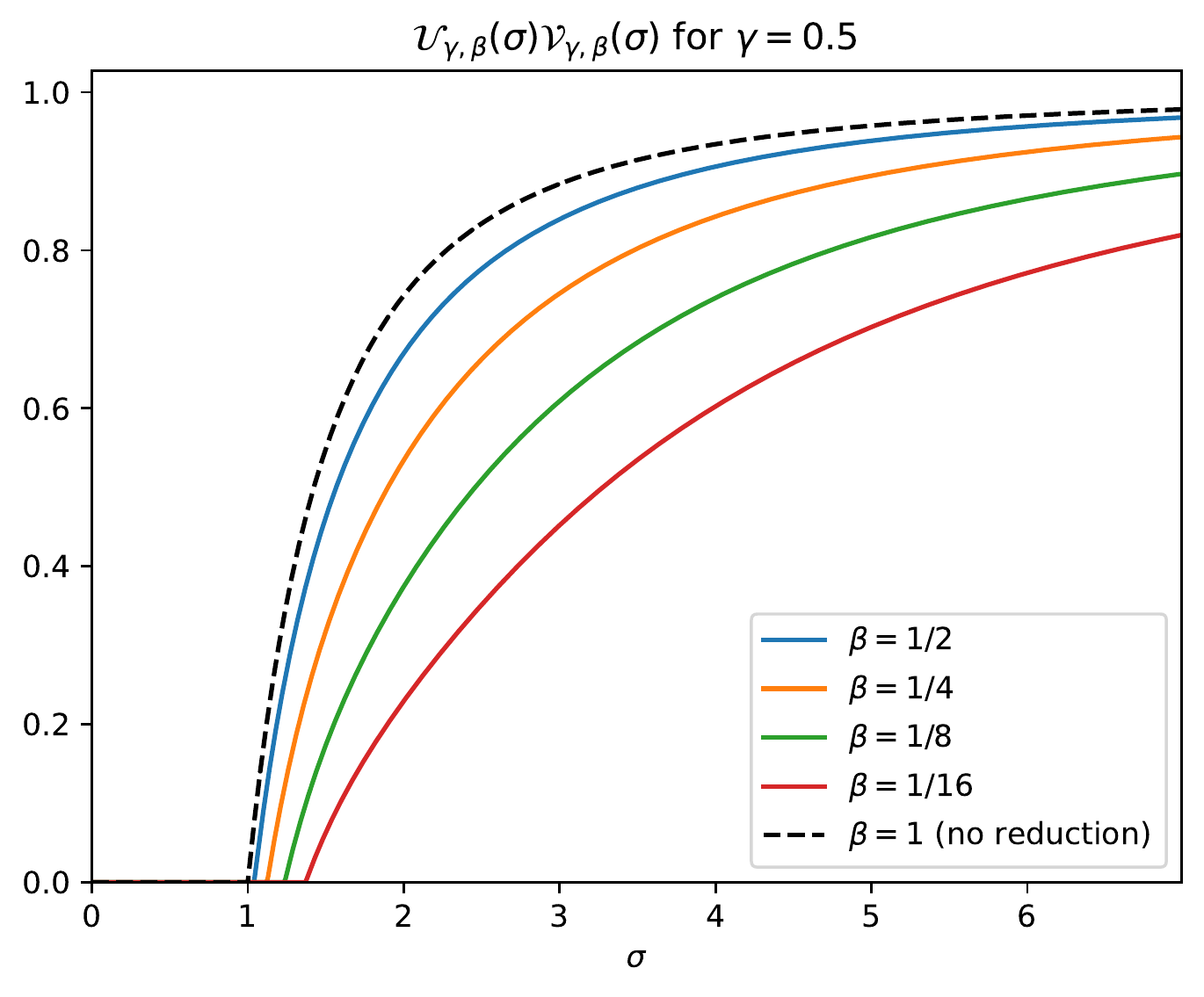}
	\end{subfigure}
	\hfill
	\begin{subfigure}[b]{0.45\textwidth}
		\centering
		\includegraphics[width=\textwidth]{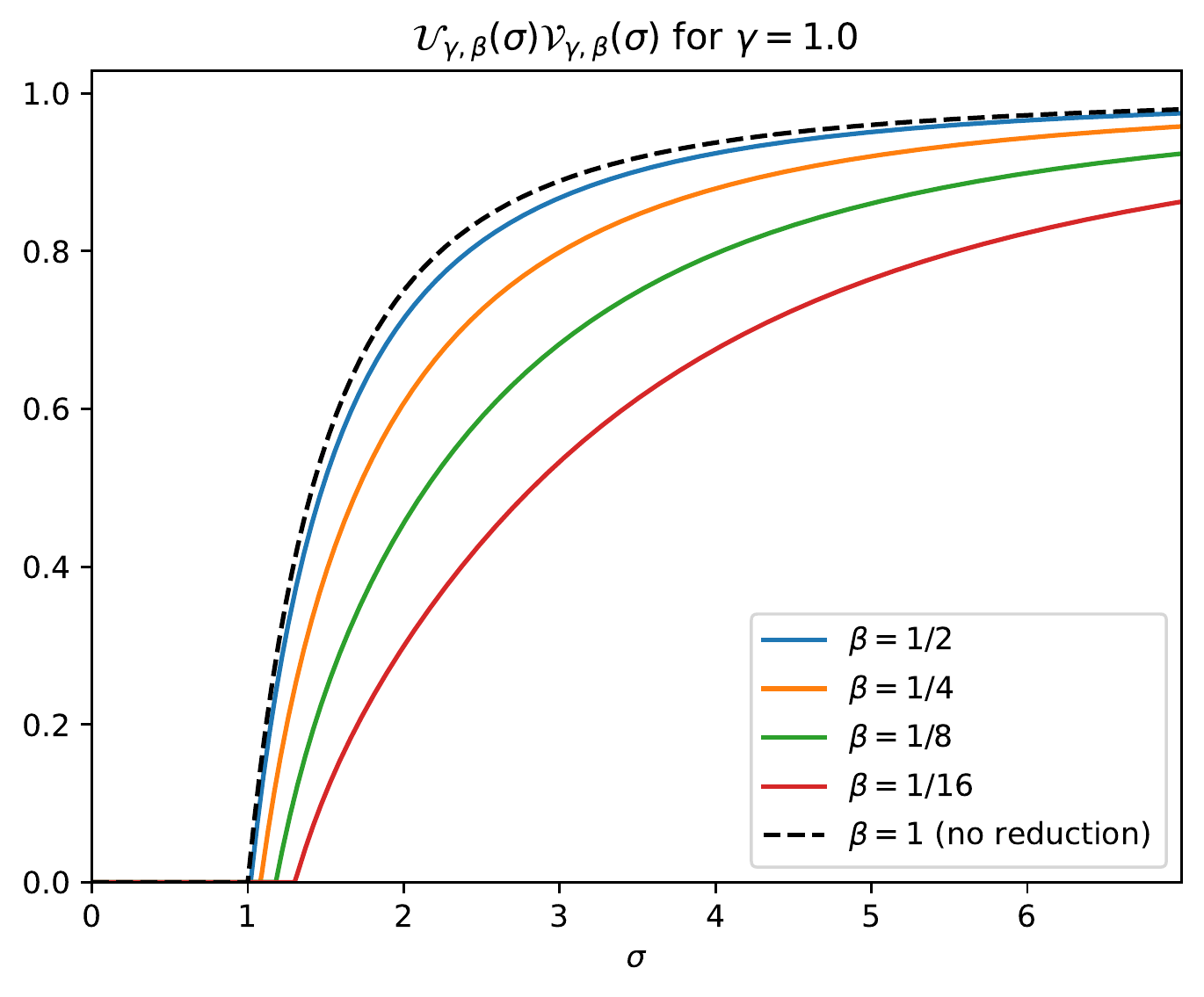}
	\end{subfigure}
	\caption{The singular vector angles $\UAngle(\sigma)\VAngle(\sigma)$ as a function of $\sigma$, for selected values of $\gamma,\beta$. When $\sigma<\BBP$, the plotted values are set to $0$. Dashed line: the limiting angles in the spiked model without dimension reduction, given in \eqref{eq:fact:vectors-u}-\eqref{eq:fact:vectors-v}.}
	\label{fig:MainAngles}
\end{figure}
    
    Theorem~\ref{thm:3:Angles} computes the limiting singular vector overlaps for super-critical (detectable) spikes; our computations do not apply for sub-critical spikes. 
    We show that as the spike intensity approaches the detectability threshold---in other words, when the spike is ``barely detectable''---the correlation between the corresponding population and empirical principal components vanishes:
    \begin{proposition}

    \label{prop:vanishing-corrs}
        We have 
        \begin{equation}
            \lim_{\sigma \downarrow \BBP} \m{U}_{\gamma,\beta}(\sigma) = 0,\qquad \lim_{\sigma \downarrow \BBP} \m{V}_{\gamma,\beta}(\sigma) = 0 . 
        \end{equation}
    \end{proposition}

    While we do not prove that principal components corresponding to non-detectable spikes are asymptotically de-correlated from their empirical counterparts, Proposition~\ref{prop:vanishing-corrs} leads us to conjecture:
    \begin{conjecture}[Detectability phase transition for singular vectors]
    \label{conj:Vectors}
        Suppose that signal spike $1\le i \le r$ is non-detectable, namely satisfies $\sigma_i\le \BBP$. Then for all $1\le j \le r$,
        \begin{equation}
            \langle \bu_i, \hbu_j\rangle \longrightarrow 0,\qquad  \langle \bv_i, \hbv_j\rangle \longrightarrow 0.
        \end{equation}
    \end{conjecture}
%    Our numerical experiments, reported in Section~\ref{?}, support the validity of Conjecture~\ref{conj:Vectors}. 

\Revision{
\paragraph*{Accompanying code.} \texttt{Python} code for computing all the aforementioned quantities is available: \url{https://github.com/eladromanov/Randomized-SVD-Code/}.

}

\section{Discussion}
\label{sec:Discussion}

% \subsection{The noise sensitivity of the R-SVD}
% \Subsection{\Revision{Behavior for small sketching dimension $\beta\ll 1$}}
% \label{sec:SmallBeta}

\Revision{Our main results, specifically
Theorems~\ref{thm:2:Outliers}-\ref{thm:3:Angles}, quantify in a very concrete sense the loss of signal incurred from undersampling in the R-SVD.
Note that from a practical standpoint, the R-SVD is most appealing when the dimensionality reduction is massive, that is $\beta\ll 1$. 
How fast does the ``effective'' SNR degrade as $\beta$ decreases? Stated differently, how {\it robust} is the R-SVD to noise at extreme undersampling ratios?}
\Revision{Theorems~\ref{thm:2:Outliers} and \ref{thm:3:Angles}, provide a concrete answer to this question, in several different senses.}  

At the bear minimum, if one wishes to get \emph{some} correlation between the signal and the PCs returned by the R-SVD, the SNR has to be, at the very least, $\sigma>\BBP$. Expanding \eqref{eq:BBP} for small $\beta\ll 1$ up to leading order,
\Revision{
\begin{equation}
	\label{eq:BBP-small-beta}
	\BBP = \left(1+ 1/\gamma\right)^{1/8}\beta^{-1/8} + \BigOh(\beta^{1/8})\qquad\textrm{as}\qquad\beta\to 0 \,.
\end{equation}
}
The scaling of \eqref{eq:BBP-small-beta} in $\beta$, namely $\beta^{-1/8}$, is consistent with the commonly held belief among practitioners that the R-SVD is quite robust to measurement noise \cite{halko2011finding}: even for very small $\beta$, $\beta^{-1/8}$ is reasonably moderate (e.g., if $\beta=1/1000$---representing a thousand-fold dimension reduction---then $\beta^{-1/8}=2.37$).

 Note that this type of fine-grained information cannot be directly obtained from the operator norm approximation bound \eqref{eq:Halko-Bound}.
 For example, applied to $\bY=\sigma\bu\bv^\T + (nm)^{-1/4}\bZ$ a noisy rank $1$ signal, \eqref{eq:Halko-Bound} would ensure (via singular vector perturbation bounds, e.g. Davis-Kahan) that the top principal component of $\hbY$ has some correlation with the signal only once $\|\bY-\hbY\|\lesssim \sigma$. When $n\asymp m$ and $\beta=d/m$ is small, the r.h.s. of \eqref{eq:Halko-Bound} scales like $\sim \beta^{-1/2}$; that is, we can deduce that the large PCs of $\hbY$ have \emph{some} correlation with $\bX$ only once $\sigma\gtrsim \beta^{-1/2}$. For small $\beta$, this can be considerably larger than $\beta^{-1/8}$.

In practice, however, one should be careful interpreting the detection threshold \eqref{eq:BBP-small-beta}.
\Revision{
    When the SNR is $\sigma>\BBP$, an outlier separates from the bulk; but how far from the edge would it be? 
    Note that the bulk edge \eqref{eq:edges-def}, is 
        \begin{equation}
        \sqrt{\NoiseEdgeUpper} = 
        \sqrt{\gamma^{1/2}+\gamma^{-1/2}}
        +
        \BigOh(\beta^{1/2}) \qquad\textrm{as}\quad\beta\to 0\,.
    \end{equation}
    % \begin{equation}
    %     \sqrt{\NoiseEdgeUpper} = 
    %     \sqrt{\gamma^{1/2}+\gamma^{-1/2}}
    %     +
    %     \gamma^{1/4}
    %         \beta^{1/2}
    %     +
    %     o(\beta^{1/2})
    % \end{equation}
    % to leading order in $\beta\to 0$. 
    We expand the spike-forward map \eqref{eq:SpikeFunc} to leading order as $\beta\to 0,\sigma\to \infty$, in the regime where $\sigma\gg \BBP\sim (1+1/\gamma)^{1/2}\beta^{-1/8}$ (that is, $\sigma$ is asymptotically larger than the detection threshold):
        \begin{align}
        \SpikeFunc(\sigma) 
        &= 
        \underbrace{    
            \sqrt{
        \frac{\left(\sqrt{\gamma }+\sigma ^2\right) \left(1+\sqrt{\gamma } \sigma ^2\right)}{\sqrt{\gamma } \sigma ^2 \left(1+\gamma
        -\beta  \gamma +\sqrt{\gamma } \sigma ^2\right)}
            } 
        }_{\gamma^{-1/4}+\BigOh(1/\sigma^2) = \gamma^{-1/4}+o(\beta^{1/2})}
        \cdot  
        \sqrt{
            \frac{\left(1+\gamma -\beta  \gamma +2 \beta  \sqrt{\gamma } (1+\gamma -\beta  \gamma ) \sigma ^2+\beta  \gamma  \sigma ^4\right)}{ \left(1+\beta  \sqrt{\gamma } \sigma ^2\right)} 
        }\,.
        \label{eq:OutlierSmallBeta:1}
    \end{align}
    Evidently, the asymptotic behavior of the above undergoes a transition at  scale $\sigma=\beta^{-1/4}$.  

    In the weak (but detectable) SNR case, $\BBP<\sigma\ll \beta^{-1/4}$, \eqref{eq:OutlierSmallBeta:1} expands as 
    \begin{equation}
        \SpikeFunc(\sigma)
        = 
        \sqrt{\gamma^{1/2}+\gamma^{-1/2}} + \frac12\frac{{\gamma^{3/4}}}{\sqrt{1+\gamma}}(\beta\sigma^4) + o(\beta\sigma^4)\,.
    \end{equation}
    The distance to the bulk is thus (to leading order),
    $\SpikeFunc(\sigma)-\sqrt{\NoiseEdgeUpper} = \Theta(\beta\sigma^4)$. Since $\beta^{-1/8}\ll \sigma \ll \beta^{-1/4}$, this is $o(1)$: the outlier is very close to the edge. Accordingly, one suspects that in practical settings (finite, reasonably moderate $n$), detecting the signal in this regime (by thresholding the statistic $\sigma_1(\hbY)-\sqrt{\NoiseEdgeUpper}$) could yield rather disappointing results. A quantification of this statement is beyond the means of our current results.

    In constrast, at the scale $\sigma\propto \beta^{-1/4}$, the outlier is at a constant distance from the bulk. When $\sigma\gg \beta^{-1/4}$, it escapes away from it as $\beta \to 0$: when $\sigma\lesssim \beta^{-1/2}$, it is located at $y\sim \gamma^{1/4}\beta^{1/2}\sigma^2$; when $\sigma\gtrsim \beta^{-1/2}$, $y\sim \sigma$.  One accordingly expects that even in relatively low-dimensional settings, detecting the signal in these regimes would be rather easy.

    % It is evident that $\SpikeFunc(\sigma) \to \infty$ only when 

}

\Revision{So far we have discussed signal detection; let us consider estimation.}
$\mathrm{SNR}_{\gamma,\beta}(0)\equiv \BBP$ is the SNR level required to retain asymptotically \emph{any} correlation between the signal and the observed PCs; \Revision{this is a rather weak notion of signal estimation.}  One could consider a stronger metric, for example $\mathrm{SNR}_{\gamma,\beta}(0.5)$, the SNR level required so that a signal dyad $\bu_i\bv_i^\T$ and an observed dyad $\hbu_i\hbv_i^\T$ have asymptotic correlation at least $\ge 0.5$ (the constant $0.5$ chosen somewhat arbitrarily). That is, we are interested in the solution $\sigma = \mathrm{SNR}_{\gamma,\beta}(0.5)$ of
\begin{equation}
	\UAngle\left( \sigma \right) \VAngle\left( \sigma \right) = 0.5,\qquad \sigma>\BBP \,.
\end{equation}
In Figure~\ref{fig:SmallBeta} we plot the functions $\beta\mapsto \mathrm{SNR}_{\gamma=1,\beta}(0),\mathrm{SNR}_{\gamma=1,\beta}(0.5)$ for small $\beta$ in a log-log scale. 

\Revision{
Interestingly, while $\mathrm{SNR}_{\gamma=1,\beta}(0)\sim \beta^{-1/8}$, and while $\sigma\sim \beta^{-1/4}$ is the scale above which a \emph{noticable} outlier escapes the bulk, we find that $\mathrm{SNR}_{\gamma=1,\beta}(0.5)\sim \beta^{-1/2}$. This suggest that {for practical purposes},  the behavior implied by the ``coarse'' operator norm bound \eqref{eq:Halko-Bound} \emph{is} actually representative of the true noise sensitivity of the R-SVD.
In particular,  at small undersampling ratios $\beta$, we find that there is a fundamental discrepancy between the notions of signal {detection} (either with ``weak'' outliers, at $\BBP\sim \beta^{-1/8}$, or ``strong'' outliers when $\sigma\sim \beta^{-1/4}$) and that of signal {estimation} ($\mathrm{SNR}_{\gamma,\beta}(0.5)\sim \beta^{-1/2}$).
}
\begin{figure}
	\centering
	\includegraphics[width=0.5\textwidth]{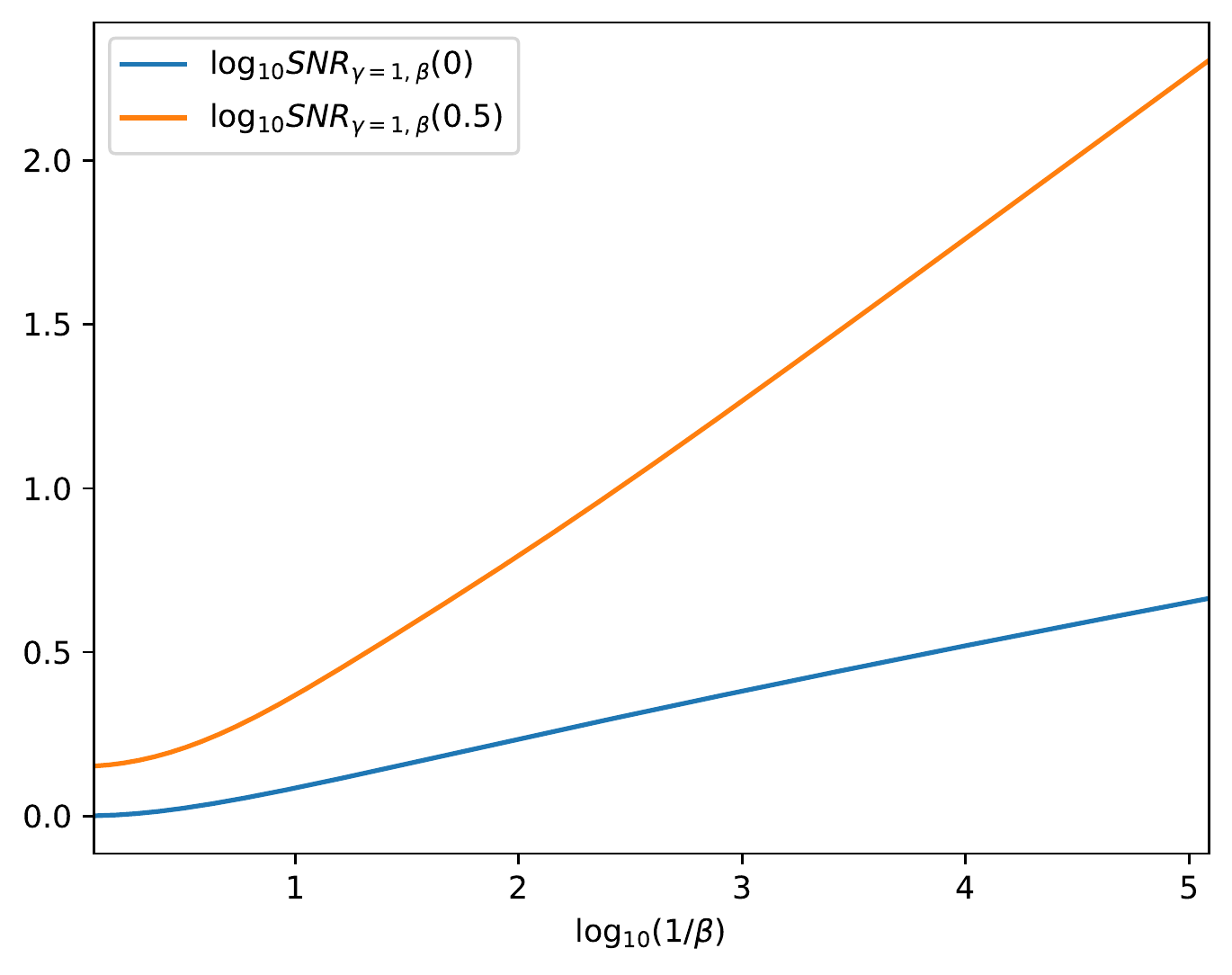}
	\caption{ The SNR level required for 1) signal detection ($\mathrm{SNR}_{\gamma=1,\beta}(0)=\BBP$); 2) for the correlation between the signal PCs and the empirical PCs of $\hbY$ to be at least 0.5 ($\mathrm{SNR}_{\gamma=1,\beta}(0.5)$). Plotted in log-log scale, for $\gamma=1$ and $\beta=d/m$ small. While detection occurs at low SNR, $\mathrm{SNR}_{\gamma=1,\beta}(0)\asymp \beta^{-1/8}$, for signal intensities close to the threshold the overlap between the $\bX$ and $\bY$ is very small ($o(1)$ as $\beta\to0$). To get a constant overlap (in our case $0.5$), one needs to work at considerably higher SNRs, $\mathrm{SNR}_{\gamma=1,\beta}(0)\asymp \beta^{-1/2}$.  }
	\label{fig:SmallBeta}
\end{figure}

\Revision{
    There is yet another interesting side to this story. Above, we discussed estimation of the entire dyad $\bu_i\bv_i^\T$; how does the error separate between $\bu_i,\bv_i$? Note that the R-SVD algorithm does not operate symmetrically on the left and right sides of $\bY$. Denote $\mathrm{SNR}^u_{\gamma=1,\beta}(0.5),\mathrm{SNR}^v_{\gamma=1,\beta}(0.5)$ the SNR level requires to achieve $\langle \bu,\hbu\rangle^2=0.5,\langle\bv,\hbv\rangle^2=0.5$ respectively. Figure~\ref{fig:SNRs-U-V}
    plots these quantities as a function of $\beta$ (for $\gamma=1$). We find that while $\mathrm{SNR}_{\gamma,\beta}^u(0.5)\sim \beta^{-1/2}$, we have $\mathrm{SNR}_{\gamma,\beta}^v(0.5)\sim \beta^{-1/4}$, which is much smaller. That is, the R-SVD provides considerably better estimates of the right singular vectors of $\bY$ than the left ones, the estimate being ``good'' essentially at the moment where the outlier ``noticeably'' separates from the bulk.
}

\begin{figure}
    \centering
    \includegraphics[width=0.5\textwidth]{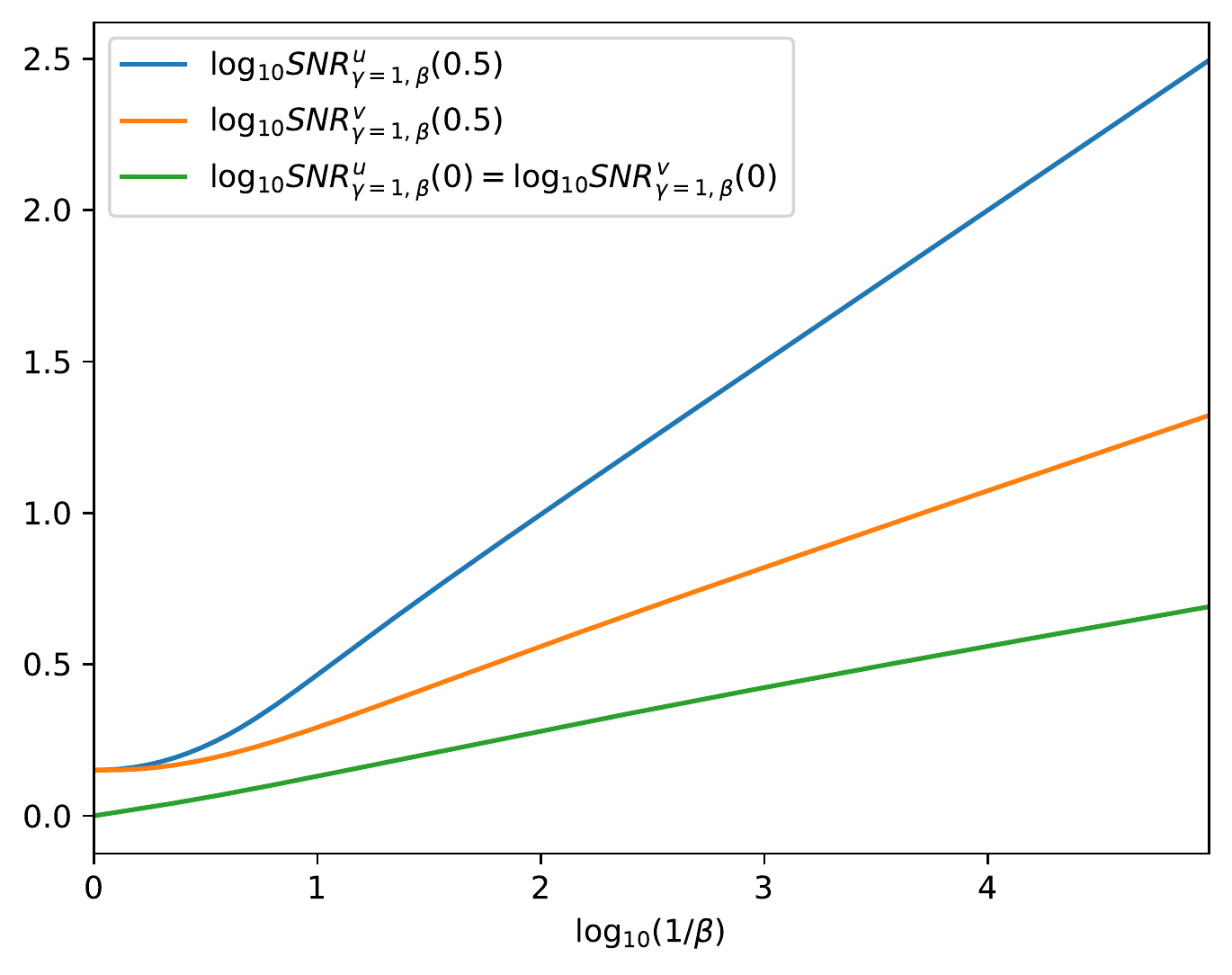}
\caption{
    \Revision{
$\mathrm{SNR}_{\gamma=1,\beta}^u(0.5),\mathrm{SNR}_{\gamma=1,\beta}^v(0.5)$, the SNR level required to attain, respectively, $\langle\bu,\hbu\rangle^2=0.5,\langle \bv,\hbv\rangle^2=0.5$. For $\beta\ll 1$, we find that $\mathrm{SNR}_{\gamma=1,\beta}^u(0.5)\sim \beta^{-1/2}$ while $\mathrm{SNR}_{\gamma=1,\beta}^v(0.5)\sim \beta^{-1/4}$; in particular, for small $\beta$, the R-SVD can estimate the right signal singular vectors at considerably smaller SNR. 
    }
}
\label{fig:SNRs-U-V}
\end{figure}

\Revision{
    The above findings suggest a nontrivial $\beta\to 0,\sigma\to\infty$ scaling limit for the left and right singular vector overlaps, $|\langle \bu,\hbu\rangle|,|\langle \bv,\hbv\rangle|$ respectively. To wit, define
    \begin{align}
        \m{U}_\gamma^\star(\alpha) &= \lim_{\beta\to 0}\UAngle(\beta^{-1/2}\alpha) \,,\\
        \m{V}_\gamma^\star(\alpha) &= \lim_{\beta\to 0} \VAngle(\beta^{-1/4}\alpha) \,.
    \end{align}
    We numerically compute\footnote{
        It would be interesting to analytically compute the corresponding limiting expressions, starting from the formulas in Theorems~\ref{thm:2:Outliers}-\ref{thm:3:Angles}. However, obtaining an analytically tractable small $\beta$ expansion via Theorem~\ref{thm:3:Angles} appears to be challenging. 
        } and plot $\m{U}_\gamma^\star,\m{V}_\gamma^\star$ for selected values of $\gamma$; see Figure~\ref{fig:U-V-Star}.
        \begin{figure}
            \centering
            \includegraphics[width=0.45\textwidth]{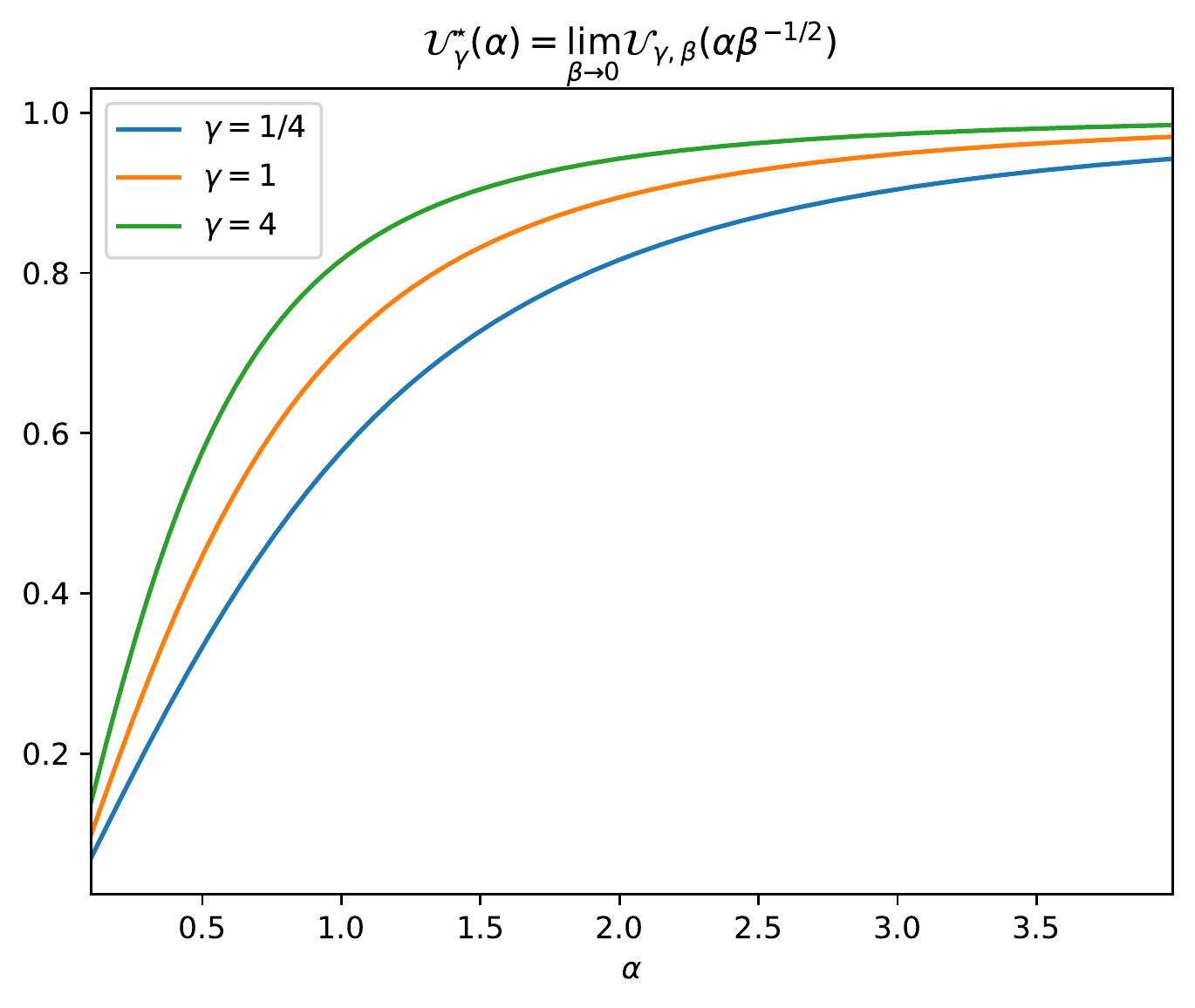}
            \includegraphics[width=0.45\textwidth]{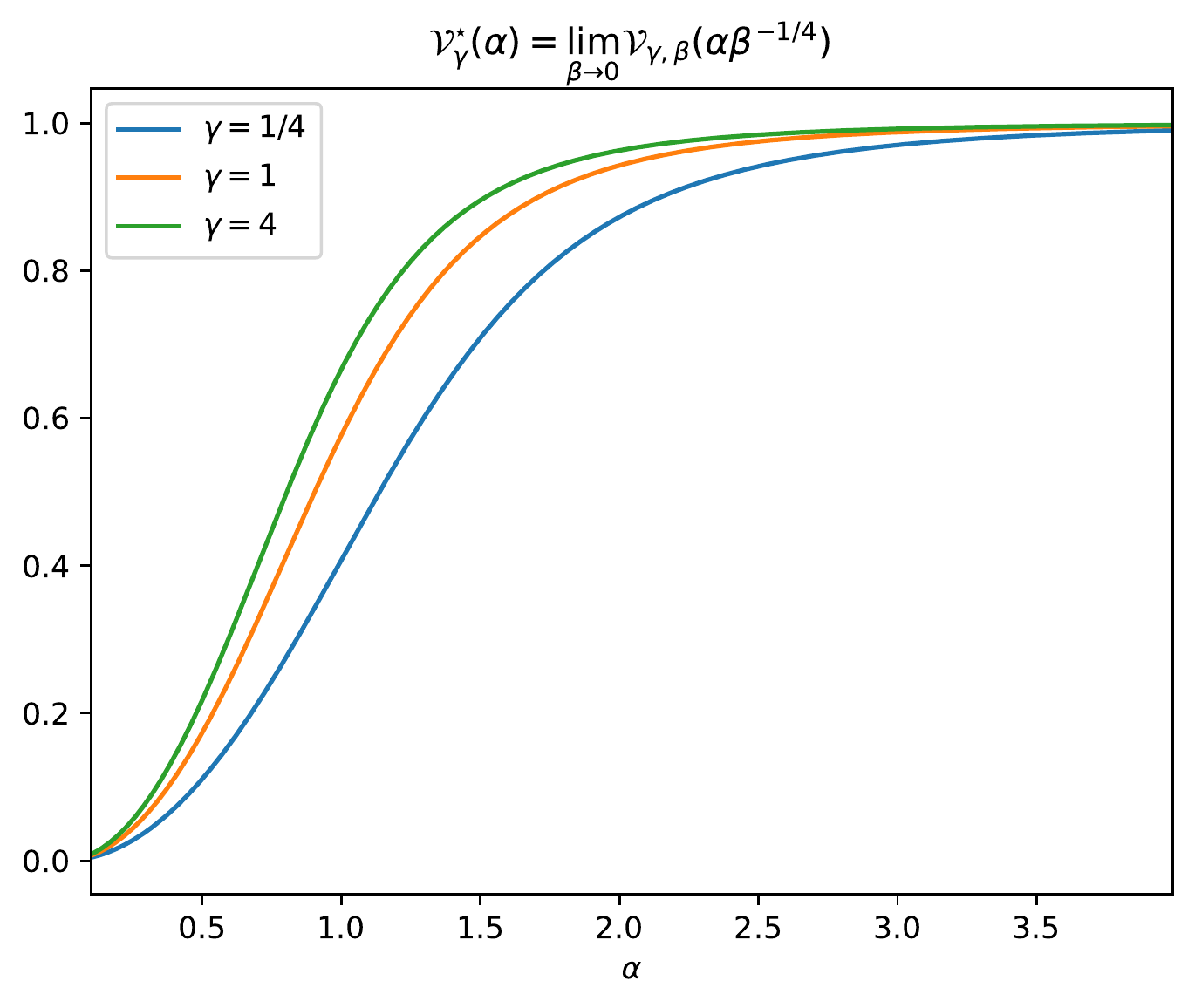}
            \caption{\Revision{Singular vector correlations in the joint limit $\beta\to 0$ and $\sigma\to 0$. To obtain nontrivial limits, we scale $\sigma=\alpha \beta^{-c}$ for $c\in \{1/4,1/2\}$. For the left singular vectors: $\m{U}_\gamma^\star(\alpha) = \lim_{\beta\to 0}\UAngle(\beta^{-1/2}\alpha)$ (left figure); for the right singular vectors: $\m{V}_\gamma^\star(\alpha) = \lim_{\beta\to 0} \VAngle(\beta^{-1/4}\alpha)$ (right figure). In particular, the R-SVD is able to estimate the right singular vectors adequately at a much lower SNR ($\propto \beta^{-1/4}$) than the left ones ($\propto \beta^{-1/2}$).}}
            \label{fig:U-V-Star}
        \end{figure}
    
}

\Revision{
    Mathematically, our discussion of extreme undersampling ratios corresponds to a double limit: first we take the large-dimensional limit ($n,m,d\to\infty$), and only then the limit $\beta\to 0$. We have further found that to get non-degenerate limits as $\beta\to 0$, one needs to increase $\sigma$ with $\beta$ as either $\sigma=\beta^{-1/4}$ or $\sigma=\beta^{-1/2}$. These findings call for a more refined analysis of the R-SVD under {\it disproportionate growth} asymptotics,
    \[
    m_n = \gamma n,\qquad \beta_n=d_n/m_n,\qquad \sigma=\beta_n^{-c},\qquad c\in \left\{\frac{1}{4},\frac{1}{2}   \right\}
    \]
    where $n\to\infty$ and $\beta_n\to 0$ simultaneously (perhaps with constraints on the decay rate of $\beta_n$). For the spiked model (and the full SVD), such analyses have appeared in the literature only recently, for example \cite{feldman2023spiked,donoho2022optimal}.

}

\Revision{
Lastly, we remark that introducing power iterations is known to improve the performance of the R-SVD considerably \cite{halko2011algorithm}. If one uses $q$ power iterations, \cite[Theorem 1.2]{halko2011finding} yields an operator norm bound, akin to \eqref{eq:Halko-Bound}, such that $\|\bY-\hbY\|/\sigma \sim \beta^{-\frac{1}{2(1+2q)}}/\sigma$ as $\beta\to 0$. 
Accordingly, in the language of this section, $\mathrm{SNR}^{(q)}(0.5)\lesssim \beta^{-\frac{1}{2(1+2q)}}$.  
A fine-grained analysis  for the R-SVD with power iterations (e.g., threshold for emergence of outlier) is an interesting direction for future research.
% \footnote{A numerical study by means of Monte-Carlo experiments seems challenging. Recall that $d=\beta m$ has to be reasonably large itself; since we are interested in $\beta\ll 1$, this means the dimensions $n,m$ have to be very large to begin with.}
}

\subsection{Pitfalls for PCA and data analysis}
\label{sec:StoryPCA}

\Revision{The findings described in the previous section raise some concerns regarding the use of the R-SVD as a tool for exploratory data analysis of large-dimensional data sets. }

% As the next section deals with PCA, it is natural to normalize the data matrix $\bY$ slightly differently; to avoid confusion, we use new notation.
% Consider a data set consisting of $m$ points of dimension $p$, denoted $\bm{t}_1,\ldots\bm{t}_m \in \RR^p$. Define the data matrix,
% \begin{equation}
%     \bm{T} = m^{-1/2}\MatL \bm{t}_1 &\ldots \bm{t}_m\MatR \in \RR^{p\times m},
% \end{equation}
% normalized so that $\bm{S}=\bm{T}\bm{T}^\T$ is the sample covariance. Denote by $\alpha=p/m$ the ratio between the dimension and the number of samples.

\Revision{
Principal component analysis (PCA) is an important tool for exploratory data analysis \cite{jolliffe2002principal}. 
Given a data matrix $\bY$, whose e.g. rows correspond to samples from some data set, the practitioner applies the SVD and 
retains the largest PCs of $\bY$.
When there are few ``emerging'' singular values---that are clearly larger than the remaining ones---it is implied that the corresponding principal components 
represent a low-dimension latent structure in the data.
In constrast, small singular values---which are often clumped in a ``bulk'' structure---are thought of as noise.
Principal component selection is a key methodological question: how large should a singular be so that its eigenvectors be considered ``informative'', as opposed to ``noise''?
In high-dimensional settings---where it is tacitly expected that most of the observed singular values correspond to ambient noise---it is essentially standard practice to select large PCs  which are well-separated (outliers) from the bulk of smaller singular values.
Among practitioners, perhaps the most well-known formalization of this practice is Cattell's ``scree test'' (also known as the ``elbow method''), proposed in 1966 \cite{cattell1966scree}; since then, more ``quantitative'' procedures were proposed as well, see for example \cite{gavish2014optimal,dobriban2019deterministic,ke2021estimation,
donoho2020screenot} among others. 

We envision the following type of scenario. Owing to the sheer magnitude of their dataset---so that taking the full SVD of the data matrix is infeasible---the practitioner naively runs an R-SVD procedure, with the aim of conducting an aforementioned kind of exploratory analysis. 
(For example, they could use an off-the-shelf implementation from a standardized software package, such as \texttt{scikit-learn} \cite{pedregosa2011scikit}.)
Could doing so have any adverse consequences on their findings? 
}

\Revision{A natural ``toy model'' to study such question is that of a low-rank factor model \cite{tipping1999probabilistic}. For simplicity, consider a rank-$1$ factor model, in which one observes $n$ data points
$\bm{y}_i = \eta u_i \bv+\bm{z}_i \in \RR^m$, where $\bm{v}\in \RR^m,\|\bm{v}\|=1$ is an unknown direction (the factor), $u_i\sim \m{N}(0,1)$ is a factor loading, $\eta^2$ is its variance, and $\bm{z}_i\sim \m{N}(0,\bI)$ is isotropic ambient noise. 
Let $\bY$ be the matrix whose rows are $\bm{y}_1,\ldots,\bm{y}_n$ divided by $(nm)^{-1/4}$ (similarly $\bm{Z}$), and denote
\[
\gamma=m/n,\quad \bu=n^{-1/2}(u_1,\ldots,u_n),\quad \sigma=\gamma^{-1/4}\eta,\quad \bY=\sigma\bu\bv^\T + \bZ\,.    
\]
Clearly, $\bY$ is an instance of a spiked random matrix, as described in Section~\ref{sec:SpikeModelBackground}.}

% Denote $\bm{\xi}=m^{-1/2}(\xi_1,\ldots,\xi_m)$, so that $\bm{T} = \eta \bm{u}\bm{\xi}^\T + \bm{W}$. Clearly, $\|\bm{\xi}\|\simeq 1$ with overwhelming probability, and therefore $\bm{Y}=(m/p)^{1/4}\bm{T}=\alpha^{-1/4}\bm{T}$ 
% follows the model of Section~\ref{sec:Setup} with $n=p$, $\gamma=1/\alpha$ and $\sigma=\alpha^{-1/4}\eta$. 
% Appealing to the results cited in Section~\ref{sec:SpikeModelBackground}, the detection threshold becomes $\eta^*=\alpha^{1/4}$, the covariance noise bulk edge is $\lambda_2(\bS)\simeq y_0(\alpha) \equiv (1+\sqrt{\alpha})^2$, and for a super-critical spike, $\eta>\eta^*$,
% \begin{equation}\label{eq:StoryPCA-1}
%     \lambda_1(\bS)\simeq   \frac{\left(\eta^2 + 1\right)  \left(\eta^2 + \alpha\right)}{\eta^2},
%     \qquad
%     \langle \bm{u},\hat{\bm{u}}_1\rangle^2 \simeq \frac{\eta^4-\alpha}{\eta^4+\alpha\eta^2}.
% \end{equation}

\Revision{
How large should an outlier be so that the corresponding observable leading singular vectors $\hbu\hbv^\T$ are well-aligned with their signal counterparts $\bu\bv^\T$?
Consider, for example, the ratio between the leading and subleading singular values (the latter is the noise bulk edge), $\alpha=\hat{\sigma}_1/\hat{\sigma}_2$. Using the formulas cited in Section~\ref{sec:SpikeModelBackground}, once can show that in the large dimensional limit (e.g. \cite{gavish2014optimal})
\begin{align*}
    \hat{\sigma}^2_1 
    &= \sigma^2 + \frac{1}{\sigma^2} + \rho,\qquad \hat{\sigma}_2^2 = 2 + \rho,\qquad 
    \rho = \gamma^{1/2}+\gamma^{-1/2},\\
    \sigma^2 &= \frac{\hat{\sigma}_1^2-\rho + \sqrt{(\hat{\sigma}_1^2-\rho)^2-4}}{2}, \qquad \Delta:= \hat{\sigma}_1^2-\hat{\sigma}_2^2 \\
    \langle \bu,\hbu\rangle\langle \bv,\hbv\rangle &= \frac{\sigma^4-1}{\sigma^3 \hat{\sigma}_1} = \frac{1}{\sigma\hat{\sigma}_1}\sqrt{(\hat{\sigma}_1^2-\rho)^2-4} = \frac{1}{\sigma\hat{\sigma}_1}\sqrt{\Delta(\hat{\sigma}_1^2-\rho+2)}\,.
\end{align*}
Using, for example, $\sigma\le \hat{\sigma}_1$ and $\hat{\sigma}_1^2-\rho+2\ge \hat{\sigma}_1^2-\rho-2=\Delta$, yields $\langle \bu,\hbu\rangle\langle \bv,\hbv\rangle \ge \Delta/\hat{\sigma}_1^2=1-1/\alpha^2$. 
From this simple (and rather loose) calculation we can deduce: if one observes an outlier, such that the ratio $\alpha$ between its location and the bulk edge is reasonably large (e.g. $\alpha \ge \sqrt{2}$), then the corresponding observed singular vectors are necessarily reasonably correlated with the ground truth (e.g. $\langle\bu,\hbu\rangle\langle \bv,\hbv\rangle\ge 0.5$). (We emphasize that this rule is not advocated as a reciple for PC selection---it is highly suboptimal as such. See, for example, \cite{gavish2014optimal} for a treatment of optimal singular value thresholding when the noise is i.i.d. and \cite{donoho2020screenot} for correlated noise.)
}

\Revision{
The above intuition is key to the practice of PC selection: if an outlier is large relative to the bulk edge, it should be retained. 
The key message 
of this section is that this intuitive understanding of the spectrum of $\bY$ is entirely \emph{incorrect} as far as the spectrum of the reduced matrix is concerned and when $\beta\ll 1$. 
In the SNR regime $\beta^{-1/4} \le \sigma \ll \beta^{-1/2}$, one may observe very strong outliers in the spectrum of $\bY$ ($\hat{\sigma}_1\sim \beta^{1/2}\sigma^2$---up to magnitude $\lesssim \beta^{-1/2}$!); their corresponding singular vectors, however, would only be very weakly correlated with the signal ($\langle \bu,\hbu\rangle\langle \bv,\hbv\rangle=o(1)$ as $\beta\to 0$).
}

\Revision{
Lastly, we remark that the \Revision{introduction} of power iterations largely mitigates the phenomenon described above; though the complete quantification is beyond the scope of this paper. Heuristically, we can consider the bound of \cite{halko2011finding}: $\|\bY-\hbY\|/\sigma\sim \beta^{-\frac{1}{2(1+2q)}}$, when $q$ power iterations are used. Hence, $\sigma \sim \beta^{-1/2(1+2q)}$ suffices for constant estimation error, with an outlier appearing within distance at most $\hat{\sigma}_1 \le (1+\beta^{-1/2(1+2q)})\sigma$. While $\hat{\sigma}_1/\sigma$ indeed blows up as $\beta\to 0$, it does so very slowly: for example, if $\beta=10^{-3}$, and $q=5$, then $\beta^{-1/2(1+2q)}=10^{3/22}\approx 1.37$. 
}

\subsection{Sketched PCA}
\label{sec:SketchedPCA}

\Revision{
In a recent paper \cite{yang2021reduce}, the authors studied the asymptotics of {sketched PCA} (S-PCA) under the spiked model. 
In S-PCA one is given a data matrix $\bY \in \RR^{n\times m}$ (whose rows, typically, are assumed to be i.i.d. vector samples), and is interested in the {right} singular vectors of $\bY$. Under a signal-plus-noise model $\bY=\bX+\bZ$, the end goal is to estimate the right singular vectors of $\bX$, which span the latent low-dimensional subspace in which the noiseless data points reside. For a random sketching matrix $\bOmega'\in \RR^{d\times n}$, one forms $\hbY=\bOmega'\bY$, and then uses its leading right singular vectors ${\hbv},\ldots,{\hbv}_r$ as proxies for the true principal directions $\bv_1,\ldots,\bv_r$. That is, 
rather than producing a projection $\bPc=\bQ\bQ^\T$ using randomly-sketched data (recall: $\bQ\bR=\tilde{\bY}=\bY\bOmega$ is obtained by e.g. a QR decomposition), sketched PCA simply projects $\bY$ from the left in a data-independent manner.
% S-PCA omits the second step of the R-SVD, where one projects $\bY$ from the left onto the column space of $\tilde{\bY}$.

It is well-known that the R-SVD generally produces better estimates for the data PCs than the simpler S-PCA \cite{halko2011finding}. 
%    Under the spiked model, our results combined with \cite{yang2021reduce} provide a precise quantification of this fact. 
The authors of \cite{yang2021reduce} showed that in the spiked model, under the setup considered in the present paper, the spectrum of the sketched matrix $\tilde{\bY}$ exhibits a similar phenomenology to that of Theorems~\ref{thm:1:Bulk}-\ref{thm:3:Angles}, namely, 1) the singular values are arranged in a bulk and outliers structure, the outliers being in a 1-to-1 correspondence with the signal PCs; and 2) the angles between signal and empirical PCs concentrate around a deterministic quantity. They further derive formulas for the limiting outlier location and singular vector correlations under several choices of the sketching matrix $\bOmega'$.

Consider the simplest case, where $\bOmega'\in \RR^{d\times n}$ is a uniformly random (Haar) projection, and assume the setup of Section~\ref{sec:Setup}, where for simplicity $r=1$.\footnote{Note that we use a different normalization for the noise variance from \cite{yang2021reduce}, and therefore the formulas given below are slightly different from \cite[Theorem III.1]{yang2021reduce}.}
Let $\tilde{\bu}=\bOmega'\bu/\|\bOmega'\bu\|$, $\tilde{\bZ}=\bOmega'\bZ$, and note that 1) $\|\bOmega'\bu\|\simeq \sqrt{d/n}\simeq (\beta\gamma)^{1/2}$; and 2) $\tilde{\bZ}\in \RR^{d\times m}$ has i.i.d. entries $\tilde{\bZ}_{i,j}\sim \m{N}(0,1/\sqrt{nm})$. Then
\begin{equation}
    \hbY=\bOmega'(\sigma\bu\bv^\T +\bZ) \simeq \sigma (\beta\gamma)^{1/2}\tilde{\bu}\bv^\T + \tilde{\bZ}\,.
\end{equation}
The distribution of the bulk, that is, the LSD of $\tilde{\bZ}\tilde{\bZ}^\T$, is a Marchenko-Patur law with shape $d/m=\beta$ and scale $\frac{1/(nm)^{1/4}}{1/(m)^{1/2}}=\gamma^{1/4}$. In particular, the bulk edge is $\|\tilde{\bZ}\|^2\simeq \gamma^{1/2}(1+\sqrt{\beta})^2$. 

By the results cited in Section~\ref{sec:SpikeModelBackground}, an outlier separates from the bulk when $\sigma(\beta\gamma)^{1/2}>\gamma^{1/4}\beta^{1/4}$, equivalently $\sigma> (\beta\gamma)^{-1/4}$. That is, $\sigma^*=(\beta\gamma)^{-1/4}$ is the detection threshold for an outlier. This should be compared with the corresponding detection threshold for the R-SVD, which for small $\beta$ scales like $\BBP=(1+1/\gamma)^{1/8}\beta^{-1/8}+o(\beta^{-1/8})$, and is considerably smaller. 

Above the detection threshold ($\sigma>\sigma^*$), the limiting outlier position is
\begin{equation}
    \sigma_1(\hbY) \simeq \sqrt{
        \sqrt{\gamma}(1+\sqrt{\beta})^2 + \frac{1}{\sigma^2}+\beta\gamma\sigma^2 - 2\sqrt{\beta\gamma}    
    } \,,\qquad \sigma>(\gamma\beta)^{-1/4}\,.
\end{equation}
In particular, for $\beta\ll 1$ small, SNRs of scale $\sigma^*=(\gamma\beta)^{-1/4}<\sigma \ll \beta^{-1/2}$ generate outliers that are very close to the bulk (in the sense that $\sigma_1(\hbY)-\sigma_2(\hbY)=o(1)$ as $\beta\to 0$), and may be hard to detect in finite-$n$ settings. In contrast, for the R-SVD, $\sigma\gtrsim \beta^{-1/4}$ already generates outliers that are a constant distance away from the bulk.

Furthermore, the limiting singular vectors correlations are 
\begin{align}
	\label{eq:DobribanFormula}
    |\langle \bv_i, \tilde{\bv}_i\rangle | 
    &\longrightarrow \begin{cases}
		\sqrt{
            \frac{\gamma\beta \sigma_i^4 - 1}{\gamma\beta\sigma_i^4 + \gamma^{1/2}\sigma_i^2}
            } \quad&\textrm{if}\quad \sigma_i > (\beta\gamma)^{-1/4},\\
		0 \quad&\textrm{if}\quad \sigma_i \le (\beta\gamma)^{-1/4}
	\end{cases} \,.
\end{align}
Similar to the discussion of Section~\ref{sec:Discussion}, let $\sigma = \mathrm{SNR}^{\mathrm{S-PCA}}_{\gamma,\beta}(0.5)$ be such that $|\langle \bv_i, \tilde{\bv}_i\rangle |^2=0.5$ asymptotically. We can compute exactly:
\begin{equation}
	\mathrm{SNR}^{\mathrm{S-PCA}}_{\gamma,\beta}(0.5) = \sqrt{
        \frac{\frac12\gamma^{1/2} + \sqrt{ \frac{1}{4}\gamma + 2\beta\gamma }}{\beta\gamma}
        } = \gamma^{-1/4}\beta^{-1/2} + \BigOh(1) \qquad\textrm{as}\quad \beta\to0 \,.
\end{equation}
In Section~\ref{sec:Discussion}, we have observed that for the R-SVD, $\mathrm{SNR}^{v}_{\gamma,\beta}(0.5)\sim \BigOh(\beta^{-1/4})$ as $\beta\to 0$. In particular, for small $\beta$, the R-SVD yields ``reasonable'' estimates for the signal right singular vectors at substantially lower SNR $\sigma$. 
}

    \section{Optimal Shrinkage of ``Fast'' Singular Values}
    \label{sec:Shrinker}

	Consider the low-rank matrix denoising problem. Let $\bX\in \RR^{n\times m}$ be an unknown, rank $r$ matrix. One observes noisy measurements $\bY=\bX+\bZ$, and wishes to estimate $\bX$. We consider this problem under the asymptotics of the spiked model, as described in Section~\ref{sec:SetupMain}, where $\bX$ has the form \eqref{eq:X-def}, $\bZ$ has i.i.d. Gaussian entries $\m{N}(0,1/\sqrt{nm})$ and $n,m\to\infty$ with the rank $r$ fixed.
	
	A key question is how one should incorporate the known low-rank structure of the signal $\bX$ into the denoising process.
	One 
	popular and practical approach for doing so is \emph{singular value shrinkage} \cite{perry2009cross,shabalin2013reconstruction,donoho2014minimax,donoho2018optimal,gavish2017optimal,nadakuditi2014optshrink}. The idea is simple: $\bX$ is estimated by taking the singular value decomposition (SVD) of $\bY$, ``killing off'' the principal components (PCs) corresponding to small singular values (which represent noise), and re-weighting (in particular deflating) the large ``signal-bearing'' PCs to correct for the effects of noise.
	There exists by now a large and	fruitful literature devoted to singular value shrinkage (and variations) in the spiked model, under various different settings, cf. \cite{shabalin2013reconstruction,nadakuditi2014optshrink,gavish2014optimal,gavish2017optimal,donoho2018optimal,hong2018asymptotic,hong2018optimally,leeb2021optimal,leeb2021matrix,leeb2022operator,donoho2020screenot,su2022optimal, gavish2022mahalanoibis,gavish2022matrix}. 
	
	In this section we derive an optimal shrinkage rule for the randomized SVD. As mentioned, while R-SVD offers computational advantages over the full SVD (these being larger as $d$, the sketching dimension, decreases), the obtained ``fast'' singular vectors are worse approximations to the unobserved signal PCs; in other words, performing dimensionality reduction introduces additional noise into our estimates. Optimal singular value shrinkers designed for the full SVD \cite{shabalin2013reconstruction,nadakuditi2014optshrink,gavish2017optimal} are blind to this new source of noise, and accordingly are sub-optimal when used with the R-SVD. The formulas for the new optimal shrinkers will be obtained in a straightforward manner as a corollary from our theoretical results, Theorems~\ref{thm:2:Outliers}-\ref{thm:3:Angles}.
	
	Let $r_0$ be the number of {detectable} signal spikes, namely, such that $\sigma>\BBP$. Recall, by Theorem~\ref{thm:2:Outliers}, that a detectable signal spike corresponds to an observable outlier in the spectrum of $\hat{\bY}$. In that case, the true spike intensity can be consistently estimated from $\bY$ by inverting the spike-forward map:
	\begin{align}\label{eq:EstSpike}
		\sigma_i = \lim_{n\to\infty} \SpikeFunc^{-1}(\sigma_i(\hbY)),\qquad 1\le i \le r_0 \,.
	\end{align}
	Moreover, $r_0$ itself can be consistently estimated. Let $\delta>0$ be a tuning parameter; when $\delta$ is small enough, specifically $\delta < \SpikeFunc(\sigma_{r_0})-\NoiseEdgeUpper$, the estimator
	\begin{align}\label{eq:EstimateRank}
		\hat{r}(\delta)=\sum_{i=1}^{m} \mathds{1}(\hsigma_i\ge \NoiseEdgeUpper+\delta)  
	\end{align}
	satisfies $\hat{r}(\delta)\to r_0$ a.s. For simplicity, let us assume below that $r_0$ is known. Consider the family of all estimators of the form 
	\begin{equation}
		\hbX_{\bw} = \sum_{i=1}^{r_0} w_i \hbu_i \hbv_i^\T ,
	\end{equation}
	where $\bw=(w_1,\ldots,w_{r_0})$ are weights, possibly dependent on $\bY$. We would like to choose $\bw=\bw(\bY)$ so to minimize the Frobenius loss (MSE)
	\begin{equation}\label{eq:FrobeniusLoss}
		\m{L}(\bw) = \|\bX-\hbX_{\bw}\|_F^2 \,.
	\end{equation}
	Note that per \eqref{eq:FrobeniusLoss}, $\m{L}(\bw)$ is a random quantity (we do not take an expectation). The shrinker we devise is optimal in an \emph{asymptotic} sense: formally, we consider a sequence of denoising problems at increasing dimensions, $n,m,d\to\infty$.  We construct weights $\bw^*=\bw^*(\bY)$ which are asymptotically optimal in that 
	\begin{align}
		\lim_{n\to\infty}\m{L}(\bw^*(\bY)) = \lim_{n\to\infty} \min_{\bw\in \RR^{r_0}}\m{L}(\bw)
	\end{align}
	holds w.p. $1$. 
	
	We start with a simple observation: point-wise minimizers $\bw=\bw(\bY,\bX)$ of $\m{L}(\bw)$ are always bounded:
	\begin{lemma}\label{lem:bounded}
		Almost surely,
		\begin{align*}
			\min_{\bw\in\RR^{r_0}}\m{L}(\bw) = \min_{\bw\in\RR^{r_0}\;:\;\|\bw\|_{\infty}\le {2}\|\bX\|_F}\m{L}(\bw) \,.
		\end{align*}
	(Recall that $\|\bX\|_F^2=\sum_{i=1}^r \sigma_i^2$ is constant.) 
	\end{lemma}
\begin{proof}
	If $\bw$ minimizes $\m{L}(\bw)$ then in particular $\m{L}(\bw)\le \m{L}(\0)=\|\bX\|_F^2$. Furthermore, for any $1\le i \le r_0$, $\sqrt{\m{L}(\bw)}=\|\bX-\hbX_{\bw}\|_F \ge \|\hbX_{\bw}\|_F-\|\bX\|_F\ge |w_i| - \|\bX\|_F$. It follows that $|w_i|\le 2\|\bX\|_F$.
\end{proof}

Let $\bw=\bw(\bX,\bY)$ be any bounded weights. Expanding, 
\begin{align}
	\m{L}(\bw)=\|\bX-\hbX_{\bw}\|_F^2 
	= \sum_{i=1}^r \sigma_i^2 + \sum_{i=1}^{r_0}w_i -2 \sum_{i=1}^r\sum_{j=1}^{r_0} \sigma_i w_j \langle \bu_i,\hbu_j\rangle\langle \bv_i,\hbv_j\rangle\,, \nonumber 
\end{align}
so that by Theorems~\ref{thm:2:Outliers}-\ref{thm:3:Angles}, almost surely,
\begin{align}
	\lim_{n\to\infty}\m{L}(\bw)= \sum_{i=1}^{r_0}\left(\sigma_i^2 + w_i^2 -2\sigma_i w_i \m{U}_\gamma(\sigma_i)\m{V}_\gamma(\sigma_i)\right) + \sum_{i=r_0+1}^r \sigma_i^2 \,. \label{eq:AsympLossW}
\end{align}
Observe that \eqref{eq:AsympLossW} can be minimized explicitly in $\bw$. Specifically, the minimizer is
\begin{align}\label{eq:OptWeightsOracle}
	w_i^* = \sigma_i \m{U}_\gamma(\sigma_i)\m{V}_\gamma(\sigma_i),\qquad 1\le i \le r_0\,.
\end{align}	
Of course, the population spikes $\sigma_i$ are unknown; but per \eqref{eq:EstSpike}, they can be estimated consistently. Accordingly, define the {\it optimal shrinkage function} $\Shrinker : \RR_+\to \RR_+$:
\begin{align}\label{eq:Upsilon}
	\Shrinker(y)= \begin{cases}
		\SpikeFunc^{-1}(y)\UAngle(\SpikeFunc^{-1}(y))\VAngle(\SpikeFunc^{-1}(y))\quad&\textrm{if}\quad y>\sqrt{\NoiseEdgeUpper}, \\
		0\quad&\textrm{if}\quad y\le \sqrt{\NoiseEdgeUpper}
	\end{cases}\,.
\end{align}
Note that $\Shrinker(\cdot)$ is continuous at $y=\sqrt{\NoiseEdgeUpper}$, by Proposition~\ref{prop:vanishing-corrs}. Clearly, the weights
\begin{align}\label{eq:EstimatedWeights}
	\hat{w}_i^* = \Shrinker(\sigma_i(\hbY)),\qquad 1\le i \le r_0 \,,
\end{align}
which can be computed directly from $\hbY$, satisfy
$\hat{\bw}^*\to \bw^*$ w.p. $1$. 

The following is an immediate corollary of Lemma~\ref{lem:bounded} and \eqref{eq:AsympLossW}:
\begin{corollary}[Asymptotic optimality]
	\label{cor:optimality}
	Let $\hat{\bw}^*=\hat{\bw}^*(\hbY)\in\RR^{r_0}$ be the weights \eqref{eq:EstimatedWeights}. Almost surely,
	\begin{equation}\label{eq:cor:optimality}
		\lim_{n\to \infty} \m{L}(\hat{\bw}^*) = \lim_{n\to \infty} \min_{\bw\in\RR^{r_0}}\m{L}(\bw) \,.
	\end{equation}
\end{corollary} 

\begin{figure}
	\centering
	\includegraphics[width=0.6\textwidth]{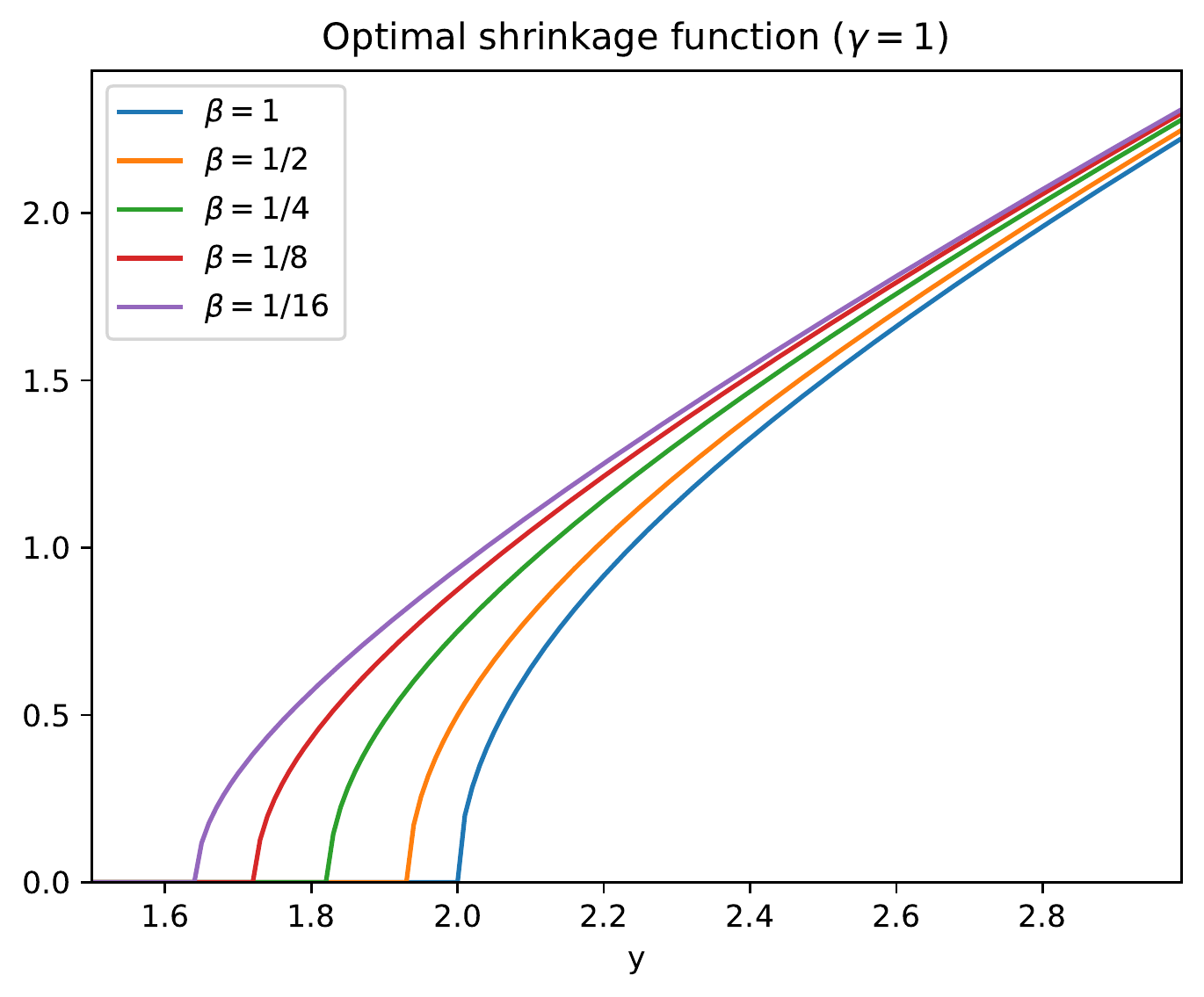}
	\caption{A plot of the optimal shrinkage function, $y\mapsto \Shrinker(y)$, given in \eqref{eq:Upsilon} for $\gamma=1$ and selected values of $\beta\in (0,1]$.}
	\label{fig:OptimalShrinker} 
\end{figure}
Figure~\ref{fig:OptimalShrinker} plots the optimal shrinkage function $\Shrinker(\cdot)$ for $\gamma=1$ and selected values of $\beta\in (0,1]$. 

\paragraph*{}
The optimal shrinker for $\beta=1$ (the full SVD) is well-known \cite{shabalin2013reconstruction,gavish2014optimal,gavish2017optimal}, and given by the explicit formula
\begin{equation}\label{eq:OldShrinker}
	\Upsilon_{\gamma,1}(y) = \begin{cases}
		\frac{1}{y}\sqrt{(y^2-\NoiseEdge^+_{\gamma,1})(y^2-\NoiseEdge^-_{\gamma,1})}\quad&\textrm{if }\quad y>\sqrt{\NoiseEdge^+_{\gamma,1}}\,,\\
		0 \quad&\textrm{if }\quad y\le \sqrt{\NoiseEdge^+_{\gamma,1}}
	\end{cases},
\end{equation}
where $\NoiseEdge^+_{\gamma,1}=\NoiseEdge^+_{\gamma}=(\gamma^{1/2}+\gamma^{-1/2})^2$. The elegant form of the shrinker \eqref{eq:OldShrinker} leads us to guess a closed-form formula for \eqref{eq:Upsilon}. Remarkably, exhaustive numerical verification supports the following claim:
\begin{conjecture}
	For all $\gamma\in (0,\infty)$ and $\beta\in (0,1)$, $\gamma\beta <1$, 
	the shrinker $\Shrinker(\cdot)$ from \eqref{eq:Upsilon} has the following closed-form expression:
	\begin{equation}\label{eq:ConjecturedShrinker}
		\Upsilon_{\gamma,\beta}(y) = \begin{cases}
			\frac{1}{y}\sqrt{(y^2-\NoiseEdge^+_{\gamma,\beta})(y^2-\NoiseEdge^-_{\gamma,\beta})}\quad&\textrm{if }\quad y>\sqrt{\NoiseEdge^+_{\gamma,\beta}}\,,\\
			0 \quad&\textrm{if }\quad y\le \sqrt{\NoiseEdge^+_{\gamma,\beta}}
		\end{cases}.
	\end{equation}
\end{conjecture}

\paragraph*{}
Note that the shrinker tuned for the full SVD, \eqref{eq:OldShrinker}, estimates the noise bulk edge at $\NoiseEdge^+_{\gamma,1}$, a higher value than the ``effective'' noise bulk edge $\NoiseEdgeUpper$ corresponding to a dimension-reduced randomized SVD with undersampling ratio $\beta$. Accordingly, it completely discards spikes which are weak but nevertheless detectable: $\NoiseEdge^+_{\gamma,1} < \SpikeFunc(\sigma_i) < \NoiseEdgeUpper$, and thus is particularly unsuitable for denoising in low SNR. In Figure~\ref{fig:MSE} we plot the asymptotic relative MSE suffered when estimating a rank-$1$ signal; that is, the curve $e(\sigma)=\lim_{n\to\infty} \frac{1}{\sigma^2}\|\sigma \bu\bv^\T - \Shrinker(\hsigma_1)\hbu_1\hbv_1^\T\|_F^2$. We also plot a similar curve, with the optimal shrinker $\Shrinker(\cdot)$ replaced by the optimal shrinkage rule for the full SVD \eqref{eq:OldShrinker}. It is clear that the improvement gained by the new shrinker is particularly noticeable when $\sigma$ is small---and accordingly, the best attainable error is quite large to begin with.
\begin{figure}
	\centering
	\includegraphics[width=0.6\textwidth]{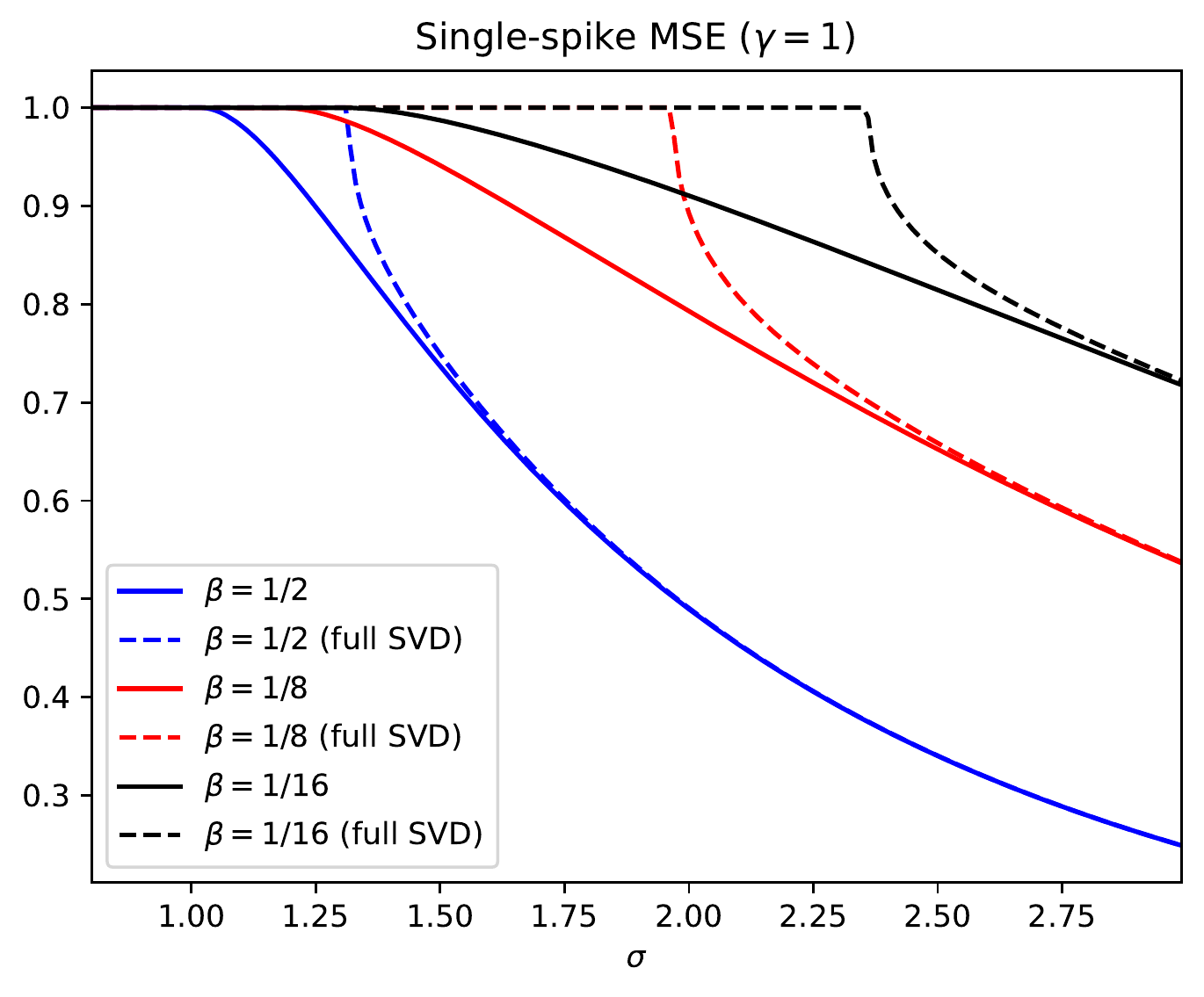}
	\caption{The asymptotic relative MSE for denoising one signal spike. Dashed curve: The performance of an optimal shrinker designed for the full SVD, \eqref{eq:OldShrinker}. The gains offered by the new shrinker are particularly noticeable when $\sigma$ is small.}
	\label{fig:MSE} 
\end{figure}

\paragraph*{Adapting to unknown noise variance.} We assumed throughout the derivation that the variance of the entries $Z_{i,j}$ is known, specifically $\Var(Z_{i,j})=1/\sqrt{nm}$. 
It is easy to see that when the variance is $\Var(Z_{i,j})=\rho^2/\sqrt{nm}$ for $\rho>0$, a spike is detectable when $\sigma_i>\BBP/\rho$. Accordingly, the previous derivation can be repeated with $r_0=\sum_{i=1}^r\mathds{1}(\sigma_i>\BBP/\rho)$, so that the optimal shrinkage rule becomes $\hat{w}_i^*=\rho\Upsilon(\sigma_i(\bY)/\rho)$. 

When $\rho$ is unknown, we propose to estimate it by the method of \cite{gavish2014optimal}. Specifically, let $\textrm{med}_{\gamma,\beta}$ be the median of the Marchenko-Pastur distribution with shape and scale parameters given in \eqref{eq:thm:QZ-LSD}. In light of Theorem~\ref{thm:1:Bulk},
\begin{align}\label{eq:NoiseVarianceEst}
	\hat{\rho}^2 = \frac{\mathrm{Median}(\sigma_1^2(\hbY),\ldots,\sigma_d^2(\hbY))}{\textrm{med}_{\gamma,\beta}}
\end{align}
is a consistent estimator for $\rho$ in the sense that $\hat{\rho}\to\rho$ a.s.
	
\paragraph*{Robustness to rank overestimation} In the preceding discussion, we proposed to estimate $r_0$, the number of outlying singular values, using \eqref{eq:EstimateRank}. We remark that the performance of our shrinker is robust with respect to overestimation of the rank. More precisely, let $k$ be a \emph{constant} (as $n,m,d\to\infty$) upper bound $k\ge r_0$. If we set \eqref{eq:EstimatedWeights} for all $1\le i \le k$ then the resulting estimator $\hbX_{\bw} = \sum_{i=1}^{k} \hat{w}_i^* \hbu_i \hbv_i^\T$ satisfies \eqref{eq:cor:optimality} as well. This follows since $\Shrinker(y)$ is continuous and vanishing for $y\le \sqrt{\NoiseEdgeUpper}$ and $\sigma_{i}(\hbY)\to\sqrt{\NoiseEdgeUpper}$ for all $r_0< i\le k$.

\paragraph*{Summary of the proposed shrinker.} For the practitioner's convenience, we repeat here the full details of the proposed shrinkage rule:
\begin{itemize}
	\item {\it Input: } Reduced data matrix $\hbY$; denote the SVD $\hbY\overset{SVD}{=}\sum_{i=1}^{d}\sigma_i(\hbY)\hbu_i \hbv_i^\T$.\\
	{\it Either:} 1) A small parameter $\delta>0$; or 2) $k$, an upper bound on the rank $k\ge r_0$.
	\item {\it Output: } Weights $\hat{w}_1^*,\ldots,\hat{w}^*_{\hat{r}}$, so that $\hbX=\sum_{i=1}^{\hat{r}}\hat{w}_i \hbu\hbv^\T$ is the denoiser.
\end{itemize}
Steps:
\begin{enumerate}
	\item Estimate the noise variance, $\hat{\rho}$, by \eqref{eq:NoiseVarianceEst}.
	\item (Optional:) Estimate the number of strong spikes ($r_0$): $\hat{r}=\hat{r}(\delta)=\sum_{i=1}^d \mathds{1}(\sigma_i(\hbY)/\hat{\rho}>\NoiseEdgeUpper+\delta)$. \\
	Conversely: set $\hat{r}=k$ (possibly over-estimating the rank).
	\item For every $1\le i \le \hat{r}$, set 
	\begin{align*}
		\hat{w}_i^* = \hat{\rho}\Shrinker(\sigma_i(\hbY)/\hat{\rho}),
	\end{align*}
	where $\Shrinker(\cdot)$ appears in \eqref{eq:ConjecturedShrinker}.
\end{enumerate}
	
%	\begin{remark}
%		We emphasize that formulas above are tailor-made for the randomized SVD used with {\bf random projections.} Since multiplying by a random projection matrix is computationally costly, many other, more structured, sketching transformations have been proposed in the literature, see for example \cite{halko2011finding,woodruff2014sketching}. Some of these proposals are, in fact, \emph{not} projection matrices, and therefore our limiting formulas do not apply for them; see also the paper \cite{yang2021reduce}, which explores the effect of the sketching matrix in sketched PCA under the spiked model.
%	\end{remark}

    \section{Numerical Experiments}
    \label{sec:Experiments}

    \subsection{Finite-$n$ scaling and universality}

	Our main results, Theorems~\ref{thm:2:Outliers}-\ref{thm:3:Angles},
	establish the asymptotic convergences of the leading singular values and singular vector angles towards deterministic expressions. The question of convergence rate (when $n,m,d$ are finite) is a natural one. Figure~\ref{fig:FiniteN} reports on a numerical experiment exploring this point. The setting is this: we consider a rank-$1$ spiked matrix,
	\begin{align*}
		\bY = \sigma \bu\bv^\T + \bZ\,,
	\end{align*}
	for $\sigma$ fixed above and below the detectability threshold, specifically $\sigma\in \{0.7\BBP,1.5\BBP\}$.
	 For increasing values of $n,m,d$ (with ratios $\gamma,\beta$ fixed), we report the absolute deviation of the largest singular value $\sigma_1(\hbY)$ and singular vector angles $\langle \bu,\hbu\rangle\langle \bv,\hbv\rangle$ from their  limiting values. Each point on the graph represents the average of $100$ Monte-Carlo trials. The results, plotted in a log-log scale, suggest that the expected absolute deviation scales roughly like $n^{-1/2}$ as $n$ increases.
	\begin{figure}
	\centering
	\begin{subfigure}[b]{0.45\textwidth}
		\centering
		\includegraphics[width=\textwidth]{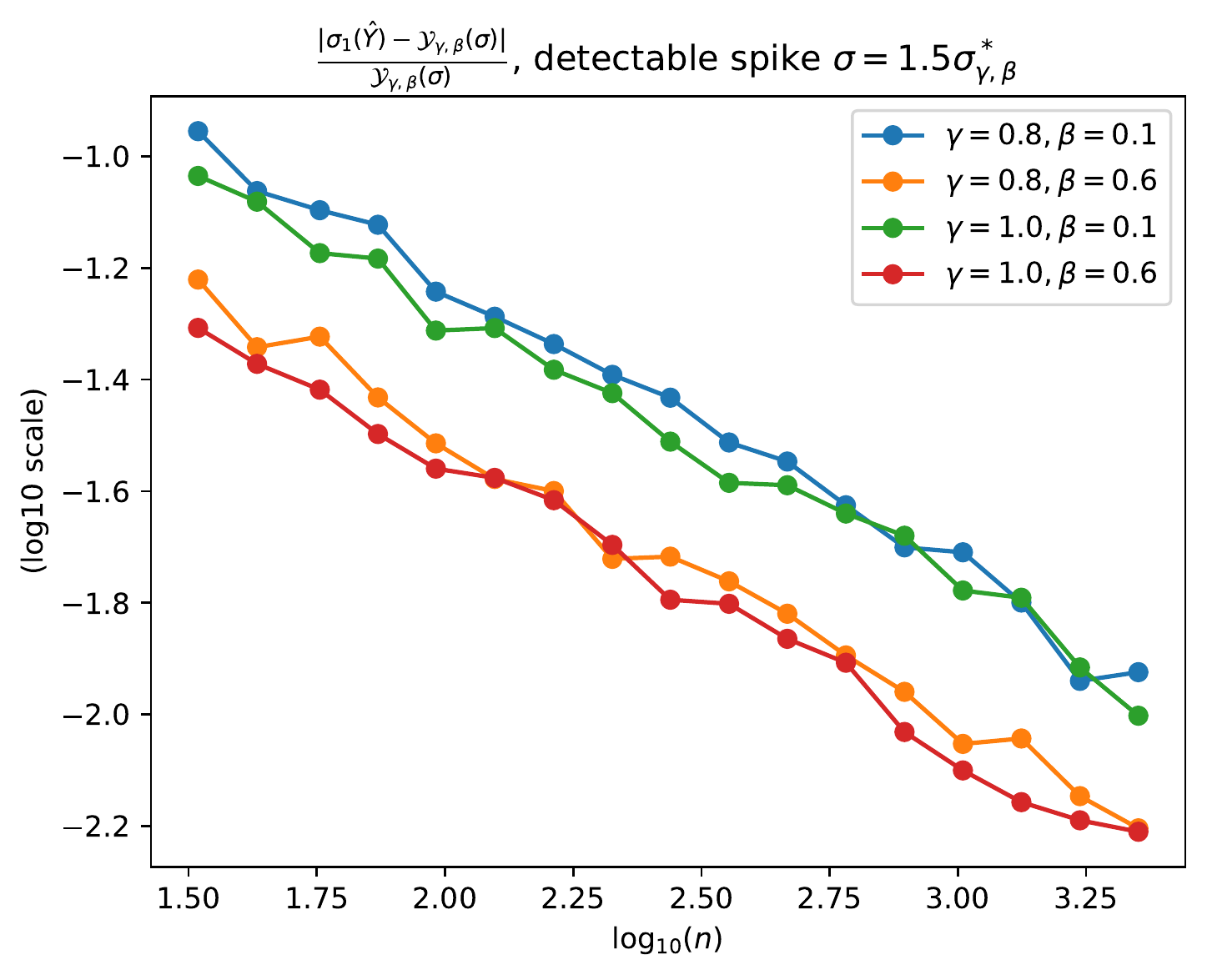}
	\end{subfigure}
	\hfill
	\begin{subfigure}[b]{0.45\textwidth}
		\centering
		\includegraphics[width=\textwidth]{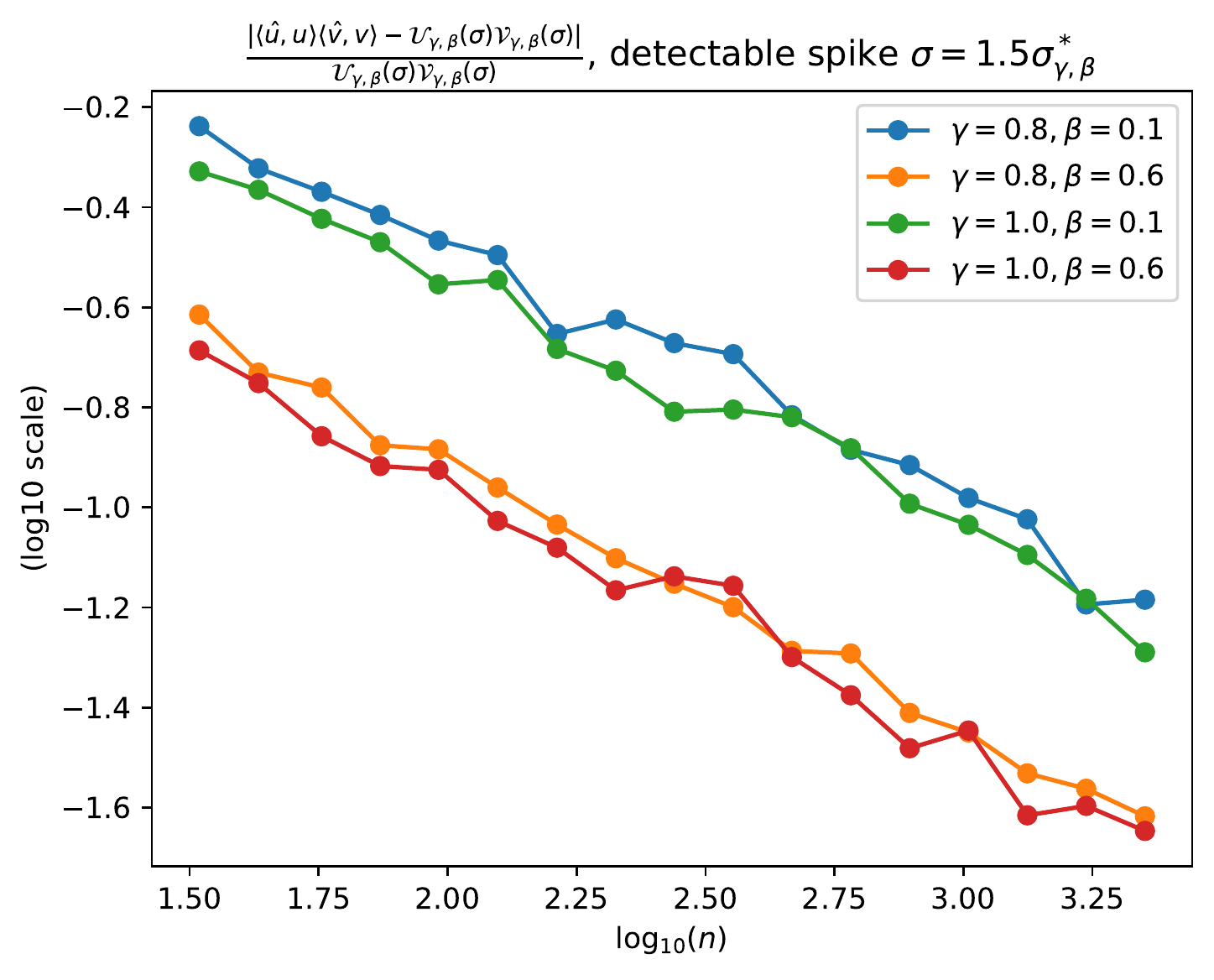}
	\end{subfigure}\\
	\vfill
	\begin{subfigure}[b]{0.45\textwidth}
		\centering
		\includegraphics[width=\textwidth]{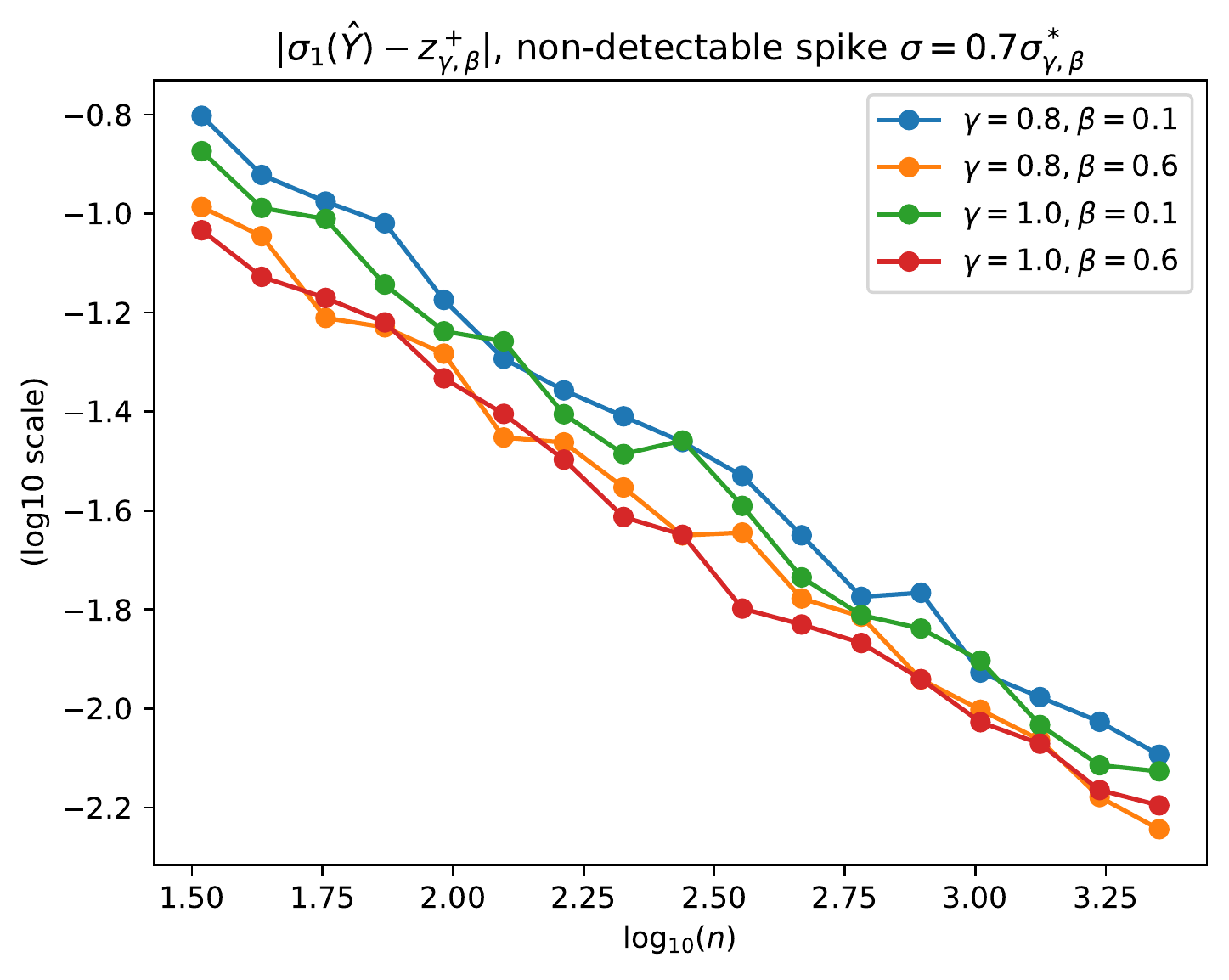}
	\end{subfigure}
	\hfill
	\begin{subfigure}[b]{0.45\textwidth}
		\centering
		\includegraphics[width=\textwidth]{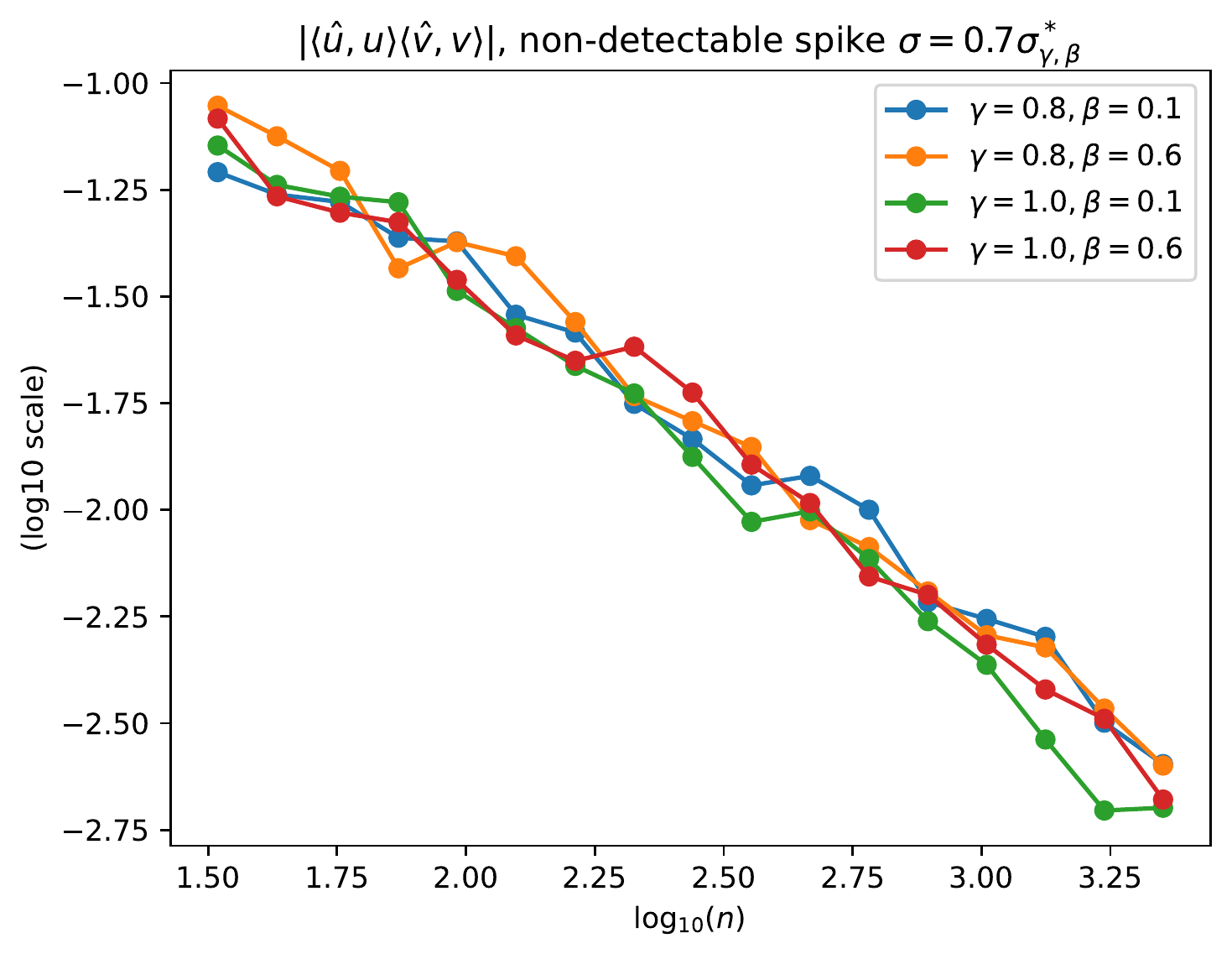}
	\end{subfigure}
	\caption{The finite-$n$ absolute deviation of the largest singular value and the singular vector angles, for a rank-$1$ spiked data matrix and under several parameters combinations. The error appears to scale, roughly, like $n^{-1/2}$. Top: detectable spiked (above the threshold); bottom: non-detectable spike (below the threshold).}
	\label{fig:FiniteN}
	\end{figure}
	
	Another important matter is the universality of our results with respect to the noise distribution.  Our theorems are stated and proved for Gaussian noise matrices $\bZ$ (in fact, our proofs explicitly use the orthogonal invariance of the Gaussian distribution). One would expect, similar to many other results in random matrix theory \cite{bai2010spectral}, that our asymptotic results should in fact be universal over a large class of (sufficiently light-tailed) i.i.d. matrices having the same first and second moments. In Figure~\ref{fig:Universality} we plot the finite-$n$ absolute deviation of the largest singular value and the singular vector angles (under a similar setting as in the experiments of Figure~\ref{fig:FiniteN}) for three different i.i.d. noise distributions: Gaussian, Rademacher and Student's t with $5$ degrees of freedom. For either choice, we observe convergence to the limiting expressions at roughly the same rate $n^{-1/2}$. The error, however, is remarkedly higher for the heavier-tailed Student's t distribution; this is, of course, not surprising.
	\begin{figure}
	\centering
	\begin{subfigure}[b]{0.45\textwidth}
		\centering
		\includegraphics[width=\textwidth]{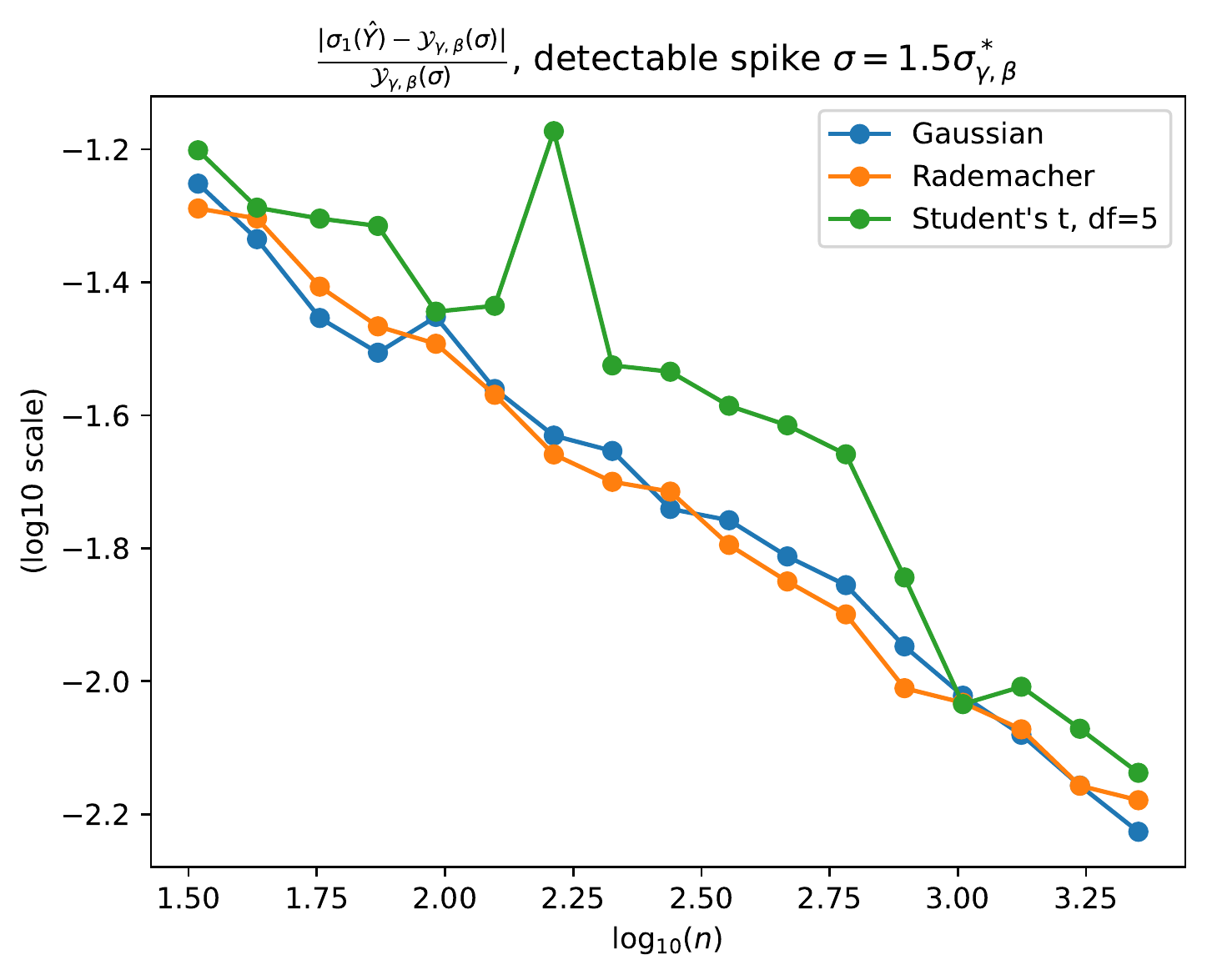}
	\end{subfigure}
	\hfill
	\begin{subfigure}[b]{0.45\textwidth}
		\centering
		\includegraphics[width=\textwidth]{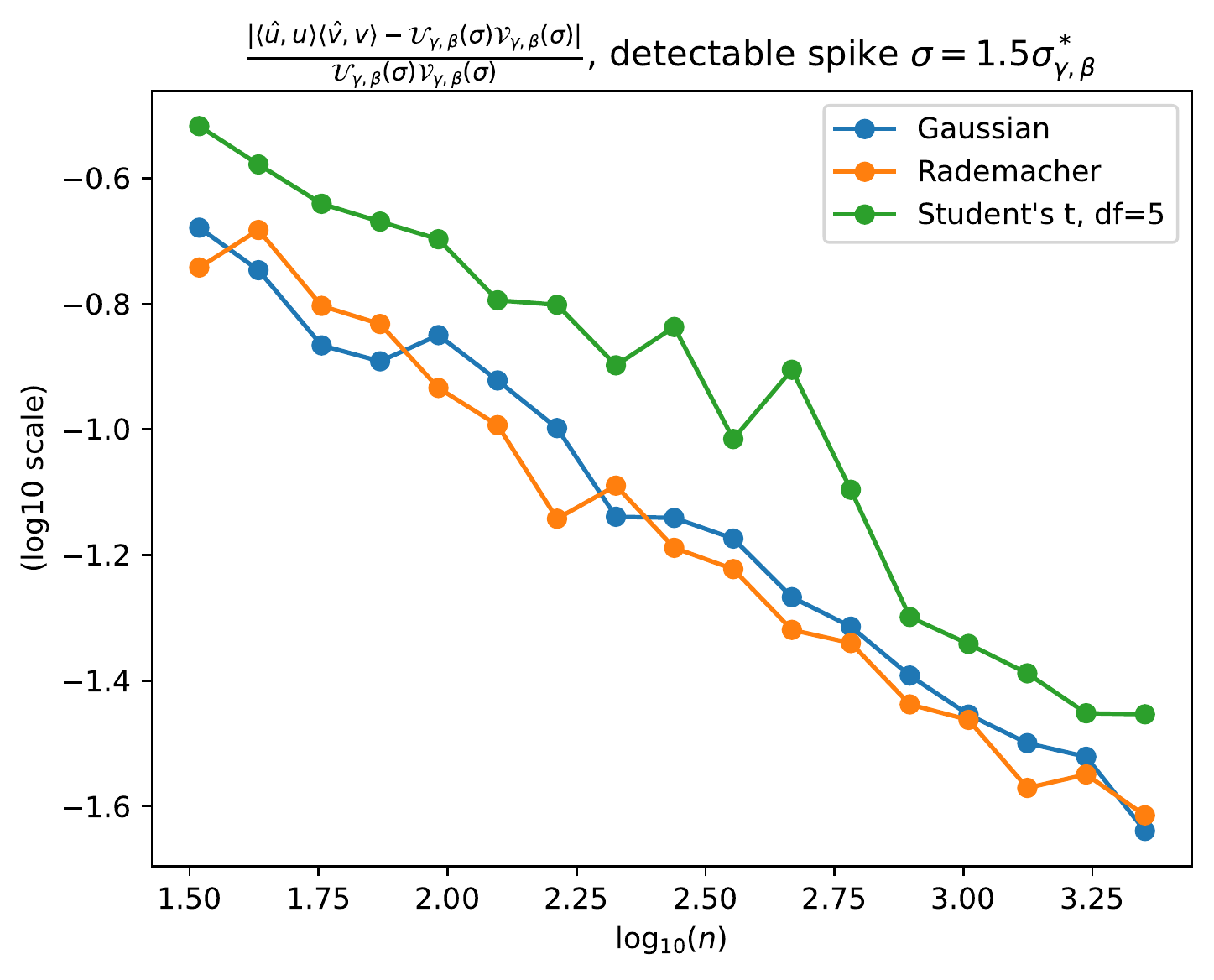}
	\end{subfigure}
	\caption{The finite-$n$ absolute deviation of the largest singular value and the singular vector angles, for three different choices of noise distributions. The parameters used are $\gamma=1,\beta=1/2$. In all cases, the expected error appears to scale, roughly, like $n^{-1/2}$. For the heavier-tailed Student's t distribution, the expected error appears to be larger than for the light-tailed Gaussian and Rademacher distributions, and likewise the variance across different Monte-Carlo trials (each point on the plot corresponds to $100$ trials).}
	\label{fig:Universality}
	\end{figure}
	
	\subsection{Singular value shrinkage}

	In Section~\ref{sec:Shrinker} we derived an optimal singular value shrinker for the randomized SVD. Its optimality is with respect to the Frobenius loss (MSE), in an \emph{asymptotic} sense. The next experiment demonstrates the validity of our theory, and examines the finite-n scaling of the error relative to an exactly optimal (oracle) shrinkage rule.
	
	For $\gamma=1,\beta=1/10$ and increasing values of $n$, we generate a $3$-spike signal 
	\begin{align*}
		\bX = \sum_{i=1}^3 \sigma_i \bu_i\bv_i^\T,\qquad\textrm{where}\qquad (\sigma_1,\sigma_2,\sigma_3)=(1.4\BBP,1.1\BBP,0.5\BBP),
	\end{align*}
	where the directions $\bu_i,\bv_i$ are uniformly random unit vectors. Setting a known upper bound on the rank $k=6$, 
	we apply the asymptotically optimal shrinkage scheme described 
	in Section~\ref{sec:Shrinker}. Denote the resulting estimation error by 
	\begin{align*}
		\eps_{\mathrm{Shrink}} = \left\| \bX - \sum_{i=1}^k \Shrinker(\sigma_i(\hbY))\hbu_i\hbv_i^\T \right\|_F^2 \,.
	\end{align*}
	We compare it to the error of a clairvoyant singular value shrinker, that retains the top $k=6$ PCs and can re-weight them optimally:
	\begin{align*}
		\eps_{\mathrm{Oracle}} = \min_{w_1,\ldots,w_k}\left\| \bX - \sum_{i=1}^k w_i \hbu_i\hbv_i^\T \right\|_F^2\,.
	\end{align*}
The quantity $\bar{\eps}\equiv (\eps_{\mathrm{Shrink}}-\eps_{\mathrm{Oracle}})/\eps_{\mathrm{Oracle}}$ is the relative {excess error} suffered by our shrinkage rule, compared to an oracle-optimal singular value shrinkage scheme that retains at most $k=6$ PCs. Note that if we had taken $k=2$ (the number of detectable spikes), Corollary~\ref{cor:optimality} would imply that $\bar{\eps}\to 0$ w.p. $1$ as $n,m,d\to\infty$. We expect (but did not prove) that subleading empirical PCs should only be weakly correlated with the signal spikes, and consequently that $\Expt [\bar{\eps}] \to 0$. We aim to check this claim.  

\begin{figure}
	\centering
	\includegraphics[width=0.6\textwidth]{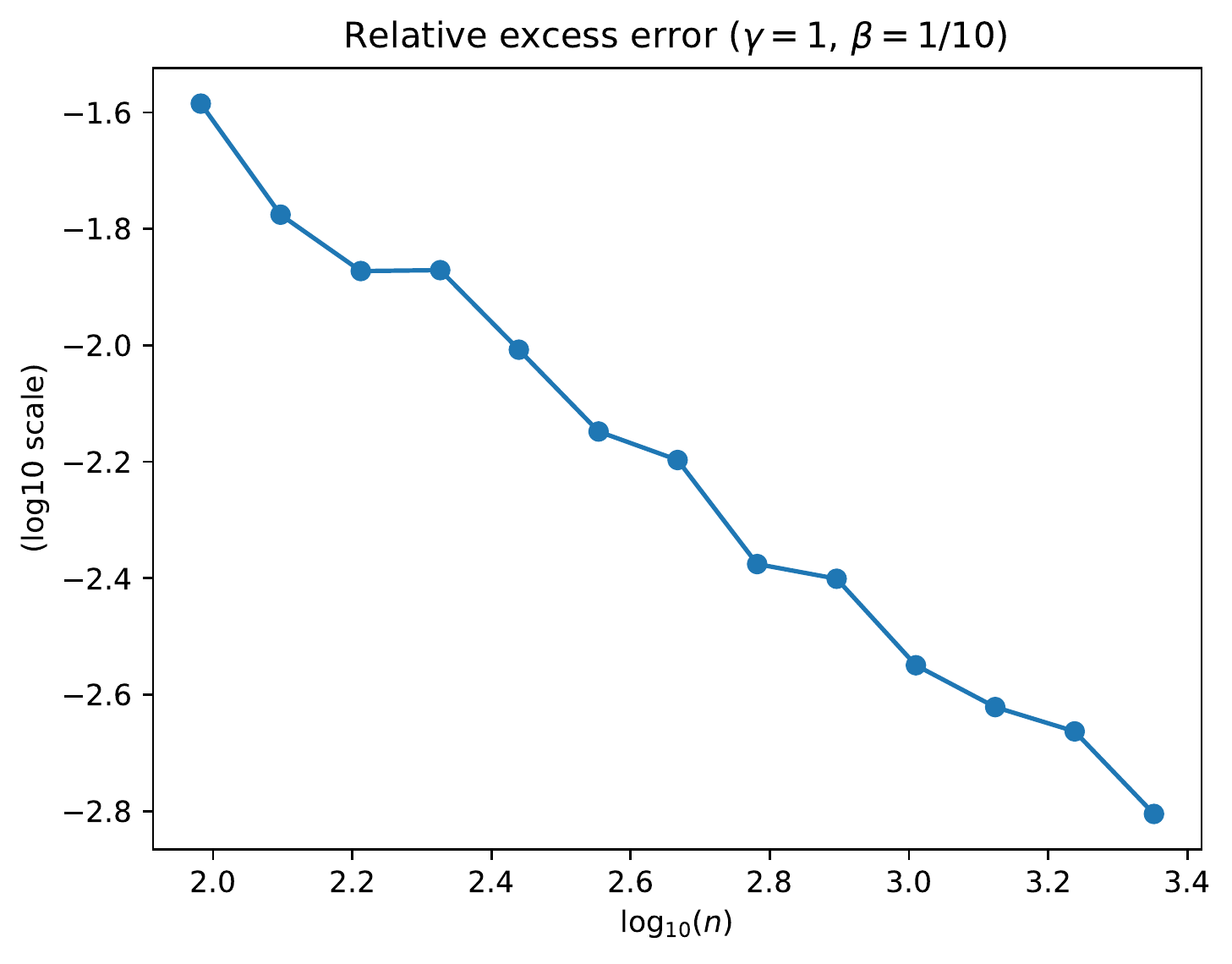}
	\caption{The relative excess error of the proposed asymptotically optimal shrinker (described in Section~\ref{sec:Shrinker}) relative to an oracle-optimal shrinker that retains at most $k\ge r$ ($r=3,k=6$) PCs. It is clear that as $n$ increases, the excess error decreases. 
%		Each point on the plot corresponds to the average excess error averaged over a Monte-Carlo trials. 
		(Plotted in log-log scale.)}
	\label{fig:FiniteNShrinker}
\end{figure}
Figure~\ref{fig:FiniteNShrinker} plots the finite-$n$ scaling of the excess error, plotted in a log-log scale; each point on the graph represents the average of $100$ Monte-Carlo trials. It is evident that the excess error indeed decays with $n$.
% at a rate roughly $\Expt [\bar{\eps}]\sim n^{-1}$. 
Moreover, already at relatively moderate dimensions ($n,m\sim 10^2$ so $d\sim 10$) the excess error for this setup is quite small ($\bar{\eps}\sim 0.03$).

    \section{Proofs}
    \label{sec:Proofs}

    \paragraph*{Notation.} For two sequences of numbers, $a=a_n,b=b_n$, we denote $a\simeq b$ if $a-b\to 0$ almost surely as $n\to\infty$. For vectors $\bm{a},\bm{b}$, possibly of diverging dimensions (e.g., $\bm{a},\bm{b}\in\RR^n$) we use $\bm{a}\simeq \bm{b}$ to mean that $\|\bm{a}-\bm{b}\|\to 0$, where $\|\cdot\|=\|\cdot\|_2$ is the Euclidean norm. Similarly, for matrices $\bm{A},\bm{B}$ the notation $\bm{A}\simeq \bm{B}$ means that $\|\bm{A}-\bm{B}\|\to 0$ where $\|\cdot\|=\|\cdot\|_{2}$ is the operator ($\ell_2$-to-$\ell_2$) norm, equivalently the largest singular value. 

    \paragraph*{}
    Throughout the analysis, we assume w.l.o.g. that the projection matrix $\bOmega \in \RR^{d\times m}$ 
    is a projection onto the first $d$ coordinates:
    % projects an $m$-dimensional vector onto its first $d$ coordinates: 
    \[
        % \bOmega = \MatL \bI_{d\times d}\\ \bm{0}_{m\times (m-d)} \MatR \,.
        % % ,\qquad \textrm{that is, on any $\bm{x}\in \RR^d$,}\quad \bOmega \MatL x_1,\cdots,x_m\MatR^\T  = \MatL x_1,\ldots,x_d\MatR^\T  .
        \bOmega = \MatL \bI_{d\times d} &\bm{0}_{d\times (m-d)}\MatR .
    \]
    We can indeed do so since the noise matrix $\bZ$ is assumed i.i.d. Gaussian, hence its distribution is orthogonally invariant. Note that under this setup, condition \eqref{eq:assum:incoherence} can be interpreted as purely an incoherence condition on the population spikes.  We decompose:
    \begin{equation}
	\bV = \MatL
		\bV_1 \in \RR^{d\times r} \\
		\bV_2 \in \RR^{(m-d)\times r}
	\MatR, \quad 
	\bZ = \MatL
		\bZ_1 \in \RR^{n\times d} & \bZ_2 \in \RR^{n\times (m-d)}
	\MatR,
    \end{equation}
    so that the sketched data matrix is
    \begin{equation}
        \bYtilde=\bY\bOmega^\T=\bU\bLambda\bV_1^\T+\bZ_1 .
    \end{equation}

    \paragraph*{}
    A key step in our analysis boils down to (approximately) decomposing $\hbY=\bPc\bY$ into the sum of a  low-rank ``signal'' plus ``noise'' matrix.
    In doing so, our aim is to mimic the form of the original data matrix $\bY=\bU\bLambda\bV^\T + \bZ$, which is a rank-$r$ additive perturbation of the noise matrix $\bZ$. 
    Note that performing such a decomposition in our setup is not immediate. This is because the projection $\bPc$ is constructed from $\bY$ in its entirety, both noise and signal included; in particular, the projected noise matrix $\bPc\bZ$ itself contains ``signal''. 

    Similar to how the noise matrix $\bZ$ is the ``benchmark'' for the matrix $\bY$, we will compare the spectrum of $\hbY$ to that of a similarly reduced data matrix, that contains only noise and no signal. To wit, let $\bQc:\RR^n\to\RR^n$ be the projection onto the column space of $\bZ\bOmega^\T=\bZ_1 \in \RR^{n\times d}$. The idea, then, is to express $\hbY$ as a perturbation of the reduced noise matrix $\hbZ=\bQc\bZ$, which is the observed matrix in the total absence of a signal ($\sigma_1=\ldots=\sigma_r=0$). 

    The first step of the computation consists of relating the signal-bearing column space projection $\bPc$ to the pure-noise projection $\bQc$. Crucially, one has closed-form expression for the projections:
    \begin{equation}
        \bPc = \bYtilde(\bYtilde^\T\bYtilde)^{-1}\bYtilde^\T,\qquad 
        \bQc = \bZ_1(\bZ_1^\T\bZ_1)^{-1}\bZ_1^\T .
    \end{equation}
    Note that since, by assumption, $d<n$ and $\bZ_1\in\RR^{n\times d}$ has a continuous distribution, the $d$-by-$d$ matrices $\bYtilde^\T\bYtilde,\bZ_1^\T\bZ_1$ are indeed invertible w.p. $1$. Denote
	\begin{equation}\label{eq:Vbar-1-def}
		\bVbar_1 = \bZ_1(\bZ_1^\T\bZ_1)^{-1}\bV_1 \in \RR^{n\times r},
	\end{equation}
    and the functions $f,g:\RR_+\to\RR_+$,
    \begin{equation}\label{eq:f-g}
		f(\sigma)=\frac{\sigma}{1+(\beta \sqrt{\gamma}) \sigma^2},\qquad g(\sigma) = \sigma f(\sigma) \,.
	\end{equation}
    % given in \eqref{eq:f-g}.

    Furthermore, let $\bQc^\perp$ be the projection onto the orthogonal complement of $\range(\bQc)$.
    \begin{lemma}\label{lem:Projection}
        The projection $\bPc$ is approximated by a low-rank perturbation of $\bQc$:
    	\begin{equation}\label{eq:Projector}
    		\begin{split}
    			\bPc 
    			&\simeq \bQc 
    			-(1-\beta\gamma) \bVbar_1g(\bLambda)\bVbar_1^\T + \frac{\beta\sqrt{\gamma}}{1-\beta\gamma}(\bQc^\perp\bU) g(\bLambda)(\bQc^\perp\bU)^\T  \\
    			&+ \bVbar_1 f(\bLambda)(\bQc^\perp\bU)^\T
    			+  (\bQc^\perp\bU)f(\bLambda)\bVbar_1^\T \,.
    		\end{split}
    	\end{equation}
    \end{lemma}
    The proof of Lemma~\ref{lem:Projection} is a straightforward (if tedious) calculation, and uses the Sherman-Morrison-Woodbury formula, some elementary concentration results for quadratic forms and explicit expressions for the low-order trace moments of the Wishart distribution. The details are deferred to Appendix, Section~\ref{sec:proof-lem:Projection}.

    Using Lemma~\ref{lem:Projection}, we approximate the reduced matrix $\hbY$ by a signal-plus-noise matrix. To this end, define 
    \begin{equation}\label{eq:A-B-def}
		\bA = \MatL  
		\bU &\bQc^\perp \bU &\bVbar_1
		\MatR \in \RR^{n\times 3r},\qquad 
		\bB = \MatL
			\bV & \bW_1 &\bW_2
		\MatR \in \RR^{m\times 3r}
	\end{equation}
	where
	\begin{equation}\label{eq:W1-W2-def}
		\bW_1 = \MatL \0\\\bV_2-\bZ_2^\T \bVbar_1 \MatR,\qquad \bW_2 = \MatL \0 \\ \bZ_2^\T \bQc^\perp \bU \MatR ,\qquad \bW_1,\bW_2\in \RR^{m\times r}.
	\end{equation}
	Also denote 
	\begin{equation}\label{eq:bSigma-def}
		\bSigma = \MatL
		\bLambda &\0 &\0 \\
		\0 &-f(\bLambda) &\frac{\beta\sqrt{\gamma}}{1-\beta\gamma}g(\bLambda) \\
		\0 &(1-\beta\gamma)g(\bLambda) &f(\bLambda)
		\MatR.
	\end{equation}

    \begin{lemma}\label{lem:PY:2}
        We have 
        \begin{equation}
		\label{eq:PY:2}
		\hbY  \simeq \bA \bSigma \bB^\T + \bQc\bZ \,.
	\end{equation}
    That is, $\hbY=\bPc\bY$ is, approximately, a rank-$3r$ perturbation of the signal-less reduced matrix $\bQc\bZ$. 
    \end{lemma}
    The proof of Lemma~\ref{lem:PY:2} appears in Appendix, Section~\ref{sec:proof-lem:PY:2}.

    We have approximated $\hbY$, in operator norm, by a low-rank plus noise matrix $\bA\bSigma\bB+\bQc\bZ$. Note that by standard perturbations results for singular values and vectors (for example, Davis-Kahan, see e.g. \cite{yu2015useful}), the singular values and \emph{outlier} singular vectors of $\hbY$ are consistently approximated (vanishing error as $n,m,d\to\infty$) by the r.h.s. of \eqref{eq:PY:2}.

    The analysis of the matrix $\bA\bSigma\bB+\bQc\bZ$ consists of two parts. First, we analyze the spectrum of the reduced pure-noise matrix $\bQc\bZ$; the LSD of this matrix defines the limiting shape of the bulk singular values of $\hbY$. Second, we analyze the outlier singular values and vectors. The computation of the limiting formulas relies, in part, on results derived in the first part. 

    \subsection{The pure-noise spectrum}

    The following are the main results of this section:

    \begin{theorem}
    [The limiting singular value distribution]
    \label{thm:QZ-LSD}
        The empirical distribution of the singular values squared of $\bQc\bZ$, namely $\sigma_1^2(\bQc\bZ),\ldots,\sigma_{d}^2(\bQc\bZ)$, converges almost surely to a Marcheko-Pastur law, with the parameters as given in Theorem~\ref{thm:1:Bulk}.
        % \begin{equation}\label{eq:thm:QZ-LSD}
        %     \phi = \frac{\gamma\beta}{1+\gamma-\gamma\beta},\qquad \eta^2 = \gamma^{-1/2}+(1-\beta)\gamma^{1/2}.
        % \end{equation} 
    \end{theorem}

    \begin{theorem}
        [The largest singular value] 
        \label{thm:QZ:largest}
        $\sigma_1(\bQc\bZ)$ converges almost surely to the upper edge of the limiting spectral distribution (LSD):
        \begin{equation}
            \sigma_1^2(\bQc\bZ)\overset{a.s.}{\longrightarrow} \NoiseEdge_{\gamma,\beta}^{+},
            % \equiv 
            % \gamma^{-1/2}+\gamma^{1/2} + 2\sqrt{\beta(1+\gamma-\gamma\beta)}.
        \end{equation}
        with $\NoiseEdge_{\gamma,\beta}^{+}$ given in \eqref{eq:edges-def}.
    \end{theorem}

    Towards proving Theorem~\ref{thm:QZ-LSD}-\ref{thm:QZ:largest}, we introduce some notation. Let $\bBc$ be an $n$-by-$n$ orthogonal matrix whose columns are eigenvectors of $\bZ_1\bZ_1^\T$; specifically, let $\bBc_1\in\RR^{n\times d}$ correspond to the non-zero eigenvalues and $\bBc_2\in \RR^{n\times (n-d)}$ correspond to the zero eigenspace (the columns chosen arbitrarily to complete an orthonormal basis of $\RR^n$). Accordingly, denote $\bZ_1\bZ_1^\T = \bBc\diag(\bmu,\0)\bBc^\T$  where $\bmu=(\mu_1,\ldots,\mu_d)\in \RR_+^d$ are the non-zero eigenvalues.  Note that the projection onto the column space of $\bZ_1$ can be written as $\bQc=\bBc_1\bBc_1^\T$. Consequently, the matrix $\bQc\bZ(\bQc\bZ)^\T =\bZ_1\bZ_1^\T + \bQc\bZ_2\bZ_2^\T\bQc$ can be written, upon a change of basis, as 
    \begin{equation}
    \label{eq:QZ-basis-change}
        \bBc^\T \bQc\bZ(\bQc\bZ)^\T \bBc = 
        \MatL
            \diag(\bmu) + \bX\bX^\T & \0 \\
            \0 & \0
        \MatR, \qquad\textrm{where}\quad \bX=\bBc_1^\T \bZ_2 .
    \end{equation}
    Clearly, the non-zero singular values squared $\sigma_1^2(\bQc\bZ),\ldots,\sigma_d^2(\bQc\bZ)$ are exactly the eigenvalues of the top-left $d$-by-$d$ block of \eqref{eq:QZ-basis-change}, namely $\diag(\bmu) + \bX\bX^\T$. Crucially, observe that since $\bBc_1$ is independent of $\bZ_2$, $\bX\in \RR^{d\times (m-d)}$ is an i.i.d. Gaussian matrix with entries $\bX_{i,j}\overset{i.i.d.}{\sim}\m{N}(0,1/\sqrt{nm})$ and independent of $\bZ_1$.

    \begin{proof}
        (Of Theorem~\ref{thm:QZ-LSD}.)
        We wish to find the limiting eigenvalue distribution of $\diag(\bmu)+\bX\bX^\T$. Recall that $\bmu\in \RR^d$ is a vector containing the $d$ non-zero eigenvalues of the $n$-by-$n$ matrix $\bZ_1\bZ_1^\T$. These eigenvalues are exactly the $d$ eigenvalues of the $d$-by-$d$ matrix $\bZ_1^\T\bZ_1$. Since $\bX$ has an orthogonally invariant distribution (being an i.i.d. Gaussian matrix) and is independent of $\bZ_1$, the eigenvalues of $\diag(\bmu)+\bX\bX^\T$ have the same distribution as those of $\bZ_1^\T \bZ_1 + \bX\bX^\T$. Now, denote $\bm{T}=\MatL \bZ_1^\T &\bX\MatR \in \RR^{d\times (n+m-d)}$ which has i.i.d. Gaussian entries $\bm{T}_{i,j}\sim \m{N}(0,1/\sqrt{nm})$. Clearly, $\bZ_1^\T \bZ_1 + \bX_1\bX_1^\T = \bm{T}\bm{T}^\T$. The matrix $\frac{\sqrt{nm}}{n+m-d}\bm{T}\bm{T}^\T$ is a sample covariance matrix corresponding to $n+m-d$ i.i.d. measurements $\bm{t}_i\sim \m{N}(0,\bI)$ in $\RR^d$. Accordingly, its eigenvalue distributions converges to the Marchenko-Pastur law with scale $1$ and shape $\phi\simeq \frac{d}{n+m-d}=\frac{\gamma\beta}{1+\gamma-\gamma\beta}$; see for example \cite{bai2010spectral}. Hence, the limiting eigenvalue distribution of $\bm{T}\bm{T}^\T$ is a Marchenko-Pastur law with the same shape, and scale $\eta^2\simeq \frac{n+m-d}{\sqrt{nm}}\simeq \frac{1+\gamma-\gamma\beta}{\sqrt{\gamma}}$.
        % $\gamma^{-1/2}+(1-\beta)\gamma^{1/2}$.
        % $\frac{1+\gamma-\gamma\beta}{\sqrt{\gamma}}$.

    \end{proof}

    \begin{proof}
        (Of Theorem~\ref{thm:QZ:largest}.)
        As noted above, in the proof of Theorem~\ref{thm:QZ-LSD}, $\sigma_1^2(\bQc\bZ)$ has the same distribution as the largest eigenvalue of the matrix $\bm{T}\bm{T}^\T$. By well-known results on the largest eigenvalue of a Gaussian sample covariance matrix, $\lambda_1(\bm{T}\bm{T}^\T)$ converges a.s. to the upper edge of the corresponding Marchenko-Pastur law; see \cite{bai2010spectral}. In our case, this upper edge is $\eta^2(1+\sqrt{\phi})^2$ with $\eta^2,\phi$ given in \eqref{eq:thm:QZ-LSD}.
    \end{proof}

    For $z>\lambda_{\phi,\eta^2}^+$, denote by $\MPStiel_{\phi,\eta^2}(z)$ the \emph{Stieltjes transform} of a Marchenko-Pastur law with shape and scale parameters $\phi,\eta^2$,
    \begin{equation}
        \MPStiel_{\phi,\eta^2}(z) = \int \frac{1}{\lambda-z}d\MPDensity_{\phi,\eta^2}(\lambda),
    \end{equation}
    which has the following closed-form formula
    (see for example \cite[Lemma 3.11]{bai2010spectral}):
    \begin{equation}\label{eq:MP-Stiel}
        \MPStiel_{\phi,\eta^2}(z) = \frac{\eta^2(1-\phi) -z + \sqrt{-(\lambda_{\phi,\eta^2}^+-z)(z-\lambda_{\phi,\eta^2}^-)}}{2\eta^2\phi z},\qquad \lambda_{\phi,\eta^2}^\pm = \eta^2(1\pm \sqrt{\phi})^2 \,.
    \end{equation}
    Denote by $\rho_{\gamma,\beta}(\cdot)$ the Stieltjes transform \eqref{eq:MP-Stiel} of the law in Theorem~\ref{thm:QZ-LSD}, namely,
    \begin{equation}\label{eq:rho}
        % \rho_{\gamma,\beta}(z) = \frac{\gamma^{-1/2}+\gamma^{1/2}(1-2\beta) - z + \sqrt{\Delta_{\gamma,\beta}(z)}}{2\beta\sqrt{\gamma} z},
                \rho_{\gamma,\beta}(z) = \frac{1+\gamma(1-2\beta) - \sqrt{\gamma}z + \sqrt{-\gamma \Delta_{\gamma,\beta}(z)}}{2\beta\gamma z},
    \end{equation}
    where $\Delta_{\gamma,\beta}(z)$ is defined in \eqref{eq:thm:1:density}.
    Consider the resolvent of the top-left $d$-by-$d$ block of \eqref{eq:QZ-basis-change}:
    \begin{equation}
        \label{eq:R(z)}
        \bR(z) = \left( \diag(\bmu) + \bX\bX^\T -z\bI \right)^{-1},\qquad\textrm{where}\quad z\notin \mathrm{spec}(\diag(\bmu) + \bX\bX^\T) .
    \end{equation}
    Note that by Theorems~\ref{thm:QZ-LSD}-\ref{thm:QZ:largest}, $d^{-1}\tr\bR(z)\to \rho_{\gamma,\beta}(z)$ a.s. for every $z\in (\NoiseEdge_{\gamma,\beta},\infty)$. 
    The following convergence result for the diagonal elements of the resolvent \eqref{eq:R(z)} will be useful in the sequel:
    \begin{lemma}\label{lem:Rho-i}
        Denote 
        \begin{equation}\label{eq:lem:Rho-i:1}
            \hat{\rho}^{(i)}_{\gamma,\beta}(z) = \left( \mu_i - z + \frac{\sqrt{\gamma}(1-\beta)}{1+\beta\sqrt{\gamma}{\rho_{\gamma,\beta}}(z)} \right)^{-1}.
        \end{equation}
        For all $z \in (\NoiseEdge_{\gamma,\beta}^+,\infty)$, a.s. as $n,m,d\to\infty$,
        \begin{equation}
            \max_{1\le i \le d} \left|  \bR(z)_{i,i}- \hat{\rho}^{(i)}_{\gamma,\beta}(z) \right| \longrightarrow 0.
        \end{equation}
        Furthermore, $d^{-1}\tr(\bR(z)) \longrightarrow \rho_{\gamma,\beta}(z)$, and the following relation holds:
        \begin{equation}\label{eq:lem:Rho-i:2}
            \rho_{\gamma,\beta}(z) = \MPStiel_{\gamma\beta,\gamma^{-1/2}}\left( z - \frac{\sqrt{\gamma}(1-\beta)}{1+\beta\sqrt{\gamma}{\rho_{\gamma,\beta}}(z)} \right),
        \end{equation}
        where $\MPStiel_{\gamma\beta,\gamma^{-1/2}}(\cdot)$ is the Stieltjes transform of the Marcheko-Pastur law with shape $\gamma\beta$ and scale $\gamma^{-1/2}$: 
        \begin{equation}\label{eq:lem:Rho-i:3}
            \MPStiel_{\gamma\beta,\gamma^{-1/2}}(\lambda) = \frac{1-\beta\gamma - \sqrt{\gamma}\lambda + 
            \sqrt{ \left(\sqrt{\gamma}\lambda -1 -\beta\gamma \right)^2 - 4\beta\gamma }
            }{2\gamma\beta \lambda}, \qquad \lambda>\frac{(1+\sqrt{\beta\gamma})^2}{\sqrt{\gamma}} .
        \end{equation}
    \end{lemma}
    The proof of Lemma~\ref{lem:Rho-i} follows by a computation technique which is classical in random matrix theory. For completeness, it appears in Appendix, Section~\ref{sec:proof-lem:Rho-i}. 

    \subsection{The outlying singular values}
    \label{sec:proof-singular-values}

    Theorems~\ref{thm:QZ-LSD}-\ref{thm:QZ:largest} characterize the behavior of the singular values of $\hbY$ in the absence of a signal. In the presence of a signal, they allow us to describe the behavior of the bulk singular values of $\hbY$. In this section, we study the behavior of the large singular values of $\hbY$, which are potentially outliers.

    Our argument follows a framework introduced by \cite{benaych2012singular}. Recall the representation \eqref{eq:PY:2} of $\hbY$ as a signal plus noise matrix. Define the $6r$-by-$6r$ matrix:
    	\begin{equation}\label{eq:M-hat}
		\widehat{\bM}(y) = \MatL 
			y\cdot \bA^\T (y^2\bI_{n\times n}-\bQc\bZ\bZ^\T\bQc)^{-1}\bA  & \bA^\T(y^2\bI_{n\times n}-\bQc\bZ\bZ^\T\bQc)^{-1}\bQc\bZ \bB \\
			\bB^\T \bZ^\T \bQc (y^2\bI_{n\times n}-\bQc\bZ\bZ^\T\bQc)^{-1}\bA  & y\cdot \bB^\T(y^2\bI_{p\times p}-\bZ^\T \bQc\bQc\bZ)^{-1}\bB  
		\MatR
		-
		\MatL \0 & (\bSigma^{-1})^\T  \\
		\bSigma^{-1} &\0 \MatR ,
	\end{equation}
 where $\bA,\bB,\bSigma$ are given in \eqref{eq:A-B-def}-\eqref{eq:bSigma-def}. Moreover, it is straightforward to verify that 
	\begin{equation}\label{eq:Sigma-Inv}
		\bSigma^{-1} = \MatL
		\bLambda^{-1} & \0 & \0 \\
		\0 &-\bLambda^{-1}	& \frac{\beta\sqrt{\gamma}}{1-\gamma\beta}\bI \\
		\0 &(1-\gamma\beta)\bI 	&\bLambda^{-1}
		\MatR 
      = \MatL
		\0 & \0 & \0 \\
		\0 &\0	& \frac{\beta\sqrt{\gamma}}{1-\gamma\beta}\bI \\
		\0 &(1-\gamma\beta)\bI 	&\0
		\MatR 
            + \bI^{-}_3 \otimes \bLambda^{-1},
	\end{equation}
 where $\bI^{-}_3$ is defined in \eqref{eq:H}. Above, $\otimes$ denotes the Kronecker (tensor) product:
\[
\MatL a &b \\ c &d \MatR \otimes \bA = \MatL a\bA & b\bA \\ c\bA & d\bA \MatR.
\]
 
 By \cite[Lemma 4.1]{benaych2012singular}, the singular values of $\bPc\bY$ which are not singular values of $\bQc$ are precisely the solutions of $\det(\widehat{\bM}(y))=0$; our goal, then, is to find the roots of this equation which are located outside the support of the bulk distribution, given in Theorem~\ref{thm:QZ-LSD}.

The next lemma is the main technical result of this section: it calculates a point-wise limit for the matrix $\widehat{\bM}(y)$. 

\begin{lemma}\label{lem:Mhat-Lim}
    Fix any $y > \sqrt{\NoiseEdge_{\gamma,\beta}^+}$. Then a.s.,
    \begin{equation}\label{eq:Mhat-Lim}
        \widehat{\bM}(y) \simeq \bMc_{\gamma,\beta}(y) \equiv \bKc_{\gamma,\beta}(y) \otimes \bI_{r\times r} - 
        \Hbb \otimes \bLambda^{-1} \,,
        % \MatL \0 &\0 &\0 &\bLambda^{-1} &\0 &\0 \\
        % \0 &\0 &\0 &\0 &-\bLambda^{-1} &\0 \\
        % \0 &\0 &\0 &\0 &\0 &\bLambda^{-1} \\ 
        % \bLambda^{-1} &\0 &\0 &\0 &\0 &\0\\
        % \0 &-\bLambda^{-1} &\0 &\0 &\0 &\0\\
        % \0 &\0 &\bLambda^{-1}&\0 &\0 &\0\MatR .
        % \MatL \0 &\bLambda_3 \\
        % \bLambda_3 &\0 \MatR\qquad\textrm{where}\quad \bLambda_3
        % = 
        % \MatL
        % \bLambda^{-1} &\0 &\0 \\
        % \0 & -\bLambda^{-1} &\0 \\
        % \0 &\0 & \bLambda^{-1} \\
        % \MatR.
    \end{equation} 
    where the matrix $\bKc_{\gamma,\beta}(y)$ is defined in \eqref{eq:K-func} and $\bHc$ in \eqref{eq:H}.
\end{lemma}
The proof of Lemma~\ref{lem:Mhat-Lim} appears in Appendix, Section~\ref{sec:proof-lem:Mhat-Lim}. The main task amounts to calculating limiting formulas for certain quadratic forms that involve the ``effective'' signal vectors $\bA,\bB$ and the pure-noise matrix $\bQc\bZ$. These quadratic forms are ultimately reduced to traces of corresponding compound matrix expressions involving $\bZ_1,\bZ_2$.
Unlike in \cite{benaych2012singular}, however, several of these mixed traces cannot simply be reduced to the Stieltjes transform (trace of the resolvent) of the noise matrix $\bQc\bZ$. This fact complicates things considerably, and it owes to the fact that the 
``signal'' part in the representation \eqref{eq:PY:2} (the matrices $\bA,\bB$) in fact depend on the noise part $\bQc\bZ$. In carrying out the necessary calculations, Lemma~\ref{lem:Rho-i} (among others) plays an important role.

\paragraph*{}
Note that while Lemma~\ref{lem:Mhat-Lim} deals with point-wise convergence, it is straightforward to ``upgrade'' it to uniform convergence on compact subsets (e.g. by Arzela-Ascoli, where both equicontinuity and equiboundedness are easily deduced from Theorem~\ref{thm:QZ:largest}) and complex arguments  $\Re(y)>\sqrt{\NoiseEdge_{\gamma,\beta}^{+}}$. Consequently, by elementary complex analysis, the set of roots $y$ of the random equation $\det(\hbM(y))=0$ converges to those of the deterministic equation $\det(\bMc_{\gamma,\beta}(y))=0$.\footnote{Convergence is in the following sense. Fix any compact interval $\m{I} \subset (\NoiseEdge_{\gamma,\beta}^+,\infty)$, and let $\hat{m}(\m{I}),m(\m{I})$ be, respectively, the number of roots in $\m{I}$ of the random, respectively determinstic, equation. Then $\hat{m}(\m{I})\longrightarrow m(\m{I})$ a.s. See also \cite{benaych2012singular}. }

We next study the roots of the deterministic equation $\det(\bMc_{\gamma,\beta}(y))=0$. 
It is easy to see from \eqref{eq:Mhat-Lim} that upon an appropriate permutation of the coordinates $\m{S}$, the  matrix $\bMc_{\gamma,\beta}(y)$ decomposes into a block diagonal matrix:\footnote{In other words, the operator $\bMc_{\gamma,\beta}(y):\RR^{6}\otimes \RR^{r}\to \RR^{6}\otimes \RR^{r}$ can be decomposed as a direct sum over its restrictions on $\RR^6\otimes \{\mathrm{span}(\bm{e}_i)\}_{1\le i\le r}$. }
\begin{equation}\label{eq:Values:Permuatation}
    \m{S}\bM_{\gamma,\beta}(y)\m{S}^\T = \bigoplus_{\ell=1}^r (\bKc_{\gamma,\beta}(y) - \sigma_\ell^{-1}\Hbb) .
\end{equation}
Accordingly, the determinant factors into a product:
\begin{equation}\label{eq:DetFactor}
    \det(\bMc_{\gamma,\beta}(y)) = \prod_{\ell=1}^r \det(\bKc_{\gamma,\beta}(y) - \sigma_\ell^{-1}\Hbb),
\end{equation}
and so the question of outliers decouples between different signal spikes. 

It remains to study the roots of the single-spike equation $\det(\bKc_{\gamma,\beta}(y) - \sigma^{-1}\Hbb)=0$. As described in Proposition~\ref{prop:PosEigFunc-properties}, when $y>\sqrt{\NoiseEdgeUpper}$ there is a unique number $\PosEigFunc(y)$ such that $\det(\bKc_{\gamma,\beta}(y) - \PosEigFunc(y)\Hbb)=0$. Accordingly, 
\begin{equation}\label{eq:aux-single-det}
    \det(\bKc_{\gamma,\beta}(y) - \sigma^{-1}\Hbb)=0\qquad
    \textrm{if and only if} \qquad
    \PosEigFunc(y)=1/\sigma \,.
\end{equation}
Recall furthermore that $y\mapsto \PosEigFunc(y)$ is decreasing, and maps the ray $(\sqrt{\NoiseEdgeUpper},\infty)$ bijectively to the interval $(1/\BBP,0)$, where $\BBP=1/\PosEigFunc(\sqrt{\NoiseEdgeUpper})$, as defined in \eqref{eq:BBP}. Thus, a solution $y>\sqrt{\NoiseEdgeUpper}$ to \eqref{eq:aux-single-det} exists if and only if $\sigma>\BBP$; if so, it is given by $y=\SpikeFunc(\sigma)$, where $\SpikeFunc(\sigma)=\PosEigFunc^{-1}(1/\sigma)$ is the spike forward function, as defined in \eqref{eq:SpikeFunc}.

The proof of Proposition~\ref{prop:PosEigFunc-properties} is deferred to the Appendix, Section~\ref{sec:proof-prop:PosEigFunc-properties}.

We are ready to conclude the proofs of Theorems~\ref{thm:1:Bulk} and \ref{thm:2:Outliers}:

\begin{proof}
    (Of Theorems~\ref{thm:1:Bulk} and \ref{thm:2:Outliers}.) By the preceding discussion, the set of large outlying singular values (exceeding the upper bulk edge $\sqrt{\NoiseEdge_{\gamma,\beta}^+}$) converges to the solution set of \eqref{eq:DetFactor} with $y\ge \sqrt{\NoiseEdge_{\gamma,\beta}^+}$. There are at most $r$ solutions: for each spike $1\le i\le r$, if $\sigma_i>\BBP$ then there is a solution $y_i=\SpikeFunc(\sigma_i)$, and if $\sigma\le \BBP$ then there is no solution. Accordingly, there are at most $r$ outliers, whose locations are given in \eqref{eq:thm:2:Outliers}.

    As for the remaining singular values, asymptotically they do not exceed the noise upper edge $\sqrt{\NoiseEdge_{\gamma,\beta}^+}$.
    By Weyl's inequality (e.g. \cite[Exercise 1.3.22]{tao2012topics}) applied to \eqref{eq:PY:2}, for all $3r<i<d-3r$, 
    \begin{equation}\label{eq:Weyl}
        \sigma_{i+3r}(\bQc\bZ) \le \sigma_i\left( \bA\bSigma\bB^\T + \bQc\bZ  \right) \le \sigma_{i-3r}(\bQc\bZ),
    \end{equation}
    where we used that $\rank(\bA\bSigma\bB^\T)\le 3r$ hence $\sigma_{3r+1}(\bA\bSigma\bB^\T)=0$. (See also \cite{benaych2012singular}.)  Theorems~\ref{thm:QZ-LSD}-\ref{thm:QZ:largest} imply that $\sigma_{i+3r}(\bQc)\to \sqrt{\NoiseEdge_{\gamma,\beta}^+}$ for every constant $i$; consequently, the leading non-outlier singular singular values of $\hbY$ must converge to the noise bulk edge. This establishes Theorem~\ref{thm:2:Outliers}.

    Finally, Theorem~\ref{thm:1:Bulk} follows from Theorem~\ref{thm:QZ-LSD} and the interlacing inequality \eqref{eq:Weyl}.

\end{proof}

\subsection{The outlying singular vectors}

Next, we aim to calculate the correlations between the observed and population spike directions.  

For brevity, denote $\hsigma_i = \sigma_i(\hbY)$, so that $\hsigma_1\ge \ldots \ge \hsigma_r$ are the $r$ largest observed singular values. Recall that we have derived a limiting expression for the $\hsigma_i$-s (Theorem~\ref{thm:2:Outliers}), which we denote for brevity $\hsigma_i \to y_i$. Furthermore,   
denote by $\hbu_1,\ldots,\hbu_r$ and $\hbv_1,\ldots,\hbv_r$ respectively the corresponding observed left and right singular vectors of $\hbY$.

\paragraph*{}
Fix any $i$ such that the corresponding spike is detectable: $\sigma_i>\BBP$, and so $y_i=\SpikeFunc(\sigma_i)$. By \cite[Lemma 5.1]{benaych2012singular}, the vector
\begin{equation}\label{eq:Vecs:Aux1}
	\hbf_i = \MatL
		\bSigma & \\
		 &\bSigma^\T 
	\MatR 
 \MatL \bB^\T \hbv_i \\  \bA^\T \hbu_i \MatR \in \RR^{6r},
\end{equation}
lies in the kernel of the matrix $\hbM(\hsigma_i)$ from \eqref{eq:M-hat}. It is convenient to permute the coordinates of \eqref{eq:Vecs:Aux1}, grouping together columns of $\bA,\bB$ that correspond to a single population spike, similarly to \eqref{eq:Values:Permuatation}. Specifically, consider the matrix 
% $\underline{\bSigma}(\sigma)$ from \eqref{eq:bSigma-single-def}, 
\begin{align}
\label{eq:bSigma-single-def}
    \underline{\bSigma}(\sigma) =
		\MatL
		\sigma &0 &0 \\
		0 &-f(\sigma) &\frac{\beta\sqrt{\gamma}}{1-\beta\gamma}g(\sigma) \\
		0 &(1-\beta\gamma)g(\sigma) &f(\sigma)
		\MatR,
\end{align}
so that 
$\m{S}\bSigma\m{S}^\T = \underline{\bSigma}(\sigma_1) \oplus \ldots \oplus \underline{\bSigma}(\sigma_r)$ is block diagonal. Define
\begin{equation}\label{eq:d-c-def}
	\hbd^{(\ell)}_i = 
	\MatL
		\underline{\bSigma}(\sigma_\ell) & \\
		 &\underline{\bSigma}(\sigma_\ell)^\T 
	\MatR
	\hbc^{(\ell)}_i,\qquad
	\hbc^{(\ell)}_i = 
	\MatL
		\langle \bm{v}_\ell , \hbv_i \rangle\\
		\langle [\bW_1]_{*,\ell} , \hbv_i \rangle\\
		\langle [\bW_2]_{*,\ell} , \hbv_i \rangle\\
		\langle \bm{u}_\ell , \hbu_i \rangle\\
		\langle \bQc^\perp \bm{u}_\ell , \hbu_i \rangle\\
		\langle [\bVbar_1]_{*,\ell} , \hbu_i\rangle
	\MatR,
\end{equation}
so that permuting the coordinates of \eqref{eq:Vecs:Aux1} results in $\m{S}\hbf_i = \hbd^{(1)}_i\oplus\ldots\oplus \hbd^{(r)}_i$. 

\begin{lemma}\label{lem:Decoupling}
    Fix any $i$ such that $\sigma_i>\BBP$. Then $\hbc^{(\ell)}_i \to \0$ for all $1\le \ell \le r$, $\ell\ne i$. 
    
    In particular, $\langle \bm{v}_\ell,\hbv_i \rangle\to 0$ and $\langle \bm{u}_\ell , \hbu_i \rangle\to 0$. 
\end{lemma}
\begin{proof}
    By the preceding discussion, $\hbd^{(1)}_i\oplus\ldots\oplus \hbd^{(r)}_i$ lies in the kernel of $\m{S}\hbM(\hsigma_i)\m{S}^\T$. By Lemma~\ref{lem:Mhat-Lim} and \eqref{eq:Values:Permuatation}, this matrix converges to $\bigoplus_{\ell=1}^r (\bKc_{\gamma,\beta}(y_i) - \sigma_\ell^{-1}\Hbb)$ where $y_i=\SpikeFunc(\sigma_i)$. By Proposition~\ref{prop:PosEigFunc-properties}, and the assumption that the signal spikes are distinct, the matrices $(\bKc_{\gamma,\beta}(y_i) - \sigma_\ell^{-1}\Hbb)$ are invertible for all $\ell\ne i$. Therefore $\hbd^{(\ell)}_i\to \0$ for $\ell\ne i$, and so does $\hbc^{(\ell)}_i\to \0$.
\end{proof}

It remains to study $\hbc^{(i)}_i$, which describes the correlation between the $i$-th signal and observed PCs. As mentioned above, the vector $\hbd^{(i)}_i$ lies asymptotically in the kernel of $(\bKc_{\gamma,\beta}(y_i) - \sigma_i^{-1}\Hbb)$. This matrix is not invertible, and moreover by Proposition~\ref{prop:PosEigFunc-properties} its kernel is $1$ dimensional. And so, it remains to find one additional equation satisfied by this vector.
 By \cite[Lemma 5.1]{benaych2012singular}, the following holds:
\begin{align}
	1 
	&= 
	\hsigma_i^2 \cdot \hbv_i^\T (\bA\bSigma\bB^\T)^\T \left(\hsigma_i^2\bI - \bQc\bZ\bZ^\T \bQc  \right)^{-2}   (\bA\bSigma\bB^\T) \hbv_i \nonumber \\
	&+ \hbu_i^\T (\bA\bSigma\bB^\T) (\bQc\bZ)^\T \left(\hsigma_i^2\bI - \bQc\bZ\bZ^\T \bQc  \right)^{-2} (\bQc\bZ)(\bA\bSigma\bB^\T)^\T \hbu_i \nonumber \\
	&+ 2\hsigma_i \cdot \hbv_i^\T (\bA\bSigma\bB^\T)^\T \left(\hsigma_i^2\bI - \bQc\bZ\bZ^\T \bQc  \right)^{-2} (\bQc\bZ)(\bA\bSigma\bB^\T)^\T\hbu_i \label{eq:Vecs:Aux2}\,.
\end{align}
Define the matrix 
\begin{equation}\label{eq:T-hat-Mat}
	\hbT(y) = \MatL
		y^2 \cdot \bA^\T \left(y^2\bI - \bQc\bZ\bZ^\T \bQc  \right)^{-2} \bA  
		& y \cdot \bA^\T  \left(y^2\bI - \bQc\bZ\bZ^\T \bQc  \right)^{-2} (\bQc\bZ) \bB \\
		y \cdot \bB^\T (\bQc\bZ)^\T \left(y^2\bI - \bQc\bZ\bZ^\T \bQc  \right)^{-2} \bA  
		& \bB^T (\bQc\bZ)^\T \left(y^2\bI - \bQc\bZ\bZ^\T \bQc  \right)^{-2} (\bQc\bZ) \bB 
	\MatR\,,
\end{equation}
so that \eqref{eq:Vecs:Aux2} can be rewritten as $\hbf_i^\T \hbT(\hsigma_i)\hbf = 1$, where $\hbf_i$ was defined in \eqref{eq:Vecs:Aux1}.

The following limit holds.
\begin{lemma}
\label{lem:T-lim}
    Fix any $i$ such that $\sigma_i>\BBP$. Then $\m{S}\hbT(\hsigma_i)\m{S}^\T \to \left(\bTc_{\gamma,\beta}(y_i)\right)^{\oplus r}$ where $y_i=\SpikeFunc(\sigma_i)$, $\m{S}$ is the permutation from \eqref{eq:Values:Permuatation} and $\bTc_{\gamma,\beta}(\cdot)$ is defined in \eqref{eq:T-func}.
\end{lemma}
The proof of Lemma~\ref{lem:T-lim} appears in Appendix, Section~\ref{sec:proof-lem:T-lim}.
We are ready to conclude the calculation of the limiting singular vector angles:
\begin{proof}
    (Of Theorem~\ref{thm:3:Angles}.)
    The first part of the theorem follows from Lemma~\ref{lem:Decoupling}. As for the second part, by Lemmas~\ref{lem:Mhat-Lim} and \ref{lem:T-lim}, 
    \[
        (\bKc_{\gamma,\beta}(y_i) - \sigma_i^{-1}\Hbb) \hbd^{(i)}_i \simeq \0, \qquad \langle \hbd^{(i)}_i, \bTc(y_i)\hbd^{(i)}_i\rangle \simeq 1 ,
    \]
    and thus any limit point (as $n,p,d\to\infty$) of $\hbd^{(i)}_i$ satisfies the above with equality. Note that the above system represent the intersection of a line and an ellipsoid, and therefore there are two solutions, which are antipodal points. Since $\hbd^{(i)}_i$ is clearly bounded a.s., every subsequence must converge to one of those two limit points; consequently, since they are antipodal, the absolute values of the entries, $|\hbd^{(i)}_i|$, converge.    
	Recall that by definition,
    \[
        |(\hbd^{(i)}_i)_1| = \sigma_i |\langle \bv_i,\hbv_i\rangle|,\qquad |(\hbd^{(i)}_i)_4| = \sigma_i |\langle \bu_i,\hbu_i\rangle| .
    \]
    
    To conclude, it remains to show that $\langle \bv_i,\hbv_i\rangle \langle \bu_i,\hbu_i\rangle$ is necessarily non-negative. Multiplying $\bY$ by $\hbu_i$ from the left and $\hbv_i$ from the right, and using that $\langle \bv_j,\hbv_i\rangle \langle \bu_j,\hbu_i\rangle\simeq 0$ for all $j\ne i$,
    \begin{align*}
    	\sigma_i(\bY) = \bu_i^\T\bY\hbv_i \simeq \sigma_i \langle \bv_i,\hbv_i\rangle \langle \bu_i,\hbu_i\rangle + \bu_i^\T\bZ\hbv_i \le  \sigma_i \langle \bv_i,\hbv_i\rangle \langle \bu_i,\hbu_i\rangle + \sigma_1(\bZ) \,.
    \end{align*}
	Since $\sigma_i>\BBP$, we have $\sigma_i(\bY)-\sigma_1(\bZ)\ge 0$ asymptotically a.s.; consequently, $\langle \bv_i,\hbv_i\rangle \langle \bu_i,\hbu_i\rangle\ge 0$.
\end{proof}

Finally, it remains prove Proposition~\ref{prop:vanishing-corrs}, namely to show that for ``barely detectable'' spikes, the correlation between the signal and observed singular vectors vanishes. This is done in Appendix, Section~\ref{sec:proof-prop:vanishing-corrs}.

%\section{Conclusion}
%
%\TODO{Not sure if there is actually need?}
%
%\TODO{Main take-home message: there IS a price to pay for using the fast randomized SVD instead of the full SVD. In statistical applications, especially when the SNR is low: should be very careful.
%
%Also: the spiked model is a natural framework for evaluating and comparing between different fast SVD algorithms. Offers very precise benchmarks to look at: the detection threshold, the PC angles.}

\subsection*{Acknowledgements}

I am grateful to Matan Gavish for introducing me to this research direction; his advice and encouragements were pivotal towards the completion of this manuscript. I warmly thank David Donoho for stimulating discussions on this work and for his thoughtful advice; and to Or Ordentlich for valuable suggestions.
    
    \bibliographystyle{alpha}
    \bibliography{refs}

    \appendix

    \section{Proof of Proposition~\ref{prop:PosEigFunc-properties}}
    \label{sec:proof-prop:PosEigFunc-properties}

    First, let us compute the roots $s$ of $\det(\bKc_{\gamma,\beta}(y)-s\bHc)=0$.  Define the following functions:
\begin{align}
    \delta(z) &= \sqrt{\gamma} z - 1 - \gamma,\\
    D(z) &= \delta(z) - \sqrt{ \delta(z)^2 - 4\beta\gamma(1+\gamma-\beta\gamma) }\,,
\end{align}
and the quadratic polynomial,
\begin{align}\label{eq:aux-P(q,A)}
    P(q;A) = 2\sqrt{\gamma}(1+\gamma-\beta\gamma)q^2 
    + (1+\gamma-\beta\gamma) \left( 2\beta\gamma-A \right)q
    - \sqrt{\gamma}A .
\end{align}
where $A$ is a free parameter. 
It is straightforward to verify that when $z>\NoiseEdge^+_{\gamma,\beta}$, the functions $\delta(z),D(z)$ are positive, with $\delta(z)$ increasing and $D(z)$ decreasing.

With the aid of a computer algebra system, one can show:
\begin{align}
    \det(\bKc_{\gamma,\beta}(y)-s\bHc) = - \frac{s^2+\beta\sqrt{\gamma}}{2\sqrt{\gamma}(1+\gamma-\beta\gamma)}P(s^2;D(y^2)).
\end{align}
We immediate identify two imaginary roots $s = \pm \beta^{1/2}\gamma^{1/4} \iu$; the real roots (if there are any) are solutions of $P(s^2;D(y^2))=0$. 

For brevity, denote $a\equiv 2\sqrt{\gamma}(1+\gamma-\beta\gamma)$, $b\equiv (1+\gamma-\beta\gamma) \left( 2\beta\gamma-A \right)$, $c\equiv - \sqrt{\gamma}A$, where we set $q=s^2,A=D(y^2)$. The roots of the quadratic $P(q;A)=0$ are 
\[
q_{\pm} = \frac{-b \pm \sqrt{b^2 - 4ac}}{2a}.    
\]
Since $a>0,c<0$, the root $q_-$ is negative (hence the corresponding roots $s=\pm \iu\sqrt{-q_-}$ are pure imaginary) and $q_+$ is positive, yielding two real roots $s=\pm \sqrt{q_+}$. In particular, $\PosEigFunc(y)=\sqrt{q_+}$. 

Lastly, it remains to show that $\PosEigFunc(y)$ is decreasing in $y$. Since $D(y^2)$ is decreasing in $y$, it suffices to show that $q_+$ is increasing in the variable $A$. For convenience, denote $\theta_1 = 2\beta\gamma$, $\theta_2=\frac{8\gamma}{1+\gamma-\beta\gamma}$ and perform a change of variables $B=A-\theta_1$, so that 
\begin{align*}
    q_+ = \frac{1+\gamma-\beta\gamma}{2a} \left( B + \sqrt{B^2 + \theta_2 B + \theta_1\theta_2} \right).
\end{align*}
Consequently, $q_+$ is increasing in $B$ if and only if $\frac{d}{dB}\sqrt{B^2 + \theta_2 B + \theta_1\theta_2} = \frac{B + \frac12\theta_2}{\sqrt{B^2 + \theta_2 B + \theta_1\theta_2}}> -1$. Since the denominator is non-negative, 
the only way this could be violated is if the numerator is negative and has magnitude $(B + \frac12\theta_2)^2 \ge B^2 + \theta_2 B + \theta_1\theta_2$. 
That is, for $q_+$ to be \emph{non-increasing} at a point  $B$, we need both
\begin{align*}
    B < -\frac12\theta_2 \qquad \textrm{and}\qquad \theta_1 \le \frac14 \theta_2 .
\end{align*}
Recall that $A=D(z)$ is decreasing in $z$, and so $A>D(\NoiseEdgeUpper)=0$ and $B=A-\theta_1>-\theta_1$. 
Thus, if a violating $B$ exists, then $-\theta_1 < -\frac12\theta_2$ equivalently $\theta_1\ge \frac12\theta_2$. But since $\theta_2>0$, this is inconsistent with the requirement $\theta_1\le \frac14\theta_2$. It follows that $q_+$ must be increasing in $B$.

We have established that $y\mapsto \PosEigFunc(y)$ is decreasing, hence invertible. Since $D(z)\to 0$ as $z\to\infty$, it is easy to verify that $\PosEigFunc(y)\to 0$ as $y\to\infty$. As for the behavior at $y=\sqrt{\NoiseEdge_{\gamma,\beta}^+}$, note that 
$D(\NoiseEdge^+_{\gamma,\beta})= 2\sqrt{\beta\gamma(1+\gamma-\beta\gamma)}$. A straightforward calculation gives the claimed expression \eqref{eq:BBP-squared}.

Finally, the formula for the inverse function $\PosEigFunc^{-1}(\cdot)$, \eqref{eq:PosEigFunc-inv}, can be obtained by straightforward calculation, e.g. starting from \eqref{eq:aux-P(q,A)}. We omit the details. 

% this holds if an only if the numerator is positive and larger than the denominator. That is, we need 
% \begin{align*}
%     B > -\frac12\theta_2 \qquad \textrm{and}\qquad \theta_1 < \frac14 \theta_2 .
% \end{align*}
% Recall that $A=D(z)$ is decreasing in $z$, and so $A>\inf_{z>\NoiseEdge_{\gamma,\beta}^+}D(z)=0$ and $B=A-\theta_1>-\theta_1$. Consequently, it suffices to just show $\theta_1<\frac14 \theta_2$. This holds, by definition, if and only if $\beta\gamma < \frac{\gamma}{1+\gamma-\beta\gamma}$, equivalently $\beta(1+\gamma-\beta\gamma)<1$. Recall that, by assumption, $\beta\gamma<1$ and also $\beta\le 1$.   

    \section{Proof of Lemma~\ref{lem:Projection}}
    \label{sec:proof-lem:Projection}

    Explicitly writing $\bYtilde=\bU\bLambda\bV_1^\T + \bZ_1$ and using the Sherman-Morrison-Woodbury formula,
    	\begin{align}
    		(\bYtilde^\T \bYtilde)^{-1} 
    		&= \left( \bZ_1^\T \bZ_1 + 
    		\begin{bmatrix}
    			\bV_1 & \bZ_1^\T \bU 
    		\end{bmatrix}  
    		\begin{bmatrix}
    			\bLambda^2 & \bLambda \\
    			\bLambda & \0
    		\end{bmatrix}
    		\begin{bmatrix}
    			\bV_1^\T  \\
    			\bU^\T \bZ_1  
    		\end{bmatrix}
    		\right)^{-1} \nonumber \\
    		&= (\bZ_1^\T \bZ_1)^{-1}
    		-
    		(\bZ_1^\T \bZ_1)^{-1}
    		\begin{bmatrix}
    			\bV_1 & \bZ_1^\T \bU 
    		\end{bmatrix}  
    		\tilde{\bPhi}
    		\begin{bmatrix}
    			\bV_1^\T  \\
    			\bU^\T \bZ_1  
    		\end{bmatrix}
    		(\bZ_1^\T \bZ_1)^{-1}
      \label{eq:Y-tilde-inv}
    	\end{align}
    	where
    	\begin{align*}
    		\tilde{\bPhi} 
    		&= 
    		\left(
    			\begin{bmatrix}
    				\0 & \bLambda^{-1} \\
    				\bLambda^{-1} & -\bI
    			\end{bmatrix}
    			+	
    			\begin{bmatrix}
    				\bV_1^\T  \\
    				\bU^\T \bZ_1  
    			\end{bmatrix}
    			(\bZ_1^\T \bZ_1)^{-1}
    			\begin{bmatrix}
    				\bV_1 & \bZ_1^\T \bU 
    			\end{bmatrix} 
    			\right)^{-1} \\
            &=
    			\begin{bmatrix}
    				\bV_1^\T (\bZ_1^\T\bZ_1)^{-1} \bV_1  & \bV_1^\T (\bZ_1^\T\bZ_1)^{-1} \bZ_1^\T \bU  + \bLambda^{-1} \\
    				\bU \bZ_1 (\bZ_1^\T\bZ_1)^{-1}\bV_1 \bLambda^{-1} & \bU^\T \bZ_1 (\bZ_1^\T\bZ_1)^{-1} \bZ_1^\T \bU_1  -\bI
    			\end{bmatrix}^{-1} \,.
    	\end{align*}
        The matrix $\bZ_1\in \RR^{n\times d}$ has i.i.d. centered Gaussian entries (with variance $1/\sqrt{nm})$, and so its distribution is invariant to multiplication by an orthogonal matrix on either side. Moreover, $\sqrt{\frac{m}{n}}\bZ_1^\T\bZ_1$ being a standard Wishart matrix in $d$ dimensions and $n$ degrees of freedom (where $d/n\simeq \beta\gamma<1$), it is well know that $\lambda_{d}(\bZ_1^\T \bZ_1) \overset{a.s.}{\longrightarrow} \gamma^{-1/2}(1-\sqrt{\beta\gamma})^2$, see \cite{bai2010spectral}, and in particular $\|(\bZ_1^\T\bZ_1)^{-1}\|$ is bounded. 
        Further recall that the columns of $\bU$ are orthonormal, and the columns of $\bV_1$ are asymptotically orthogonal with norm $\simeq \sqrt{\beta}$ (recall \eqref{eq:assum:incoherence}). By standard concentration inequalities for quadratic forms, e.g. the Hanson-Wright inequality (see for example \cite{vershynin2018high}), 
        \begin{align*}
            \bV_1^\T (\bZ_1^\T\bZ_1)^{-1} \bV_1 &\simeq \frac{\beta}{d}\tr (\bZ_1^\T\bZ_1)^{-1} \otimes \bI_{r\times r} \overset{(\star)}{\simeq} \frac{\beta\sqrt{\gamma}}{1-\beta\gamma} \otimes \bI_{r\times r},\\
            \bU^\T \bZ_1 (\bZ_1^\T\bZ_1)^{-1} \bZ_1^\T \bU_1 &\simeq \frac1n \tr \left(  \bZ_1 (\bZ_1^\T\bZ_1)^{-1} \bZ_1^\T\right) \otimes \bI_{r\times r} = \frac1n\tr\bI_{d\times d} \otimes \bI_{r\times r} \simeq \beta\gamma \otimes \bI_{r\times r} 
        \end{align*}
        and $\bV_1^\T (\bZ_1^\T\bZ_1)^{-1} \bZ_1^\T \bU \simeq \bm{0} $. Above, ($\star$) follows since for a Wishart matrix $\bW$ with dimension $d$ and $n$ d.o.f's, $d^{-1}\tr\bW^{-1}\simeq \frac{1}{1-d/n}$. 
        Thus, $\tilde{\bPhi}\simeq \bPhi$, where
        % Plugging these into the expression for $\bPhi$ from above yields
        \begin{align}
            \bPhi 
            =
			\begin{bmatrix}
				\frac{\beta\sqrt{\gamma}}{1-\beta\gamma}\bI & \bLambda^{-1} \\
				\bLambda^{-1} & -(1-\beta\gamma)\bI
			\end{bmatrix}^{-1}  
		= \MatL (1-\gamma\beta)g(\bLambda) & f(\bLambda) \\ f(\bLambda) & -\frac{\beta\sqrt{\gamma}}{1-\gamma\beta}g(\bLambda) \MatR,
        \label{eq:Phi-def}
        \end{align}
        where $f(\cdot),g(\cdot)$ are defined in \eqref{eq:f-g}. Recall also the notation $\bVbar_1 = \bZ_1(\bZ_1^\T\bZ_1)^{-1}\bV_1$ from \eqref{eq:Vbar-1-def}.
        
        We have computed  an expression for $(\bYtilde^\T \bYtilde)^{-1}$, \eqref{eq:Y-tilde-inv}; our goal is to calculate need $\bYtilde(\bYtilde^\T \bYtilde)^{-1}\bYtilde^\T$. Again using $\bYtilde=\bU\bLambda\bV_1^\T+\bZ_1$, we simplify the resulting expression term-by-term. First, 
    	\begin{equation}\label{eq:proj:aux-1}
    		\bZ_1(\bYtilde^\T \bYtilde)^{-1}\bZ_1^\T \simeq
    		\bQc 
    		-
    		\begin{bmatrix}
    			\bVbar_1 & \bQc \bU 
    		\end{bmatrix}  
    		\bPhi
    		\begin{bmatrix}
    			\bVbar_1^\T  \\
    			(\bQc \bU)^\T  
    		\end{bmatrix} .
    	\end{equation}
    	Next,
    	\begin{align}
    		\bZ_1(\bYtilde^\T\bYtilde)^{-1}\bV_1\bLambda\bU^\T 
    		&\simeq
    		\bVbar_1\bLambda\bU^\T 
    		-
    		\begin{bmatrix}
    			\bVbar_1 & \bQc \bU 
    		\end{bmatrix}  
    		\bPhi 
    		\begin{bmatrix}
    			\bV_1^\T (\bZ_1\bZ_1)^{-1} \bV_1 \\
    			(\bQc \bU)^\T \bV_1 
    		\end{bmatrix}
    		\bLambda \bU^\T \nonumber \\
    		&\simeq 
    		\bVbar_1\bLambda\bU^\T 
    		-
    		\begin{bmatrix}
    			\bVbar_1 & \bQc \bU 
    		\end{bmatrix}  
    		\bPhi 
    		\begin{bmatrix}
    			\frac{\beta\sqrt{\gamma}}{1-\beta\gamma}\bI\\
    			\0 
    		\end{bmatrix}
    		\bLambda \bU^\T \nonumber \\
    		&= \bVbar_1 \bLambda \left[ \bI - \frac{\beta\sqrt{\gamma}}{1-\beta\gamma}\bPhi_{11} \right] \bU^\T - \frac{\beta\sqrt{\gamma}}{1-\beta\gamma}\bQc\bU\bPhi_{21}\bLambda\bU^\T ,\label{eq:proj:aux-2}
    	\end{align}
    	where $\bPhi_{1,1}$ and $\bPhi_{2,1}$ are, respectively, the top and bottom left blocks of $\bPhi$ in \eqref{eq:Phi-def}. Similarly,
    	\begin{align}
    		(\bU\bLambda\bV_1^\T)(\bYtilde^\T\bYtilde)^{-1}\bZ_1^\T 
    		\simeq
    		\bU \bLambda \left[ \bI - \frac{\beta\sqrt{\gamma}}{1-\beta\gamma}\bPhi_{11} \right] \bVbar_1^\T - \frac{\beta\sqrt{\gamma}}{1-\beta\gamma}\bU\bPhi_{21}\bLambda(\bQc\bU)^\T .
      \label{eq:proj:aux-3}
    	\end{align}
    	Lastly,
    	\begin{align}
    		(\bU\bLambda\bV_1^\T)(\bYtilde^\T\bYtilde)^{-1}(\bU\bLambda\bV_1^\T)^\T 
    		\simeq \frac{\beta\sqrt{\gamma}}{1-\beta\gamma}\bU\bLambda^2\bU^\T - \left(\frac{\beta\sqrt{\gamma}}{1-\beta\gamma}\right)^2\bU\bLambda\bPhi_{11}\bLambda\bU^\T  \,.
      \label{eq:proj:aux-4}
    	\end{align}
    	Finally, we decompose $\bU=\bQc\bU + \bQc^\perp \bU$, and collect all the terms in \eqref{eq:proj:aux-1}-\eqref{eq:proj:aux-4} above. 
        Straightforward (if tedious) algebra yields the claimed expression \eqref{eq:Projector}.
        \qed
     % A very tedious (if straightforward) calculation yields 
    	% \begin{equation}\label{eq:Projector}
    	% 	\begin{split}
    	% 		\bPc 
    	% 		&= \bQc \\
    	% 		&-(1-\beta\gamma) \bVbar_1\bLambda^2(\bI + \beta\bLambda^2)^{-1}\bVbar_1^\T \\
    	% 		&+ \bVbar_1 \bLambda(\bI + \beta\bLambda^2)^{-1}(\bQc^\perp\bU)^\T
    	% 		+  (\bQc^\perp\bU)\bLambda(\bI + \beta\bLambda^2)^{-1}\bVbar_1^\T \\
    	% 		&+ (\bQc^\perp\bU) \left[ \frac{\beta}{1-\beta\gamma}\bLambda^2(1+\beta\bLambda^2)^{-1} \right](\bQc^\perp\bU)^\T \,.
    	% 	\end{split}
    	% \end{equation}

    \section{Proof of Lemma~\ref{lem:PY:2}}
    \label{sec:proof-lem:PY:2}

    We would like to derive an approximate formula for $\bPc\bY=\bPc\bU\bLambda\bV^\T + \bPc\bZ$, where $\bPc$ is approximated by \eqref{eq:Projector}.

    We start with $\bPc\bU\bLambda\bV^\T$. 
	Note that $\bU^\T (\bQc^\perp \bU) \simeq \frac{n-d}{n}\bI_{r\times r}\simeq (1-\beta\gamma)\bI_{r\times r}$ and $\bVbar_1^\T \bU\simeq \bm{0}$. Thus,
	\begin{equation}
		\begin{split}
			\bPc \bU\bLambda\bV^\T 
			&\simeq 
			(\bQc\bU)\bLambda\bV^\T + \beta\sqrt{\gamma} (\bQc^\perp\bU) g(\bLambda)\bLambda\bV^\T  + (1-\beta\gamma)\bVbar_1 g(\bLambda)\bV^\T \\
			&= \bU\bLambda\bV^\T - 
			(\bQc^\perp\bU) f(\bLambda)\bV^\T  + (1-\beta\gamma)\bVbar_1 g(\bLambda)\bV^\T \,,
		\end{split}
	\end{equation}
	where the second equality is obtained by adding and substracting $(\bQc\bU)\bLambda\bV^\T$. 
	% \begin{equation}
	% 	\begin{split}
	% 		\bPc(\bU\bLambda\bV^\T) 
	% 		&\simeq (\bQc\bU)\bLambda\bV^\T  
	% 		+ (1-\beta\gamma)\bVbar_1(\Lambda^2(\bI+\beta\bLambda^2)^{-1})\bV^\T \\
	% 		&+ \beta(\bQc^\perp\bU)(\bLambda^3(\bI+\beta\bLambda^2)^{-1})\bV^\T \\
	% 		&= \bU\bLambda\bV^\T  
	% 		+ (1-\beta\gamma)\bVbar_1(\Lambda^2(\bI+\beta\bLambda^2)^{-1})\bV^\T \\
	% 		&- (\bQc^\perp\bU)(\bLambda(\bI+\beta\bLambda^2)^{-1})\bV^\T \,.
	% 	\end{split}
	% \end{equation}
	% Splitting $\bZ=[\bZ_1, \bZ_2]$, $\bV=[\bV_1,\bV_2]$,
	% \begin{equation}
	% 	\begin{split}
	% 		\bPc\bZ 
	% 		&= \bQc \bZ  
	% 		- (1-\beta\gamma)\bVbar_1(\bLambda^2(\bI+\beta\bLambda^2)^{-1})\MatL \bV_1^\T & \bVbar_1^\T \bZ_2\MatR \\
	% 		&+ \bVbar_1 (\bLambda(\bI+\beta\bLambda^2)^{-1})\MatL\0& (\bQc^\perp\bU)^\T \bZ_2\MatR \\
	% 		&+ (\bQc^\perp \bU) (\bLambda(\bI+\beta\bLambda^2)^{-1})\MatL \bV_1^\T & \bVbar_1^\T \bZ_2\MatR \\
	% 		&+ \frac{\beta}{1-\beta\gamma} (\bQc^\perp \bU)(\bLambda^2(\bI+\beta\bLambda^2)^{-1}) \MatL \0 & (\bQc^\perp\bU)^\T \bZ_2\MatR \,.
	% 	\end{split}
	% \end{equation}

	We next compute $\bPc \bZ$. Again, starting from (\ref{eq:Projector}),
	\begin{equation}\label{eq:PZ}
		\begin{split}
			\bPc \bZ
			&\simeq \bQc\bZ 
			-(1-\beta\gamma) \bVbar_1g(\bLambda)(\bZ^\T \bVbar_1)^\T + \frac{\beta\sqrt{\gamma}}{1-\beta\gamma}(\bQc^\perp\bU) g(\bLambda)(\bZ^\T \bQc^\perp\bU)^\T  \\
			&+ \bVbar_1 f(\bLambda)(\bZ^\T \bQc^\perp\bU)^\T
			+  (\bQc^\perp\bU)f(\bLambda)(\bZ^\T \bVbar_1)^\T \,.
		\end{split}
	\end{equation}
	One may verify that 
	\begin{align*}
		\bZ^\T \bQc^\perp\bU = \MatL \0 \\ \bZ_2^\T \bQc^\perp \bU \MatR\,,
		\quad
		\bZ^\T \bVbar_1 = \MatL \bV_1 \\ \bZ_2^\T \bVbar_1 \MatR .
	\end{align*}
	(Recall that $\bQc^\perp$ projects onto the orthogonal complement of the columns space of $\bZ_1$.) Thus,
	\begin{equation}\label{eq:PY:1}
		\begin{split}
			\bPc\bY 
			&\simeq \bU\bLambda\bV^\T + \bQc\bZ  \\
			&+ (1-\beta\gamma) \bVbar_1g(\bLambda)\MatL \0\\\bV_2-\bZ_2^\T \bVbar_1 \MatR^\T 
			- 
			(\bQc^\perp\bU) f(\bLambda) \MatL \0\\\bV_2-\bZ_2^\T \bVbar_1 \MatR^\T \\
			&+ \frac{\beta\sqrt{\gamma}}{1-\beta\gamma}(\bQc^\perp\bU) g(\bLambda)\MatL \0 \\ \bZ_2^\T \bQc^\perp \bU \MatR^\T  +
			\bVbar_1 f(\bLambda)\MatL \0 \\ \bZ_2^\T \bQc^\perp \bU \MatR^\T .
		\end{split}
	\end{equation}
    Let $\bA,\bB,\bSigma$ be as in \eqref{eq:A-B-def}-\eqref{eq:bSigma-def}. Adding up \eqref{eq:PZ} and \eqref{eq:PY:1} readily yields \eqref{eq:PY:2}.
    \qed

    \section{Proof of Lemma~\ref{lem:Rho-i}}
    \label{sec:proof-lem:Rho-i}

    The following calculation is classical within random matrix theory \cite{bai2010spectral}. We provide the full details for the sake of accessibility and completeness.

    Define the empirical Stieltjes transform:
    \begin{equation}
        \hat{\rho}(z) = d^{-1}\tr\bR(z).
    \end{equation}
    Denote also
    \begin{equation}
        \hat{\zeta}(z) = d^{-1}\tr(\bR(z)\diag(\bmu)) .
    \end{equation}
    (To lighten the notation, we always omit the subscripts $\gamma,\beta$.)

    For any $1\le i \le d$, denote by $\bP_i \in \RR^{(d-1)\times d}$ the projection matrix onto the coordinate set $[d]\setminus \{i\}$. The resolvent $\bR(z)$ decomposes as follows into blocks, up to a coordinate permutation,
	\begin{equation}\label{eq:Resolvement:Aux:1}
		\bR(z) = \MatL  
			\mu_i+(\bX\bX^\T)_{i,i} -z & \be_i^\T (\bX\bX^\T)\bP_i^\T \\
			\bP_i (\bX\bX^\T)\be_i & \bP_i( \diag(\bmu) + \bX\bX^\T - z\bI)\bP_i^\T 
		\MatR^{-1} \,.
	\end{equation}
	Applying the block matrix inversion formula in \eqref{eq:Resolvement:Aux:1},
	\begin{align}\label{eq:Resolvement:Aux:2}
		 \bR(z)_{i,i}  = \left(\mu_i+ (\bX\bX^\T)_{i,i} -z - \be_i^\T (\bX\bX^\T)\bP_i^\T \left( \bP_i(  \diag(\bmu) + \bX\bX^\T - z\bI)\bP_i^\T\right)^{-1} \bP_i (\bX\bX^\T)\be_i \right)^{-1}.
	\end{align}
    
    Recall that $\bX\in \RR^{d\times (m-d)}$ has Gaussian i.i.d. entries $\m{N}(0,1/\sqrt{nm})$. Consequently, $(\bX\bX^T)_{i,i}$ is a (scaled) $\chi$-squared random variable with $m-d$ degrees of freedom, and so concentrates around its expectation $\frac{m-d}{\sqrt{nm}}\simeq \sqrt{\gamma}(1-\beta)$. Standard $\chi$-squared tail bounds allow one to further deduce that a.s. $\max_{1\le i \le d}|(\bX\bX^T)_{i,i} - \sqrt{\gamma}(1-\beta)|\to 0$ as $n,m,d\to\infty$. 

    Furthermore, $\RR^{m-d}\ni \bX^\T \be_i \sim \m{N}(0,(nm)^{-1/2}\bI)$ and is independent $\bP_i\bX_i$. Moreover, the operator norm
    \[
        \| \bX^\T \bP_i^\T \left( \bP_i(  \diag(\bmu) + \bX\bX^\T - z\bI)\bP_i^\T\right)^{-1} \bP_i\bX \|    
    \] 
    is asymptotically a.s. bounded by a constant; this is because $z>\NoiseEdgeUpper \equiv \lim_{n,m,d\to\infty}\lambda_1(\diag(\bmu)+\bX\bX^\T)$ (recall Theorem~\ref{thm:QZ:largest}). By standard concentration inequalities for quadratic forms (e.g. Hanson-Wright, see for example \cite{vershynin2018high}), 
    \begin{align}
        &\max_{1\le i \le d}\left| \be_i^\T (\bX\bX^\T)\bP_i^\T \left( \bP_i(  \diag(\bmu) + \bX\bX^\T - z\bI)\bP_i^\T\right)^{-1} \bP_i (\bX\bX^\T)\be_i 
        \right. \nonumber\\
        &\qquad\qquad \left.- \frac{1}{\sqrt{nm}}\tr\left( 
 \bX^\T \bP_i^\T \left( \bP_i(  \diag(\bmu) + \bX\bX^\T - z\bI)\bP_i^\T\right)^{-1} \bP_i\bX\right)\right| \longrightarrow 0.
 \label{eq:Resolvement:Aux:3}
    \end{align}
    The average trace in \eqref{eq:Resolvement:Aux:3} involves the matrix $\bP_i$, whose effect is simply the removal of one row. It should be intuitively clear that any single row should only have a negligible influence on the whole trace. One can show (see e.g. \cite[Lemma 20]{gavish2022matrix}) it may indeed be omitted,
        \begin{align}
        &\max_{1\le i \le d}\left|\frac{1}{\sqrt{nm}}\tr\left( 
 \bX^\T \bP_i^\T \left( \bP_i(  \diag(\bmu) + \bX\bX^\T - z\bI)\bP_i^\T\right)^{-1} \bP_i\bX\right) 
        \right. \nonumber\\
        &\qquad\qquad \left.- \frac{1}{\sqrt{nm}}\tr\left( 
 \bX^\T  \left(   \diag(\bmu) + \bX\bX^\T - z\bI\right)^{-1} \bX\right)\right| \longrightarrow 0.
 \label{eq:Resolvement:Aux:4}
    \end{align}
    The bottom term in \eqref{eq:Resolvement:Aux:4} can be written as 
\begin{equation}\label{eq:Resolvement:Aux:5}
    \frac{1}{\sqrt{nm}}\tr\left(  \bR(z) \bX\bX^\T \right) = \frac{d}{\sqrt{nm}}\left( 1  -\hat{\zeta}(z) + z\hat{\rho}(z) \right) \simeq \beta\sqrt{\gamma} \left( 1  -\hat{\zeta}(z) + z\hat{\rho}(z) \right).
\end{equation}
Combining with \eqref{eq:Resolvement:Aux:2}, we deduce that $\max_{1\le i \le d} \left|\bR(z)_{i,i} 
 -\hat{\rho}_i(z) \right|\to 0$, where
\begin{align}\label{eq:Resolvement:Aux:6}
      \hat{\rho}_i(z) = \left(\mu_i+ \sqrt{\gamma}(1-2\beta) -z + \beta\sqrt{\gamma}\left(  \hat{\zeta}(z) - z\hat{\rho}(z) \right) \right)^{-1}.
\end{align}
Recall that by Theorems~\ref{thm:QZ-LSD}-\ref{thm:QZ:largest}, we know that $\hat{\rho}(z)\to \rho(z)$. It therefore remains to find a limiting formula for $\hat{\zeta}(z)$, which we now do. The convergence statement above implies that $\max_{1\le i \le d}\left| \bR_(z)_{i,i}(\hat{\rho}_i(z))^{-1} - 1 \right| \to 0$. Note that we can freely take the inverse $z>\NoiseEdge_+$ and therefore asymptotically almost surely, all the eigenvalues of $\bR(z)$ are negative and bounded by a constant; consequently, similarly must be all its diagonal elements. Thus,
\begin{align}
    1 
    &\simeq d^{-1}\sum_{i=1}^d \bR(z)_{i,i}(\hat{\rho}_i(z))^{-1}  \nonumber \\
    &= d^{-1}\sum_{i=1}^d \bR(z)_{i,i} \left(\mu_i+ \sqrt{\gamma}(1-2\beta) -z + \beta\sqrt{\gamma}\left(  \hat{\zeta}(z) - z\hat{\rho}(z) \right) \right) \nonumber \\
    &= \hat{\zeta}(z) \left( 1 + \beta\sqrt{\gamma}\hat{\rho}(z) \right) + \left( \sqrt{\gamma}(1-2\beta) -z \right)\hat{\rho}(z) - \beta\sqrt{\gamma}z(\hat{\rho}(z))^2,\label{eq:Resolvement:Aux:7}
\end{align}
where we used $\hat{\rho}(z) = d^{-1}\sum_{i=1}^d\bR(z)_{i,i}$, $\hat{\zeta}(z)=d^{-1}\sum_{i=1}^d \mu_i \bR(z)_{i,i} $, which hold by definition. 
Solving Eq. \eqref{eq:Resolvement:Aux:7} for $\hat{\zeta}(z)$ gives:
\begin{equation}
    \hat{\zeta}(z) \simeq z\hat{\rho}(z) + \frac{1-\sqrt{\gamma}(1-2\beta)\hat{\rho}(z)}{1+\beta\sqrt{\gamma}\hat{\rho}(z)}.
\end{equation}
Plugging this expression into \eqref{eq:Resolvement:Aux:6} and using $\hat{\rho}(z)\simeq \rho(z)$ yields the desired Eq. \eqref{eq:lem:Rho-i:1}. 

Finally, to verify \eqref{eq:lem:Rho-i:2}, write using \eqref{eq:lem:Rho-i:1}
\begin{align*}
    \rho(z) \simeq d^{-1}\sum_{i=1}^d \bR(z)_{i,i} \simeq d^{-1} \sum_{i=1}^d \left( \mu_i - z + \frac{\sqrt{\gamma}(1-\beta)}{1+\beta\sqrt{\gamma}{\rho_{\gamma,\beta}}(z)} \right)^{-1} .
\end{align*}
Recalling that the counting measure of $\mu_1,\ldots,\mu_d$ converges weakly to a Marchenko-Pastur law with shape $d/n\simeq \beta\gamma$ and scale $\gamma^{-1/2}$, we deduce \eqref{eq:lem:Rho-i:2}.
Note that formula \eqref{eq:lem:Rho-i:3} is only  valid for arguments larger than the upper edge. To see that this is indeed the case, recall that $z>\NoiseEdge_{\gamma,\beta}$, so by Theorem~\ref{thm:QZ:largest} all the eigenvalues of the resolvent, and thereby the diagonal elements, are upper bounded by a negative constant. Accordingly, $\mu_i - z + \frac{\sqrt{\gamma}(1-\beta)}{1+\beta\sqrt{\gamma}{\rho_{\gamma,\beta}}(z)}$ is asymptotically a.s. negative for all $i$; this holds, in particular for $i=1$, wherein $\mu_i$ converges to the upper edge of the corresponding Marchenko-Pastur law. 
\qed

\section{Proof of Lemma~\ref{lem:Mhat-Lim}}
\label{sec:proof-lem:Mhat-Lim}

% We will work with the matrix $\widehat{\bM}(y)$ directly, and rearrange its columns according to $\m{S}$ only at the very end.

Since the calculation is somewhat long, we will devote a separate subsection to each one of the blocks of $\hbM(y)$. Denote
\begin{align}
    \hbM_{1,1}(y) &= y\cdot \bA^\T (y^2\bI_{n\times n}-\bQc\bZ\bZ^\T\bQc)^{-1}\bA, \label{eq:hbM-11}\\
    \hbM_{1,2}(y) &= \bA^\T(y^2\bI_{n\times n}-\bQc\bZ\bZ^\T\bQc)^{-1}\bQc\bZ \bB, \label{eq:hbM-12} \\
    \hbM_{2,2}(y) &= y\cdot \bB^\T(y^2\bI_{m\times m}-\bZ^\T \bQc\bQc\bZ)^{-1}\bB, \label{eq:hbM-22}
\end{align}
each block being a $3r$-by-$3r$ matrix, and $\bA,\bB$ are defined in \eqref{eq:A-B-def}. Thus, \eqref{eq:M-hat} reads
\begin{equation}\label{eq:lem:Mhat-Lim-preproof}
        \hbM(y) = \MatL 
			\hbM_{1,1}(y)& \hbM_{1,2}(y) \\ 
			\hbM_{1,2}(y)^\T  &   \hbM_{2,2}(y)
		\MatR
  -
  \MatL \0 & (\bSigma^{-1})^\T  \\
		\bSigma^{-1} &\0 \MatR .
\end{equation}

 \subsection{The top left block}

 Consider the block $\hbM_{1,1}(y)$ in \eqref{eq:hbM-11}. For brevity, denote 
 \begin{equation}
     \barR(z) = \left( \bQc\bZ\bZ^\T \bQc - z\bI \right)^{-1} ,
 \end{equation}
 so that $\hbM_{1,1}(y^2)=-y\cdot \bA^\T \barR(y^2)\bA$. 
 Furthermore, recall the change of basis \eqref{eq:QZ-basis-change} and the notation for the $d$-by-$d$ resolvent \eqref{eq:R(z)}. Note that upon this change of basis,
 \begin{equation}\label{eq:barR-basis}
     \bBc^\T \barR(z) \bBc = \MatL
        \bR(z) &\0 \\
        \0 & -z^{-1}\bI_{(n-d)\times (n-d)}
     \MatR .
 \end{equation}

 We proceed with the computation. First,
 \begin{align}
     -y\bU^\T \barR(y^2)\bU 
     &\overset{(\star)}{\simeq} -y\cdot  n^{-1}\tr(\barR(y^2))\otimes \bI_{r\times r}\nonumber \\
     &= y\cdot \frac{d}{n}\left( -d^{-1}\tr\bR(y^2) +\frac1{y^2}\frac{n-d}{d}\right) \otimes \bI_{r\times r} \nonumber \\
     &\overset{(\star\star)}{\simeq} \left( (1-\beta\gamma)\frac1y - (\beta\gamma) y \rho_{\gamma,\beta}(y^2) \right)\otimes \bI_{r\times r} \equiv \kappa_{\gamma,\beta}^{(1)}(y) \otimes \bI_{r\times r}.\label{eq:M11-1}
 \end{align}
 Above, $(\star)$ follows since $\barR(y^2)$ has an orthogonally-invariant distribution and $\bU^\T\bU=\bI_{r\times r}$, and $(\star\star)$ follows from Lemma~\ref{lem:Rho-i}. Straightforward calculations, starting from the formula \eqref{eq:rho}, yields the formula \eqref{eq:kappa-1} for $\kappa_{\gamma,\beta}^{(1)}$.

 Next, 
 \begin{equation}\label{eq:M11-2}
     -y\bU^\T \barR(y^2) \bQc^\perp \bU \overset{(\star)}{=} -y(\bQc^\perp \bU^\T) \barR(y^2) \bQc^\perp \bU \overset{(\star\star)}{=} (1-\beta\gamma)\frac1y \otimes \bI_{r\times r} \equiv \kappa_{\gamma,\beta}^{(2)}(y) \otimes \bI_{r\times r} ,
 \end{equation}
 where $(\star)$ follows since the subspaces $\range(\bQc),\range(\bQc)^\perp$ are preserved by $\barR(y^2)$, and $(\star\star)$ follows since $\barR(y^2)$ acts like $(1/y)\bI$ on $\range(\bQc)^\perp$, and since this subspace is random (with dimension $n-d$), $(\bQc\bU)^\T ((\bQc\bU)) \simeq \frac{n-d}{n}\bI_{r\times r} $.
 
Next, we have 
\begin{equation}\label{eq:M11-3}
    -y\bU^\T \barR(y^2)\bVbar_1 = -y\bU^\T \barR(y^2)\bZ_1(\bZ_1^\T\bZ_1)^{-1}\bV_1 \simeq \0.
\end{equation}
To see this, observe that the distribution of $\bZ_1$ is orthogonally invariant from the right, namely for every orthogonal $\bO$, $\bZ_1\overset{d}{=}\bZ_1\bO$. Replacing $\bZ_1$ by $\bZ_1\bO$, note that $\barR(y^2)$ does not change, and \eqref{eq:M11-3} becomes $-y\bU^\T \barR(y^2)\bZ_1(\bZ_1^\T\bZ_1)^{-1}\bO\bV_1$. Picking $\bO\sim \mathrm{Haar}(O(d))$, we deduce that this expression must be asymptotically vanishing. Similarly,
\begin{equation}\label{eq:M11-4}
    -y(\bQc^\perp \bU)^\T \barR(y^2)\bVbar_1 \simeq \0.
\end{equation}

The last term is 
\begin{align}
    -y\bVbar_1^\T \barR(y^2)\bVbar_1 
    &= -y \bV_1^\T (\bZ_1^\T\bZ_1)^{-1} \bZ_1^\T \barR(y^2)\bZ_1(\bZ_1^\T\bZ_1)^{-1}\bV_1 \nonumber \\
    &\simeq -y \beta d^{-1}\tr\left( (\bZ_1^\T\bZ_1)^{-1} \bZ_1^\T \barR(y^2)\bZ_1(\bZ_1^\T\bZ_1)^{-1} \right) \otimes \bI_{r\times r},\label{eq:M11-5-a}
\end{align}
where the second line follows from the orthogonal invariance of $\bZ$ and the assumption \eqref{eq:assum:incoherence}. We now need to calculate the trace on the r.h.s. of \eqref{eq:M11-5-a}. To this end, we apply the change of basis from \eqref{eq:QZ-basis-change}, 
\begin{align}
    d^{-1}\tr\left( (\bZ_1^\T\bZ_1)^{-1} \bZ_1^\T \barR(y^2)\bZ_1(\bZ_1^\T\bZ_1)^{-1} \right)
    &= d^{-1}\tr\left( \barR(y^2)\bZ_1(\bZ_1^\T\bZ_1)^{-2} \bZ_1^\T \right) \nonumber \\
    &= d^{-1}\tr\left( \bBc^\T\barR(y^2)\bBc\bBc^\T \bZ_1(\bZ_1^\T\bZ_1)^{-2} \bZ_1^\T \bBc^\T \right) \nonumber \\
    &= d^{-1} \tr \left( 
    \MatL
        \bR(y^2) &\0 \\
        \0 & -y^{-2}\bI_{(n-d)\times (n-d)}
     \MatR
    \MatL
        \diag(1/\bmu) &\0 \\
        \0 &\0
    \MatR
    \right) \nonumber \\
    &= d^{-1}\tr \left( \bR(y^2)\diag(1/\bmu) \right)
    \label{eq:M11-5-b}
\end{align}
    where we denote $1/\bmu=(1/\mu_1,\ldots,1/\mu_d)$. To evaluate \eqref{eq:M11-5-b}, we use Lemma~\ref{lem:Rho-i}:
\begin{align}
    d^{-1}\tr \left( \bR(y^2)\diag(1/\bmu) \right) 
    &= d^{-1}\sum_{i=1}^d \bR(y^2)_{i,i}\frac{1}{\mu_i} \simeq d^{-1}\sum_{i=1}^d \frac{1}{\mu_i(\mu_i-C(y^2))}
\end{align}
where we denote
\begin{equation}
    C(z) = z - \frac{\sqrt{\gamma}(1-\beta)}{1+\beta\sqrt{\gamma}{\rho_{\gamma,\beta}}(z)} .
\end{equation}
One can write $\frac{1}{\mu_i\left( \mu_i - C(y^2) \right)} = \frac{1}{C(z)}\left(\frac{1}{\mu_i-C(y^2)} - \frac{1}{\mu_i}\right)$. Moreover, recall that the empirical distribution of $\mu_1,\ldots,\mu_d$ converges weakly to a Marcheko-Pastur law with shape $\beta\gamma$ and scale $\gamma^{-1/2}$. Thus,
\begin{align}
    d^{-1}\sum_{i=1}^d \frac{1}{\mu_i(\mu_i-C(y^2))} 
    &= \frac{1}{C(y^2)}d^{-1}\sum_{i=1}^d \left(\frac{1}{\mu_i-C(y^2)} - \frac{1}{\mu_i}\right) \nonumber \\
    &\simeq \frac{1}{C(y^2)}\left( \MPStiel_{\beta\gamma,\gamma^{-1/2}}\left( C(y^2) \right) - \frac{\sqrt{\gamma}}{1-\beta\gamma} \right) \nonumber\\
    &= \frac{1}{C(y^2)}\left( \rho_{\beta,\gamma}(y^2) - \frac{\sqrt{\gamma}}{1-\beta\gamma} \right),\label{eq:M11-5-c}
\end{align}
where the last equality is due to \eqref{eq:lem:Rho-i:2}. Combining Eqs. \eqref{eq:M11-5-a}-\eqref{eq:M11-5-c}, we finally get:
\begin{equation}
    -y\bVbar_1^\T \barR(y^2)\bVbar_1 \simeq -\frac{\beta y}{C(y^2)}\left( \rho_{\beta,\gamma}(y^2) - \frac{\sqrt{\gamma}}{1-\beta\gamma}\right) \otimes \bI_{r\times r}\equiv \kappa_{\gamma,\beta}^{(3)}(y) \otimes \bI_{r\times r} .
\end{equation}
One can obtain the formula \eqref{eq:kappa-3} for $\kappa_{\gamma,\beta}^{(3)}(y)$ by straightforward (if tedious) calculation.

To summarize, we have computed the following asymptotic formula for the block $\hbM_{1,1}(y)$:
\begin{equation}\label{eq:M11-final}
    \hbM_{1,1}(y) \simeq
    \MatL 
    \kappa_{\gamma,\beta}^{(1)}(y) &\kappa_{\gamma,\beta}^{(2)}(y) &0 \\
    \kappa_{\gamma,\beta}^{(2)}(y) &\kappa_{\gamma,\beta}^{(2)}(y) &0 \\
    0 &0 &\kappa_{\gamma,\beta}^{(3)}(y)
    \MatR \otimes \bI_{r\times r} .
\end{equation}  

\subsection{The top right block}

We next consider the block $\hbM_{1,2}(y)$ given in \eqref{eq:hbM-12}. 

One can verify that the following terms vanish asymptotically:
\begin{align}
    &-\bU^\T \barR(y^2)\bQc\bZ \bV \simeq \0 \label{eq:M12-1} \\
    &-\bU^\T \barR(y^2)\bQc\bZ \bW_1 = -\bU^\T \barR(y^2)\bQc\bZ_2(\bV_2-\bZ_2^\T \bVbar_1)\simeq \0 \label{eq:M12-2}  \\
    &-\bU^\T \barR(y^2)\bQc\bZ \bW_2 = -\bU^\T \barR(y^2)\bQc\bZ_2(\bZ_2^\T\bQc^\perp\bU) \simeq \0 \label{eq:M12-3}  \\
    &-(\bQc^\perp\bU)^\T \barR(y^2)\bQc\bZ \bV \simeq \0 \label{eq:M12-4}  \\
    &-(\bQc^\perp\bU)^\T \barR(y^2)\bQc\bZ \bW_1 = -(\bQc^\perp\bU)^\T \barR(y^2)\bQc\bZ_2(\bV_2-\bZ_2^\T \bVbar_1) \simeq \0 \label{eq:M12-5}  \\
    &-(\bQc^\perp\bU)^\T \barR(y^2)\bQc\bZ \bW_2 = -(\bQc^\perp\bU)^\T\barR(y^2)\bQc\bZ_2(\bZ_2^\T\bQc^\perp\bU) \simeq \0 \label{eq:M12-6} .
\end{align}
To immediately see this:
\begin{itemize}
    \item \eqref{eq:M12-1} and \eqref{eq:M12-4} vanish due to the right orthogonal invariance of $\bZ$, which is never ``balanced out'', similarly to \eqref{eq:M11-3}.
    \item \eqref{eq:M12-3} and \eqref{eq:M12-6} vanish because $\bQc^\perp\bZ_2$ is independent of $\bQc\bZ$, hence $\bZ_2^\T\bQc^\perp\bU$ is an i.i.d. Gaussian matrix which is independent of all the other matrices.
    \item As for \eqref{eq:M12-2} and \eqref{eq:M12-5}, the least immediate terms are those involving $\bVbar_1$. Using its definition \eqref{eq:Vbar-1-def}, 
    \begin{equation*}
        \bU^\T \barR(y^2)\bQc\bZ_2(\bZ_2^\T \bVbar_1) = \bU^\T \barR(y^2)\bQc\bZ_2\bZ_2^\T\bZ_1(\bZ_1^\T\bZ_1)^{-1}\bV_1 \simeq \0
    \end{equation*}
    again due to the right orthogonal invariance of $\bZ_1$, similarly to \eqref{eq:M11-3}.
\end{itemize}

The remaining terms are those that involve $\bVbar_1$ on the left. 

First,
\begin{align}
    -\bVbar_1^\T \barR(y^2)\bQc\bZ \bV 
    &= -\bV_1^\T (\bZ_1^\T\bZ_1)^{-1} \bZ_1^\T \barR(y^2)\bQc\bZ \bV \nonumber \\
    &= -\bV_1^\T (\bZ_1^\T\bZ_1)^{-1} \bZ_1^\T \barR(y^2)\bQc\bZ_1 \bV_1 + \underbrace{\bV_1^\T (\bZ_1^\T\bZ_1)^{-1} \bZ_1^\T \barR(y^2)\bQc\bZ_2 \bV_2}_{\simeq \0\;\textrm{(Orth. Inv.)}} \nonumber  \\
    &\simeq -\beta d^{-1}\tr\left( (\bZ_1^\T\bZ_1)^{-1} \bZ_1^\T \barR(y^2)\bQc\bZ_1 \right) \otimes \bI_{r\times r} \nonumber \\
    &\overset{(\star)}{=} -\beta d^{-1}\tr \bR(y^2) \otimes \bI_{r\times r} \nonumber \\
    &\overset{(\star\star)}{\simeq} -\beta\rho_{\gamma,\beta}(y^2) \otimes \bI_{r\times r} \equiv \kappa_{\gamma,\beta}^{(4)}(y) \otimes \bI_{r\times r} .
    \label{eq:M12-7}
\end{align}
Above, $(\star)$ follows by the change of basis \eqref{eq:QZ-basis-change} and \eqref{eq:barR-basis}; $(\star\star)$ follows from Lemma~\ref{lem:Rho-i}.
One can easily verify the claimed formula \eqref{eq:kappa-4} for $\kappa_{\gamma,\beta}^{(4)}(y)$, starting from \eqref{eq:rho}.

Next, we consider
\begin{align}
    -\bVbar_1^\T \barR(y^2)\bQc\bZ \bW_1
    &= -\bVbar_1^\T \barR(y^2)\bQc\bZ_2 (\bV_2-\bZ_2^\T \bVbar_1) \nonumber \\
    &\overset{(\star)}{\simeq} \bV_1^\T \bZ_1^\T (\bZ_1^\T \bZ_1)^{-1}\barR(y^2)\bQc\bZ_2\bZ_2^\T \bZ_1(\bZ_1^\T \bZ_1)^{-1}\bV_1 \nonumber \\
    % &\overset{(\star\star)}{\simeq}
    &\simeq \beta d^{-1}\tr\left( \bZ_1^\T (\bZ_1^\T \bZ_1)^{-1}\barR(y^2)\bQc\bZ_2\bZ_2^\T \bZ_1(\bZ_1^\T \bZ_1)^{-1} \right) \otimes \bI_{r\times r} \label{eq:M12-8}
\end{align}
where $(\star)$ follows by dropping the asymptotically vanishing term involving $\bV_2$ (same reasoning as \eqref{eq:M12-1}-\eqref{eq:M12-6} from before).
Doing a change of basis, similar to \eqref{eq:M11-5-b},
\begin{align}
    \eqref{eq:M12-8} 
    &= \beta d^{-1}\tr\left( \bR(y^2)\bX\bX^\T \diag(1/\bmu) \right) \nonumber \\
    &= \beta d^{-1}\tr\left(  \diag(1/\bmu) \right) \label{eq:M12-9-1}\\
    &\quad- \beta d^{-1}\tr\left( \bR(y^2) \right) \label{eq:M12-9-2} \\
    &\quad+ \beta y^2 d^{-1}\tr\left( \bR(y^2) \diag(1/\bmu) \right) 
    \label{eq:M12-9-3}
\end{align}
where \eqref{eq:M12-9-1}-\eqref{eq:M12-9-3} follows by replacing $\bX\bX^\T$ with $(\bX\bX^\T+\diag(\bmu)-y^2\bI) - \diag(\bmu) + y^2\bI$. Observe that \eqref{eq:M12-9-3} appeared previously in \eqref{eq:M11-5-b} and evaluates to \eqref{eq:M11-5-c}. Thus, combining \eqref{eq:M12-8}-\eqref{eq:M12-9-3}, 
\begin{align}
    -\bVbar_1^\T \barR(y^2)\bQc\bZ \bW_1 
    &\simeq \left( 
 \frac{\beta\sqrt{\gamma}}{1-\beta\gamma} - \beta\rho_{\gamma,\beta}(y^2) + \frac{\beta y^2}{C(y^2)}\left( \rho_{\beta,\gamma}(y^2) - \frac{\sqrt{\gamma}}{1-\beta\gamma} \right)\right) \otimes \bI_{r\times r}\nonumber \\
 &= \bar{\kappa}^{(5)}_{\gamma,\beta}(y) \otimes \bI_{r\times r}
\end{align}
A straightforward calculations gives the following formula: 
\begin{equation}\label{eq:kappa-bar-5}
    \bar{\kappa}^{(5)}_{\gamma,\beta}(y) = \frac
    {(1-\beta) \left( -\sqrt{\gamma}y^2 + (1+\gamma-2\beta\gamma) + \sqrt{-\gamma\Delta_{\gamma,\beta}(y^2)} \right)}
    {2(1-\beta\gamma)y^2}
\end{equation}
We denote
\begin{equation}
    {\kappa}^{(5)}_{\gamma,\beta}(y) = \bar{\kappa}^{(5)}_{\gamma,\beta}(y) - \frac{\beta\sqrt{\gamma}}{1-\gamma\beta}
\end{equation}
which has the explicit formula \eqref{eq:kappa-5}.

Lastly, one can verify (similarly to \eqref{eq:M12-3} and \eqref{eq:M12-6}) that 
\begin{equation}
    -\bVbar_1^\T\bR(y^2)\bQc\bZ\bW_2 \simeq \0.
\end{equation}

To summarize, we have the following asymptotic formula for $\hbM_{1,2}(y^2)$:
\begin{equation}\label{eq:M12-final}
    \hbM_{1,2}(y) \simeq
    \MatL 
    0 &0 &0 \\
    0 &0 &0 \\
    \kappa_{\gamma,\beta}^{(4)}(y) &\kappa_{\gamma,\beta}^{(5)}(y) + \frac{\beta\sqrt{\gamma}}{1-\gamma\beta}  &0
    \MatR \otimes \bI_{r\times r}.
\end{equation} 

\subsection{The bottom right block}

It remains to compute the block $\hbM_{2,2}(y)$ given in \eqref{eq:hbM-22}. 

The first term to consider is 
\begin{align}
    y\bV^\T \left( y^2\bI - (\bQc\bZ)^\T (\bQc\bZ) \right)^{-1}\bV
    &\simeq \frac{y}{m} \tr \left( y^2\bI - (\bQc\bZ)^\T (\bQc\bZ) \right)^{-1} \otimes \bI_{r\times r} \nonumber \\
    &\overset{(\star)}{=} \frac{y}{m} \left[ \tr \left( y^2\bI -  (\bQc\bZ)(\bQc\bZ)^\T \right)^{-1} - (n-m)\frac{1}{y^2} \right] \nonumber \\
    &\overset{(\star\star)}{=} \frac{y}{m} \left[ -\tr \bR(y^2) + (n-d)\frac{1}{y^2} - (n-m)\frac{1}{y^2} \right] \nonumber \\
    &\simeq -y\beta \rho_{\gamma,\beta}(y^2) +  (1-\beta)\frac{1}{y}    
    \equiv \kappa_{\gamma,\beta}^{(6)}(y) .
    \label{eq:M22-1}
\end{align}
Above, $(\star)$ follows from the elementary fact that for every matrix $\bA$, the matrices $\bA\bA^\T$ and $\bA^\T \bA$ have the same non-zero eigenvalues; $(\star\star)$ follows from \eqref{eq:barR-basis}. The function $\kappa^{(6)}_{\gamma,\beta}(y)$ evaluates to  \eqref{eq:kappa-6}.

Next, one can verify that cross terms which involve $\bW_2$ are asymptotically vanishing (similarly to \eqref{eq:M12-3}, \eqref{eq:M12-6}):
\begin{align}
    y\bV^\T \left( y^2\bI - (\bQc\bZ)^\T (\bQc\bZ) \right)^{-1}\bW_2 &\simeq \0, \label{eq:M22-2}\\
    y\bW_1^\T \left( y^2\bI - (\bQc\bZ)^\T (\bQc\bZ) \right)^{-1}\bW_2 &\simeq \0. \label{eq:M22-3}
\end{align}

The remaining terms are somewhat more involved to evaluate. Note that $\bW_1,\bW_2$ are only supported on the their last $m-d$ coordinates (see \eqref{eq:W1-W2-def}). Accordingly, we will need a formula for the bottom $(m-d)$-by-$(m-d)$ block of the matrix $( y^2\bI - (\bQc\bZ)^\T (\bQc\bZ) )^{-1}$. This block shall henceforth be denoted $\downD(y)$. Writing in block form,
\begin{equation}\label{eq:D-block-form}
    \left( y^2\bI - (\bQc\bZ)^\T (\bQc\bZ) \right)^{-1} = \MatL 
 y^2\bI - \bZ_1^\T \bZ_1 & -\bZ_1^\T (\bQc \bZ_2) \\
	-(\bQc\bZ_2)^\T \bZ_1 & y^2\bI - (\bQc\bZ_2)^\T \bQc \bZ_2
    \MatR^{-1} .
\end{equation}
(Recall that $\bQc\bZ_1$, since $\bQc$ is the projection onto the column space of $\bZ_1$.) By the block matrix inversion formula,
\begin{align}
    \downD(y) 
    &= \left( y^2\bI - (\bQc\bZ_2)^\T(\bQc \bZ_2)  -(\bQc\bZ_2)^\T \bZ_1(y^2\bI-\bZ_1^\T\bZ_1)^{-1}\bZ_1^\T (\bQc\bZ_2) \right)^{-1} \nonumber \\
    &= \left( y^2\bI - (\bQc\bZ_2)^\T \left[  \bI  +  \bZ_1(y^2\bI-\bZ_1^\T\bZ_1)^{-1}\bZ_1^\T \right] (\bQc\bZ_2) \right)^{-1}  . \label{eq:M22-4}
\end{align}
Next, we use the change of basis \eqref{eq:QZ-basis-change}, namely write $\bZ_1\bZ_1^\T = \bBc_1\diag(\bmu)\bBc_1^\T$ and $\bX=\bBc_1^\T \bZ_2$ in \eqref{eq:M22-4}. This gives 
\begin{equation}\label{eq:downD}
    \downD(y) = \frac{1}{y^2}\left( \bI - \bX^\T \diag(\bnu^2) \bX  \right)^{-1},\qquad\textrm{where}\quad\bnu = \sqrt{\frac{1}{y^2-\bmu}} .
\end{equation}
We will also need the $d$-by-$d$ matrix
\begin{equation}\label{eq:barD}
    \barD(y) = \frac1{y^2} \left( \bI - \diag(\bnu) \bX \bX^\T \diag(\bnu)  \right)^{-1} .
\end{equation}
The following lemma will play a similar role as Lemma~\ref{lem:Rho-i} did in the preceding calculation.

\begin{lemma}
\label{lem:D-i}
    Define 
    \begin{align}
        \delta(y) 
        &= \frac
     {\sqrt{\gamma}(1-\beta)y^2 - (1+\gamma+2\beta^2\gamma-3\beta-3\beta\gamma) - (1-\beta)\sqrt{-\gamma \Delta_{\gamma,\beta}(y^2)}}
     {2\beta(1+\gamma-\beta\gamma)y^2} \\
     \underline{\delta}(y) 
     &= \frac
     {\sqrt{\gamma}y^2 + \left(1+\gamma -2\beta\gamma\right) - \sqrt{-\gamma\Delta_{\gamma,\beta}(y^2)}}
     {2(1+\gamma-\beta\gamma)y^2}
    \end{align}
    and for every $1\le i \le d$,
    \begin{equation}\label{eq:lem:D-i-3}
        \hat{\delta}_i(y) = \frac{1}{y^2} \left( 1 - \frac{ B(y) }{\mu_i - \left[ y^2 - B(y) \right] } \right),\qquad \textrm{where}\quad B(y) = \sqrt{\gamma}\left( 1-2\beta + \beta y^2{{\delta}}(y) \right).
    \end{equation}
    Then
    \begin{equation}
    \begin{split}
        d^{-1}\tr\barD(y) \longrightarrow& \delta(y),\quad(m-d)^{-1}\tr\downD(y)\longrightarrow\underline{\delta}(y),\quad \max_{1\le i \le d} |\barD(y)_{i,i}-\hat{\delta}_i(y)| \longrightarrow 0 .     
    \end{split}
    \end{equation}
    Furthermore, $\delta(y)$ satisfies 
    \begin{equation}\label{eq:lem:D-i-5}
        y^2 {{\delta}}(y) = 1 - B(y) \MPStiel_{\gamma\beta,\gamma^{-1/2}}\left( y^2 - B(y) \right).
    \end{equation}
\end{lemma}
The proof of Lemma~\ref{lem:D-i} is deferred to Appendix, Section~\ref{sec:proof-lem:D-i}.

We continue. Let $\bD_{1,2},\bD_{2,2}$ be respectively the top-right and bottom-right blocks of \eqref{eq:D-block-form}. Note: $\bD_{2,2}\equiv \downD(y)$.
% ; this new notation is provisional, and will only be used now.) 
Unpacking the definition of $\bW_1$ \eqref{eq:W1-W2-def},
\begin{align}
    y \bV^\T \left( y^2\bI - (\bQc\bZ)^\T (\bQc\bZ) \right)^{-1} \bW_1 = y\bV_1^\T \bD_{1,2} (\bV_2-\bZ_2^\T \bVbar_1) + y\bV_2^\T \bD_{2,2} (\bV_2-\bZ_2^\T \bVbar_1).
\end{align}
First, we claim that the terms $y\bV_1^\T \bD_{1,2}\bV_2$ and $-y\bV_2^\T \bD_{2,2}\bZ_2^\T \bVbar_1$ are asymptotically vanishing. 
% To verify that $\bV_1^\T \bD_{1,2}\bV_2\simeq \0$, 
To see this,
one can use the right orthogonal invariance of $\bZ=[\bZ_1,\bZ_2]$, replacing it by $\bZ(\bO_1\oplus\bO_2)=[\bZ_1\bO_1,\bZ_2\bO_2]$ where $\bO_1\in O(d),\bO_2\in O(m-d)$ are Haar-distributed; we omit the precise details. 
% As for $\bV_1^\T \bD_{1,2}\bZ_2^\T \bVbar_1$, introduce a rotation $\bO$ which is uniformly random on the column space of $\bZ_1$, namely $\range(\bQc)$, and fixes $\range(\bQc)^\perp$ (and is independent of $\bZ_2$). Replacing $\bZ_2$ by $\bO\bZ_2$, note that the matrix $(\bQc\bZ)^\T (\bQc\bZ)$ does not change (since $\bO$ and $\bQc$ commute). Moreover, $\bVbar_1$ is clearly in $\range(\bQc)$ (see \eqref{eq:Vbar-1-def}) and so $\bO^\T\bVbar_1$ has a uniformly random direction. Thus,
% \begin{align*}
%     \bV_1^\T \bD_{1,2}\bZ_2^\T\bVbar_1 \overset{d}{=} \bV_1^\T \bD_{1,2}\bZ_2^\T(\bO^\T \bVbar_1) \simeq \0.
% \end{align*}
% By a similar argument, we have $\bV_2^\T \bD_{2,2}\bZ_2^\T \bVbar_1\simeq \0$. 
Thus, the remaining terms are 
\begin{align}\label{eq:prekappa7-1}
    y \bV^\T \left( y^2\bI - (\bQc\bZ)^\T (\bQc\bZ) \right)^{-1} \bW_1 \simeq 
    -y\bV_1^\T \bD_{1,2}\bZ_2^\T \bVbar_1 + y\bV_2^\T \bD_{2,2}\bV_2. 
\end{align}
The second term on the r.h.s. of \eqref{eq:prekappa7-1} is easy to calculate: 
\begin{align}
    y\bV_2^\T \bD_{2,2}\bV_2 \nonumber 
    &\simeq (1-\beta)y\frac{1}{m-d}\tr\downD(y) \otimes \bI_{r\times r}
    \nonumber \\
    &\simeq (1-\beta)y\underline{\delta}(y) \otimes \bI_{r\times r} \label{eq:prekappa7-1a},
\end{align}
where we used Lemma~\ref{lem:D-i}. As for the other term, by the block matrix inversion formula, applied to \eqref{eq:D-block-form},
\begin{align}
    \bD_{1,2}= (y^2\bI - \bZ_1^\T\bZ_1)^{-1}\bZ_1^\T(\bQc\bZ_2)\downD(y),
\end{align}
hence
\begin{align}
    -y\bV_1^\T \bD_{1,2} \bZ_2^\T \bVbar_1 
    &\overset{(\star)}{=} -y\bV_1^\T \bD_{1,2}(\bQc\bZ_2)^\T \bVbar_1 \nonumber \\ 
    &= -y\bV_1^\T (y^2\bI - \bZ_1^\T\bZ_1)^{-1}\bZ_1^\T(\bQc\bZ_2)\downD(y)(\bQc\bZ_2)^\T \bZ_1(\bZ_1^\T \bZ_1)^{-1}\bV_1 \nonumber \\
    &\simeq -\beta y \cdot d^{-1}\tr\left( (y^2\bI - \bZ_1^\T\bZ_1)^{-1}\bZ_1^\T(\bQc\bZ_2)\downD(y)(\bQc\bZ_2)^\T \bZ_1(\bZ_1^\T \bZ_1)^{-1} \right) \otimes \bI_{r\times r}\nonumber \\
    &\overset{(\star\star)}{=} -\beta y \cdot d^{-1}\tr\left( \diag(\bnu) \bX\downD(y) \bX^\T \diag(\bnu) \right) \otimes \bI_{r\times r} . \label{eq:prekappa7-2}
\end{align}
Above, $(\star)$ follows since $\bVbar_1$ is in $\range(\bQc)$; $(\star\star)$ uses the change of basis \eqref{eq:QZ-basis-change} and the definition $\bnu=1/\sqrt{y^2-\bmu}$. Using the expression \eqref{eq:downD} for $\downD(y)$, 
\begin{align}
    \eqref{eq:prekappa7-2}
    &= -\beta y \cdot d^{-1}\tr\left( (\diag(\bnu) \bX)\frac{1}{y^2}\left( \bI - (\diag(\bnu)\bX)^\T (\diag(\bnu)\bX)  \right)^{-1} (\diag(\bnu) \bX)^\T \right) \nonumber \\
    &= -\beta y \cdot d^{-1}\tr\left( \frac{1}{y^2}\left( \bI - (\diag(\bnu) \bX)(\diag(\bnu)\bX)^\T   \right)^{-1} (\diag(\bnu)\bX)(\diag(\bnu) \bX)^\T \right) \nonumber \\
    &= \frac{\beta}{y} - \beta y \cdot d^{-1}\tr(\barD(y)) \simeq \frac{\beta}{y} - \beta y \delta(y) \label{eq:prekappa7-3},
\end{align}
where we used Lemma~\ref{lem:D-i}. Finally, combining \eqref{eq:prekappa7-1}, \eqref{eq:prekappa7-1a}, \eqref{eq:prekappa7-2} and \eqref{eq:prekappa7-3}, 
\begin{align}
 y \bV^\T \left( y^2\bI - (\bQc\bZ)^\T (\bQc\bZ) \right)^{-1} \bW_1
 &\simeq \left( (1-\beta)y\underline{\delta}(y) + \frac{\beta}{y} - \beta y \delta(y)\right) \otimes \otimes \bI_{r\times r}\nonumber \\
     &\equiv \kappa_{\gamma,\beta}^{(7)}(y) \otimes \bI_{r\times r} .\label{eq:M22-3}
\end{align}
A straightforward calculation gives the formula \eqref{eq:kappa-7} for $\kappa_{\gamma,\beta}^{(7)}$. 

Next, consider
\begin{align}\label{eq:M22-4a}
    y \bW_1^\T \left( y^2\bI - (\bQc\bZ)^\T (\bQc\bZ) \right)^{-1} \bW_1 
    \simeq y\bV_2^\T \downD(y) \bV_2  + y\bVbar_1^\T \bZ_2 \downD(y)\bZ_2^\T\bVbar_1 ,
\end{align}
where we used $\bV_2^\T \downD(y)\bZ_2^\T \bVbar_1\simeq \0$, which follows by a similar argument as above. The first term was calculated in \eqref{eq:M22-3}. 
Unpacking the definition of $\bVbar_1$,
\begin{align}
    y\bVbar_1^\T \bZ_2 \downD(y)\bZ_2^\T\bVbar_1 
    &= y\bV_1^\T (\bZ_1^\T\bZ_1)^{-1}\bZ_1^\T \bZ_2 \downD(y)\bZ_2^\T \bZ_1 (\bZ_1^\T\bZ_1)^{-1} \bV_1 \nonumber \\
    &\simeq \beta y\cdot d^{-1}\tr \left( (\bZ_1^\T\bZ_1)^{-1}\bZ_1^\T \bZ_2 \downD(y)\bZ_2^\T \bZ_1 (\bZ_1^\T\bZ_1)^{-1} \right) \otimes \bI_{r\times r}. \label{eq:M22-4b}
\end{align}
Using the change of basis \eqref{eq:QZ-basis-change}, and the formula \eqref{eq:downD} for $\downD(y)$, 
\begin{align}
    \eqref{eq:M22-4b} 
    &= \beta y \cdot d^{-1} \tr \left( \bX \frac{1}{y^2}( \bI - \bX^\T \diag(\bnu^2) \bX  g)^{-1}\bX^\T \diag(1/\bmu)\right) \nonumber \\
    &= \beta y \cdot d^{-1} \tr \left( (\diag(\bnu)\bX) \frac{1}{y^2}( \bI - (\diag(\bnu)\bX)^\T (\diag(\bnu)\bX)  g)^{-1}(\diag(\bnu)\bX)^\T \diag(1/(\bnu^2\bmu))\right) \nonumber \\
    &= \beta y \cdot d^{-1} \tr \left(  \frac{1}{y^2}( \bI - (\diag(\bnu)\bX) (\diag(\bnu)\bX)^\T )^{-1}(\diag(\bnu)\bX) (\diag(\bnu)\bX)^\T \diag(1/(\bnu^2\bmu))\right),
    \label{eq:M22-5}
\end{align}
where \eqref{eq:M22-5} follows since for every matrix $\bA$ and analytic $f(\cdot)$, $\bA f(\bA^\T\bA)\bA^\T=f(\bA\bA^\T)\bA\bA^\T$. Rewriting the matrix in the middle as $(\diag(\bnu)\bX) (\diag(\bnu)\bX)^\T=((\diag(\bnu)\bX) (\diag(\bnu)\bX)^\T-\bI)+\bI$ and recalling the definition of $\barD(y)$ and $\bnu^2=\frac{1}{y^2-\bmu}$, 
\begin{align}
    \eqref{eq:M22-5}
    &= \frac{\beta}{y}\cdot d^{-1}\tr\left(\diag\left(\bm{1}-\frac{y^2}{\bmu}\right)\right) - \beta y \cdot d^{-1}\tr \left( \barD(y) \diag\left(\bm{1}-\frac{y^2}{\bmu}\right)\right). \label{eq:M22-6}
\end{align}
Recall that the empirical distribution of $\bmu$ converges to a Marchenko-Pastur law with shape $\beta\gamma$ and scale $\gamma^{-1/2}$. Accordingly, the first term is 
\begin{align}
    \label{eq:M22-7}
    \frac{\beta}{y}\cdot d^{-1}\tr\left(\diag\left(\bm{1}-\frac{y^2}{\bmu}\right)\right) \simeq \frac{\beta}{y} - \frac{\sqrt{\gamma}\beta}{1-\beta\gamma}y .
\end{align}
To treat the second term, we use Lemma~\ref{lem:D-i}:
\begin{align}
    -\beta y \cdot d^{-1}\tr &\left( \barD(y) \diag\left(\bm{1}-\frac{y^2}{\bmu}\right)\right) \nonumber \\
    &\simeq  
    -\beta y \delta(y) +
    \beta y\cdot d^{-1}\sum_{i=1}^d \left( 1 - \frac{ B(y) }{\mu_i - \left[ y^2 - B(y) \right] } \right)\frac{1}{\mu_i} \nonumber \\
    &= -\beta y \delta(y) + \frac{\sqrt{\gamma}\beta}{1-\beta\gamma}y - \beta y \cdot d^{-1} \sum_{i=1}^d \frac{B(y)}{\mu_i\left(\mu_i-\left[y^2-B(y)\right]\right)} \nonumber \\
    &= -\beta y \delta(y) + \frac{\sqrt{\gamma}\beta}{1-\beta\gamma}y - \frac{\beta y B(y)}{y^2-B(y)} d^{-1}\sum_{i=1}^d \left( \frac{1}{\mu_i-\left[y^2-B(y)\right]}  - \frac{1}{\mu_i} \right) \nonumber \\
    &\simeq -\beta y \delta(y) +\frac{\sqrt{\gamma}\beta y}{1-\beta\gamma}\left( 1 + \frac{B(y)}{y^2-B(y)} \right) - \frac{\beta y B(y)}{y^2-B(y)} \MPStiel_{\gamma\beta,\gamma^{-1/2}}\left( y^2-B(y)\right) \nonumber \\
    &\overset{(\star)}{=}  -\beta y \delta(y) + \frac{\sqrt{\gamma}\beta y}{1-\beta\gamma}\left( 1 + \frac{B(y)}{y^2-B(y)} \right)
    - \frac{\beta y B(y)}{y^2-B(y)} \cdot \frac{1-y^2 \delta(y)}{B(y)}.
    \label{eq:M22-8}
\end{align}
Above, $(\star)$ uses \eqref{eq:lem:D-i-5}. Finally, combining \eqref{eq:M22-4a}-\eqref{eq:M22-8},
\begin{align}
    y \bW_1^\T &\left( y^2\bI - (\bQc\bZ)^\T (\bQc\bZ) \right)^{-1} \bW_1 \nonumber \\
    &\simeq \left( (1-\beta)y\underline{d}(y) + \frac{\beta}{y} - \beta y \delta(y) + 
    % \frac{\beta y}{y^2-B(y)}\left[ \frac{\sqrt{\gamma}}{1-\beta\gamma} B(y)-1+y^2 \delta(y) \right]
    \frac{\sqrt{\gamma}\beta y}{1-\beta\gamma}\cdot\frac{B(y)}{y^2-B(y)} -  \frac{\beta y (1-y^2 \delta(y))}{y^2-B(y)} 
    \right) \otimes \bI_{r\times r} 
    \nonumber \\
    &\equiv \kappa^{(8)}_{\gamma,\beta}(y) \otimes \bI_{r\times r} .
\end{align}
The formula (\ref{eq:kappa-8}) for $\kappa^{(8)}_{\gamma,\beta}(y)$ can be verified by direct computation.

Finally, it remains to compute
\begin{align}
    y \bW_2^\T \left( y^2\bI - (\bQc\bZ)^\T (\bQc\bZ) \right)^{-1} \bW_2 = y\bU^\T \bQc^\perp \bZ_2 \downD(y)\bZ_2^\T \bQc^\perp \bU.
\end{align}
Note that $\bZ_2^\T \bQc^\perp \bU$ is independent of $\downD(y)$. Moreover, $(\bQc^\perp\bU)^\T(\bQc^\perp\bU)\simeq (1-\beta\gamma)\bI_{r\times r}$. Introducing, for convenience, a new Gaussian matrix $\bZ_3 \overset{d}{=}\bZ_2$, 
\begin{align}
    y\bU^\T \bQc^\perp \bZ_2 \downD(y)\bZ_2^\T \bQc^\perp \bU
    &\overset{d}{=} y\bU^\T \bQc^\perp \bZ_3 \downD(y)\bZ_3^\T \bQc^\perp \bU \nonumber \\
    &\simeq y(1-\beta\gamma)n^{-1}\tr\left( \bZ_3 \downD(y)\bZ_3^\T \right) \otimes \bI_{r\times r}\nonumber \\
    &\simeq y(1-\beta\gamma)\gamma(1-\beta)\cdot (m-d)^{-1}\tr\left(  \downD(y)\bZ_3^\T \bZ_3\right) \otimes \bI_{r\times r} \nonumber \\
    &\overset{(\star)}{\simeq}  y(1-\beta\gamma)\gamma(1-\beta)\cdot (m-d)^{-1}\tr\left(  \downD(y)\right) \cdot (m-d)^{-1}\tr\left( \bZ_3^\T \bZ_3\right) \otimes \bI_{r\times r}\nonumber \\
    &\simeq y\frac{(1-\beta\gamma)\gamma(1-\beta)}{\sqrt{\gamma}} \underline{\delta}(y) \otimes \bI_{r\times r} \equiv \kappa_{\gamma,\beta}^{(9)}(y) \otimes \bI_{r\times r}.
\end{align}
The last equality follows from Lemma~\ref{lem:D-i}, and $(\star)$ hold because $\bZ_3^\T\bZ_3$ and $\downD(y)$ are asymptotically free random matrices. The formula for $\kappa_{\gamma,\beta}^{(9)}(y)$, \eqref{eq:kappa-9}, can be readily verified.

To summarize, we have the following asymptotic formula for $\hbM_{2,2}(y^2)$:
\begin{equation}\label{eq:M22-final}
    \hbM_{2,2}(y) \simeq
    \MatL 
    \kappa_{\gamma,\beta}^{(6)}(y) &\kappa_{\gamma,\beta}^{(7)}(y) &0 \\
    \kappa_{\gamma,\beta}^{(7)}(y) &\kappa_{\gamma,\beta}^{(8)}(y) &0 \\
   0 &0 &\kappa_{\gamma,\beta}^{(9)}(y)
    \MatR \otimes \bI_{r\times r}.
\end{equation} 

Finally, the proof of Lemma~\ref{lem:Mhat-Lim} is concluded by combining \eqref{eq:lem:Mhat-Lim-preproof}, \eqref{eq:Sigma-Inv}, \eqref{eq:M11-final}, \eqref{eq:M12-final} and \eqref{eq:M22-final}.
\qed

\subsection{Proof of Lemma~\ref{lem:D-i}}
\label{sec:proof-lem:D-i}

The calculation is similar to the proof of Lemma~\ref{lem:Rho-i}, and likewise is classical \cite{bai2010spectral}. 

For any $1\le i \le d$, denote by $\bP_i \in \RR^{(d-1)\times d}$ the projection matrix onto the coordinate set $[d]\setminus \{i\}$. We have the following block decomposition (up to a coordinate permutation):
	\begin{equation}\label{eq:D:Aux:1}
		\barD(y) = \frac{1}{y^2} \MatL  
			1-\nu_i^2 (\bX\bX^\T)_{i,i} & -\be_i^\T (\diag(\bnu)\bX\bX^\T\diag(\bnu))\bP_i^\T \\
			-\bP_i (\diag(\bnu)\bX\bX^\T\diag(\bnu))\be_i & \bP_i( \bI - \diag(\bnu)\bX\bX^\T\diag(\bnu))\bP_i^\T 
		\MatR^{-1} .
	\end{equation}
	By the block matrix inversion formula,
	\begin{align}\label{eq:D:Aux:2}
		 \barD(y)_{i,i}  &= \frac{1}{y^2}\left(1 - \nu_i ^2 (\bX\bX^\T)_{i,i} - \right. \nonumber\\
   &\left.\nu_i^2 \be_i^\T (\bX\bX^\T\diag(\bnu))\bP_i^\T \left( \bP_i( \bI - \diag(\bnu)\bX\bX^\T\diag(\bnu))\bP_i^\T  \right)^{-1} \bP_i (\diag(\bnu)\bX\bX^\T)\be_i \right)^{-1}.
	\end{align}
 Recall that $\bX^\T \be_i$ has i.i.d. Gaussian entries $\m{N}(0,1/\sqrt{nm})$. By the same line of reasoning as in Section~\ref{sec:proof-lem:Rho-i}, Eqs. \eqref{eq:Resolvement:Aux:3}-\eqref{eq:Resolvement:Aux:6}, the following concentration holds \emph{simultaneously} over all $1\le i\le d$: 
 \begin{align}
     \label{eq:D:Aux:3}
     \barD(y)_{i,i}  &\simeq \frac{1}{y^2}\left(1 - \nu_i ^2 \sqrt{\gamma}(1-\beta) - \nu_i^2\frac{1}{\sqrt{nm}}\tr 
   \left(\bX^\T\diag(\bnu) \left(  \bI - \diag(\bnu)\bX\bX^\T\diag(\bnu)  \right)^{-1} \diag(\bnu)\bX\right) \right)^{-1}.
 \end{align}
 Consider the trace in \eqref{eq:D:Aux:3}:
\begin{align}
    &\tr 
   \left(\bX^\T\diag(\bnu) \left(  \bI - \diag(\bnu)\bX\bX^\T\diag(\bnu)  \right)^{-1} \diag(\bnu)\bX\right) 
   \nonumber \\
   &\qquad\qquad\overset{(\star)}{=}  \tr 
   \left( \left(  \bI - \bX^\T\diag(\bnu^2)\bX  \right)^{-1} \bX^\T \diag(\bnu^2)\bX\right) \nonumber \\
   &\qquad\qquad= -(m-d) + \tr 
   \left(  \bI - \bX^\T\diag(\bnu^2)\bX  \right)^{-1} \nonumber \\
   &\qquad\qquad= -(m-d) + \tr 
   \left(  \bI - \diag(\bnu)\bX (\diag(\bnu)\bX)^\T \right)^{-1} + (m-2d) \label{eq:D:Aux:4}
\end{align}
Above, $(\star)$ follows since for every matrix $\bA$ and analytic $f(\cdot)$, $\bA f(\bA^\T\bA)\bA^\T=f(\bA\bA^\T)\bA\bA^\T$ (one can readily verify this using the SVD of $\bA$). Denote
\begin{equation}
\label{eq:D:Aux:5}
    \hat{{\delta}}(y) = \frac1{y^2} d^{-1}\tr 
   \left(  \bI - \bX^\T\diag(\bnu^2)\bX  \right)^{-1} ,
\end{equation}
so that $\hat{{\delta}}(y) = d^{-1}\sum_{i=1}^d \barD(y)_{i,i}$. Combining Eqs. \eqref{eq:D:Aux:3}-\eqref{eq:D:Aux:5} and using $\nu_i^2=\frac{1}{y^2-\mu_i}$ yields, after some algebraic manipulations,
\begin{equation}
    \label{eq:D:Aux:6}
    \barD(y)_{i,i} \simeq  \frac{1}{y^2} \left( 1 - \frac{\sqrt{\gamma}\left( 1-2\beta + \beta y^2\hat{{\delta}}(y) \right)  }{\mu_i - \left[ y^2 - \sqrt{\gamma}\left( 1-2\beta + \beta y^2\hat{{\delta}}(y) \right) \right] } \right) .
\end{equation}
Recall that the empirical distribution of $(\mu_1,\ldots,\mu_d)$ converges to a Marchenko-Pastur law with shape $\beta\gamma$ and scale $\gamma^{-1/2}$. Thus, averaging \eqref{eq:D:Aux:6} over $1\le i \le d$, 
\begin{equation}
    \label{eq:D:Aux:7}
    y^2 \hat{{\delta}}(y) \simeq 1 - \sqrt{\gamma}\left( 1-2\beta + \beta y^2\hat{{\delta}}(y) \right) \MPStiel_{\gamma\beta,\gamma^{-1/2}}\left( y^2 - \sqrt{\gamma}\left( 1-2\beta + \beta y^2\hat{{\delta}}(y) \right) \right).
\end{equation}
The Stieltjes transform of the Marcheko-Pastur law has a well-known inverse (see e.g. \cite{bai2010spectral}):
\begin{equation}\label{eq:MPInv}
		\MPStiel_{\gamma\beta, \gamma^{-1/2}}^{-1}(s) = \frac{-\sqrt{\gamma} + (1-\gamma\beta)s}{s(\sqrt{\gamma} +\gamma\beta s)} \,.
	\end{equation}
 Combining \eqref{eq:D:Aux:7} and \eqref{eq:MPInv} yields a quadratic equation for $\hat{{\delta}}(y)$. It has two solutions:
 \begin{equation}\label{eq:D:Aux:8}
     \hat{{\delta}}(y) \simeq \frac
     {\sqrt{\gamma}(1-\beta)y^2 - (1+\gamma+2\beta^2\gamma-3\beta-3\beta\gamma) \pm (1-\beta)\sqrt{-\gamma \Delta_{\gamma,\beta}(y^2)}}
     {2\beta(1+\gamma-\beta\gamma)y^2} .
 \end{equation}
 We choose the ``correct'' solution according to the known tail behavior for large $y$: $\hat{\delta}(y)\asymp 1/y^2$ as $y\to \infty$. Consequently, the solution we need is ``$-$''; call it $\delta(y)$.

Finally, 
\begin{align}
    \frac{1}{m-d}\tr\downD(y) 
    = \frac{1}{m-d} \left[ d\cdot  d^{-1}\tr\barD(y) + (m-2d) \right] \simeq \frac{\beta}{1-\beta}\delta(y) + \frac{1-2\beta}{1-\beta} \equiv \underline{\delta}(y).
\end{align}
\qed

\section{Proof of Lemma~\ref{lem:T-lim}}
\label{sec:proof-lem:T-lim}

Denote the blocks of the matrix $\hbT(y)$,
\begin{align}
    \hbT_{1,1}(y) &= y^2  \bA^\T \left(y^2\bI - \bQc\bZ\bZ^\T \bQc  \right)^{-2} \bA, \label{eq:T11-def} \\
    \hbT_{1,2}(y) &= y  \bA^\T  \left(y^2\bI - \bQc\bZ\bZ^\T \bQc  \right)^{-2} (\bQc\bZ) \bB \label{eq:T12-def}\\
    \hbT_{2,2}(y) &= \bB^T (\bQc\bZ)^\T \left(y^2\bI - \bQc\bZ\bZ^\T \bQc  \right)^{-2} (\bQc\bZ) \bB . \label{eq:T22-def} 
\end{align}
so that 
\begin{align*}
    \hbT(y) = \MatL \hbT_{1,1}(y) &\hbT_{1,2}(y) \\ \hbT_{1,2}(y)^\T &\hbT_{2,2}(y) \MatR .
\end{align*}
We relate them to those of the matrix $\hbM({y})$, Eqs. \eqref{eq:hbM-11}-\eqref{eq:lem:Mhat-Lim-preproof}, whose limit we computed previously. 

First, the top-left block. 
\Revision{Recall 
\begin{align}
    \hbM_{1,1}(y) &= y\cdot \bA^\T (y^2\bI_{n\times n}-\bQc\bZ\bZ^\T\bQc)^{-1}\bA\nonumber
\end{align}
}
Denote by $(\cdot)'$ the derivative with respect to $y$. \Revision{It is straightforward to  verify that}
\begin{align}
    \hbT_{1,1}(y) = -\frac12 y \left(\frac{\hbM_{1,1}(y)}{y}\right)' .
\end{align}
Taking the derivative in the limiting expression \eqref{eq:M11-final} yields the claimed formula. Note that we can exchange the orders of the limit and the $y$-derivative since $\hbM_{1,1}(y)$ is analytic, and convergence in \eqref{eq:M11-final} is uniform on compact sets (see discussion after Lemma~\ref{lem:Mhat-Lim}).

Next, the top-right block. It is clear that 
\begin{align}
    \hbT_{1,2}(y) = -\frac12 \hbM_{1,2}'(y) .
\end{align}
The claimed formulas follow from \eqref{eq:M12-final}.

Lastly, the bottom right block. We can write 
\begin{align}
    \hbT_{2,2}(y) 
    &= \bB^T  \left(y^2\bI - (\bQc\bZ)^\T\bQc\bZ  \right)^{-2} (\bQc\bZ)^\T(\bQc\bZ) \bB \nonumber \\
    &= -\bB^T  \left(y^2\bI - (\bQc\bZ)^\T\bQc\bZ  \right)^{-1}  \bB + y^2 \bB^T  \left(y^2\bI - (\bQc\bZ)^\T\bQc\bZ  \right)^{-2}  \bB \,.\nonumber 
    % \\
    % &= - \frac{\hbM_{2,2}(y) }{y}-\frac12 y \left(\frac{\hbM_{2,2}(y)}{y}\right)' 
    % = -\frac1{2y} \left( y \hbM_{2,2}(y) \right)'.
\end{align}
\Revision{We have
\begin{align*}
    \hbM_{2,2}(y) &= y\cdot \bB^\T(y^2\bI_{m\times m}-\bZ^\T \bQc\bQc\bZ)^{-1}\bB\,, \\
    \left( \hbM_{2,2}(y)/y\right)' &= -2y \bB^\T(y^2\bI_{m\times m}-\bZ^\T \bQc\bQc\bZ)^{-2}\bB\,.
\end{align*}
Hence,
\begin{align}
    \hbT_{2,2}(y) = -\hbM_{2,2}(y)/y -\frac{y}{2} \left( \hbM_{2,2}(y)/y\right)' = -\frac{1}{2y}\left( y\hbM_{2,2}(y)\right)'\,,
\end{align}
where the last equality is straightforward to verify.
}
Finally, \eqref{eq:M22-final} gives the claimed expressions.
\qed

\section{Proof of Proposition~\ref{prop:vanishing-corrs}}
\label{sec:proof-prop:vanishing-corrs}

For $\sigma>\BBP$, denote $y=\SpikeFunc(\sigma)$, and let $\bd(\sigma)$ be the solution (up to a global sign) of 
\[
\bd \in \ker(\bKc_{\gamma,\beta}(y)-(1/\sigma)\bHc),\quad \langle \bd,\bTc_{\gamma,\beta}(y)\bd\rangle = 1 \,.    
\]
In light of Theorem~\ref{thm:3:Angles}, we need to show that $d_1,d_4\to 0$ as $\sigma\downarrow \BBP$.

First, observe that the vectors $\bd(\sigma) \in \RR^6$ are bounded for $\sigma$ in the vicinity of $\BBP$. This holds because $\bm{d}$ is an a.s. limit point of the vector $\hbd$, defined in \eqref{eq:d-c-def}, which is clearly bounded. 
 Next, the matrix $\bTc_{\gamma,\beta}(y)$ is PSD, being the limit of the PSD matrix $\hbT(y)$ in \eqref{eq:T-hat-Mat}. 
 Thus, $\langle \bm{d}(\sigma),\bTc(y)\bm{d}(\sigma)\rangle=1$ implies that necessarily
 \begin{align*}
     \MatL d_1(\sigma) \\ d_2(\sigma)\MatR^\T \MatL \tau_{\gamma,\beta}^{(1)}(y) & \tau_{\gamma,\beta}^{(2)}(y) \\
     \tau_{\gamma,\beta}^{(2)}(y) &\tau_{\gamma,\beta}^{(2)}(y)\MatR \MatL d_1(\sigma) \\ d_2(\sigma)\MatR \le 1,\qquad
     \MatL d_3(\sigma) \\ d_4(\sigma) \\d_5(\sigma)\MatR^\T \MatL \tau_{\gamma,\beta}^{(3)}(y) & \tau_{\gamma,\beta}^{(4)}(y) & \tau_{\gamma,\beta}^{(5)}(y) \\
     \tau_{\gamma,\beta}^{(4)}(y) &\tau_{\gamma,\beta}^{(6)}(y) & \tau_{\gamma,\beta}^{(7)}(y) \\
     \tau_{\gamma,\beta}^{(5)}(y) & \tau_{\gamma,\beta}^{(7)}(y) & \tau_{\gamma,\beta}^{(8)}(y)\\
     \MatR \MatL d_3(\sigma) \\ d_4(\sigma) \\d_5(\sigma)\MatR \le 1 .
 \end{align*}
 and $\tau_{\gamma,\beta}^{(9)}(y)d_6(\sigma)^2 \le 1$.

    Recall that as $\sigma \downarrow \BBP$, we have $y\downarrow \sqrt{\NoiseEdge_{\gamma,\beta}^+}$ and $\Delta_{\gamma,\beta}(y^2)\to \Delta_{\gamma,\beta}(\NoiseEdge_{\gamma,\beta}^+)=0$.
    Moreover, one can verify that $\tau_{\gamma,\beta}^{(1)}(y),\tau_{\gamma,\beta}^{(9)}(y)\to \infty$ as $\sigma\downarrow \BBP$, while $\tau_{\gamma,\beta}^{(2)}(y)$ is bounded. Consequently, $d_1(\sigma),d_2(\sigma),d_6(\sigma)\to 0$.   

    It remains to treat $d_3(\sigma),d_4(\sigma),d_5(\sigma)$. Let $\bd^*=(0,0,d_3^*,d_4^*,d_5^*,0)$ be a limit point of $\bd(\sigma)$ as $\sigma\downarrow\BBP$.
    For brevity, denote
    \begin{equation}
        \tilde{\bTc}(y) = \MatL \tau_{\gamma,\beta}^{(3)}(y) & \tau_{\gamma,\beta}^{(4)}(y) & \tau_{\gamma,\beta}^{(5)}(y) \\
        \tau_{\gamma,\beta}^{(4)}(y) &\tau_{\gamma,\beta}^{(6)}(y) & \tau_{\gamma,\beta}^{(7)}(y) \\
        \tau_{\gamma,\beta}^{(5)}(y) & \tau_{\gamma,\beta}^{(7)}(y) & \tau_{\gamma,\beta}^{(8)}(y)\\
        \MatR \,,\qquad 
        \eps(y) \equiv \sqrt{\frac{-\gamma \Delta_{\gamma,\beta}(y^2)}{\beta(1+\gamma-\beta\gamma)}} ,
    \end{equation}
    and $\bar{\bd}(\sigma)=(d_3(\sigma),d_4(\sigma),d_5(\sigma))$, $\bar{\bd}^*=(d_3^*,d_4^*,d_5^*)$ so that e.g. $\bd(\sigma)=(\bd_1(\sigma),\bd_2(\sigma))\oplus \bar{\bd}(\sigma) \oplus (\bd_6(\sigma))$. 

     One can calculate:
    \begin{equation}
    \begin{split}
        \lim_{y\downarrow\sqrt{\NoiseEdge_{\gamma,\beta}^+}} 
        \eps(y)
        \bTc_{\gamma,\beta}(y) 
        = 
        \MatL
        \star &\star &\star   \\
        1/\sqrt{\NoiseEdgeUpper} &1 &0 \\
        \frac{\gamma(1-\beta)}{1-\beta\gamma} (-1/\sqrt{\NoiseEdgeUpper}) &0 &\frac{\gamma(1-\beta)}{1-\beta\gamma} 
        \MatR 
        \equiv \tilde{\bTc}_{\infty} 
\end{split}
    \end{equation}
    where the first row, $\bm{t}_1$, is a linear combination of the second and third rows: $\bm{t}_1 = (1/\sqrt{\NoiseEdgeUpper})\bm{t}_2 + (-1/\sqrt{\NoiseEdgeUpper})\bm{t}_3$. 
    Multiplying the inequatlity $0\le \langle \bar{\bd}(\sigma),\tilde{\bTc}(y)\bar{\bd}(\sigma)\rangle \le 1$ by $\varepsilon(y)$, and taking the limit $\sigma\downarrow\BBP$ along a subsequence such that $\bd(\sigma)\to\bd^*$, we get $\langle \bar{\bd}^*,\tilde{\bTc}_{\infty}\bar{\bd}^*\rangle = 0 $. Consequently, since $\tilde{\bTc}_{\infty}$ is PSD, $\bar{\bd}^*$ is in the kernel of this matrix, that is 
    \begin{equation}
        \bm{d}^* = d_3^*\cdot \be  ,\qquad\textrm{where}\qquad \be\equiv (0, 0, 1, -1/\sqrt{\NoiseEdgeUpper},1/\sqrt{\NoiseEdgeUpper}, 0) .
    \end{equation}

    Finally, recall the additional linear constraint $(\bKc_{\gamma,\beta}(\sqrt{\NoiseEdgeUpper})-\BBP \bHc)\bd^* = \0$. We claim that the vector $\be$ is not in the kernel of this linear constraint, and therefore $d_3^*=0$. To see this, consider for example the second entry of $(\bKc_{\gamma,\beta}(\sqrt{\NoiseEdgeUpper})-\BBP \bHc)\be$. It is clear from the definition of $\bKc_{\gamma,\beta}$, \eqref{eq:K-func}, that $(\bKc_{\gamma,\beta} \be)_2=0$; in contrast, $(\bHc\be)_2 = -e_5=-1/\sqrt{\NoiseEdgeUpper} \ne 0$.

    To summarize, we have shown that for every limiting point $\bd^*$ of $\bd(\sigma)$ as $\sigma\downarrow\BBP$, we necessarily have $\bd^*=0$. This concludes the proof of the proposition.
    \qed

\end{document}

%% file: packages.tex
\usepackage[utf8]{inputenc} 
\usepackage[T1]{fontenc}
\usepackage{amsmath,amssymb,amsthm}
\usepackage{url}
\usepackage{ifthen}
\usepackage{graphicx}
\usepackage{xcolor}
\usepackage{calc}
\usepackage{bbm}
\usepackage{dsfont}
\usepackage{centernot}
\usepackage{bm}
\usepackage{enumitem}
\usepackage[colorlinks=true]{hyperref}
\usepackage[]{authblk}
\usepackage{mathtools}
\usepackage{setspace}
\usepackage{caption}
\usepackage{subcaption}
\usepackage{authblk}
\usepackage[linesnumbered,ruled,vlined]{algorithm2e}

\hypersetup{
	colorlinks,
	linkcolor={red!40!gray},
	citecolor={blue!40!gray},
	urlcolor={blue!70!gray}
}

\usepackage[margin=2.5cm]{geometry}

%% file: environments.tex
\newtheorem{theorem}{Theorem}
\newtheorem{conjecture}{Conjecture}
\newtheorem{proposition}{Proposition}
\newtheorem{corollary}{Corollary}
\newtheorem{lemma}{Lemma}

\theoremstyle{remark}
\newtheorem{remark}{Remark}

\newtheorem{fact}{Fact}

%% file: commands.tex
\newcommand{\Revision}[1]{{#1}}

\newcommand{\m}{\mathcal}

\newcommand{\RR}{\mathbb{R}}

\newcommand{\Var}{\mathrm{Var}}

\newcommand{\diag}{\mathop{\mathrm{diag}}}

\newcommand{\rank}{\mathop{\mathrm{rank}}}

\newcommand{\tr}{\mathop{\mathrm{tr}}}

\newcommand{\Expt}{\mathbb{E}}
\newcommand{\eps}{\varepsilon}

\newcommand{\T}{\top}
\newcommand{\BigOh}{\m{O}}

\newcommand{\bI}{\bm{I}}
\newcommand{\bY}{\bm{Y}}
\newcommand{\bX}{\bm{X}}
\newcommand{\bYtilde}{\tilde{\bY}}

\newcommand{\bZ}{\bm{Z}}

\newcommand{\bU}{\bm{U}}
\newcommand{\bV}{\bm{V}}
\newcommand{\bW}{\bm{W}}
\newcommand{\bM}{\bm{M}}
\newcommand{\bR}{\bm{R}}

\newcommand{\bP}{\bm{P}}
\newcommand{\bA}{\bm{A}}
\newcommand{\bB}{\bm{B}}
\newcommand{\bD}{\bm{D}}
\newcommand{\bT}{\bm{T}}
\newcommand{\be}{\bm{e}}
\newcommand{\bVbar}{\bar{\bm{V}}}
\newcommand{\bLambda}{\bm{\Lambda}}
\newcommand{\bSigma}{\bm{\Sigma}}
\newcommand{\bOmega}{\bm{\Omega}}
\newcommand{\bPhi}{\bm{\Phi}}
\newcommand{\bQc}{\bm{\mathcal{Q}}}
\newcommand{\bPc}{\bm{\mathcal{P}}}
\newcommand{\bBc}{\bm{\mathcal{B}}}
\newcommand{\bMc}{\bm{\mathcal{M}}}
\newcommand{\bKc}{\bm{\mathcal{K}}}

\newcommand{\bTc}{\bm{\mathcal{T}}}
\newcommand{\0}{\bm{0}}
\newcommand{\hbX}{\hat{\bX}}

\newcommand{\MatL}{\begin{bmatrix}}
\newcommand{\MatR}{\end{bmatrix}}
\newcommand{\bmu}{\bm{\mu}}

\newcommand{\MPStiel}{\mathsf{m}}

\newcommand{\MPDensity}{\mathsf{MP}}

\renewcommand{\Re}{\mathrm{Re}}
\newcommand{\SpikeFunc}{\m{Y}_{\gamma,\beta}}
\newcommand{\PosEigFunc}{\mathsf{L}_{\gamma,\beta}}
\newcommand{\BBP}{\sigma^{*}_{\gamma,\beta}}

\newcommand{\hbu}{\hat{\bm{u}}}
\newcommand{\hbv}{\hat{\bm{v}}}
\newcommand{\hbc}{\hat{\bm{c}}}
\newcommand{\hbd}{\hat{\bm{d}}}
\newcommand{\hbf}{\hat{\bm{f}}}
\newcommand{\bu}{\bm{u}}
\newcommand{\bv}{\bm{v}}
\newcommand{\bw}{\bm{w}}

\newcommand{\bd}{\bm{d}}
\newcommand{\hsigma}{\hat{\sigma}}
\newcommand{\hbY}{\hat{\bY}}
\newcommand{\hbZ}{\hat{\bZ}}
\newcommand{\bHc}{\bm{\m{H}}}
\newcommand{\hbM}{\widehat{\bM}}
\newcommand{\hbT}{\widehat{\bT}}

\newcommand{\bO}{\bm{O}}
\newcommand{\bQ}{\bm{Q}}

\newcommand{\NoiseEdge}{\mathsf{z}}

\newcommand{\NoiseEdgeUpper}{\NoiseEdge^+_{\gamma,\beta}}

\newcommand{\UAngle}{\m{U}_{\gamma,\beta}}
\newcommand{\VAngle}{\m{V}_{\gamma,\beta}}

\newcommand{\Shrinker}{\Upsilon_{\gamma,\beta}}

\newcommand{\barR}{\overline{\bR}}
\newcommand{\downD}{\underline{\bD}}
\newcommand{\barD}{\overline{\bD}}

\newcommand{\bnu}{\bm{\nu}}

\newcommand{\Hbb}{\bHc}

\newcommand{\Indic}[1]{ \ensuremath{\mathds{1}\left\{#1\right\}} }
\newcommand{\range}{\mathrm{range}}

\newcommand{\iu}{\mathrm{i}\mkern1mu}